\documentclass[11pt]{article}
\usepackage[T1]{fontenc}
 
\usepackage{graphicx}
\graphicspath{{Fig/}}

\usepackage{amssymb,amsmath,fullpage,hyperref,supertabular}

\usepackage{scalerel}
\usepackage{tikz}
\usetikzlibrary{svg.path}

\definecolor{orcidlogocol}{HTML}{A6CE39}
\tikzset{
  orcidlogo/.pic={
    \fill[orcidlogocol] svg{M256,128c0,70.7-57.3,128-128,128C57.3,256,0,198.7,0,128C0,57.3,57.3,0,128,0C198.7,0,256,57.3,256,128z};
    \fill[white] svg{M86.3,186.2H70.9V79.1h15.4v48.4V186.2z}
                 svg{M108.9,79.1h41.6c39.6,0,57,28.3,57,53.6c0,27.5-21.5,53.6-56.8,53.6h-41.8V79.1z M124.3,172.4h24.5c34.9,0,42.9-26.5,42.9-39.7c0-21.5-13.7-39.7-43.7-39.7h-23.7V172.4z}
                 svg{M88.7,56.8c0,5.5-4.5,10.1-10.1,10.1c-5.6,0-10.1-4.6-10.1-10.1c0-5.6,4.5-10.1,10.1-10.1C84.2,46.7,88.7,51.3,88.7,56.8z};
  }
}

\newcommand\orcidicon[1]{\href{https://orcid.org/#1}{\mbox{\scalerel*{
\begin{tikzpicture}[yscale=-1,transform shape]
\pic{orcidlogo};
\end{tikzpicture}
}{|}}}}

\newtheorem{Property}{Property} 
\newtheorem{Remark}{Remark}
\newtheorem{Example}{Example}
\newtheorem{Proposition}{Proposition}
\newenvironment{proof}{\paragraph{Proof:}}{\hfill$\square$}

\def\hatP{{\hat{P}}}
\def\arccosh{\mathrm{arccosh}}
\def\bbP{\mathbb{P}}
\def\FR{\mathrm{FR}}
\def\SSPD{\mathrm{SSPD}}
\def\RieSEB{\mathrm{RieSEB}}
\def\calS{\mathcal{S}}
\def\calD{\mathcal{D}}
\def\supp{\mathrm{supp}}
\def\ball{\mathrm{ball}}
\def\calM{\mathcal{M}}
\def\calG{\mathcal{G}}
\def\SH{\mathbb{SH}}
\def\Sym{\mathrm{Sym}}
\def\dZ{\mathrm{d}Z}
\def\dY{\mathrm{d}Y}
\def\dP{\mathrm{d}P}
\def\dZbar{\mathrm{d}\bar Z}
\def\barZ{{\bar{Z}}}
\def\dtheta{\mathrm{d}\theta}
\def\SPC{\mathrm{SPC}}
\def\trace{\mathrm{trace}}
\def\proj{\mathrm{proj}}
\def\barN{{\overline{\mathcal{N}}}}
\def\dP{\mathrm{d}P}
\def\CO{\mathrm{CO}}

\def\GL{\mathrm{GL}}
\def\KL{\mathrm{KL}}

\def\Killing{\mathrm{Killing}}
\def\diag{\mathrm{diag}}

\def\SPD{\mathrm{SPD}}
\def\Aff{\mathrm{Aff}}
\def\bbR{\mathbb{R}}
\def\Fisher{\mathrm{Fisher}}
\def\vectortwo#1#2{{\left[\begin{array}{l}#1 \cr #2\end{array}\right]}}
\def\mattwotwo#1#2#3#4{{\left[\begin{array}{ll}#1 & #2\cr #3 & #4\end{array}\right]}}
\def\tr{\mathrm{tr}}

\def\inner#1#2{\langle #1, #2\rangle}
\def\calN{\mathcal{N}}

\def\SPD{\mathrm{SPD}}
\def\calP{\mathcal{P}}
 
\def\Length{\mathrm{Length}}
\def\dt{\mathrm{d}t}

\def\SS{\mathrm{SS}}
\def\dbeta{\mathrm{d}\beta}

\def\dim{\mathrm{dim}}
\def\Aff{\mathrm{Aff}}
\def\Var{\mathrm{Var}}
\def\vech{\mathrm{vech}}
\def\std{\mathrm{std}}

\def\barP{{\bar P}}
\def\dSigma{\mathrm{d}\Sigma}
\def\dmu{\mathrm{d}\mu}
\def\ds{\mathrm{d}s}
\def\st{\ :\ }
\def\bbR{\mathbb{R}}

\def\Cov{\mathrm{Cov}}
\def\SL{\mathrm{SL}}
\def\SO{\mathrm{SO}}

\begin{document}

\title{A numerical approximation method for the Fisher-Rao distance between multivariate normal distributions}

\author{Frank Nielsen\protect\orcidicon{0000-0001-5728-0726}\\
Sony Computer Science Laboratories Inc, Tokyo, Japan. }
 
\date{}

\sloppy 
\maketitle              
\begin{abstract}
We present a simple method to approximate Rao's distance between multivariate normal distributions based on discretizing curves joining normal distributions and approximating Rao's distances between successive nearby normal distributions on the curves by the square root of Jeffreys divergence, the symmetrized Kullback-Leibler divergence. 
We consider experimentally the linear interpolation curves in the ordinary, natural and expectation parameterizations of the normal distributions, and compare these curves with a curve derived from the Calvo and Oller's isometric embedding  of the Fisher-Rao $d$-variate normal manifold into the cone of $(d+1)\times (d+1)$ symmetric positive-definite  matrices [Journal of multivariate analysis 35.2 (1990): 223-242].
We report on our experiments and assess the quality of our approximation technique by comparing the numerical approximations with both lower and upper bounds.
Finally, we present several information-geometric properties of the Calvo and Oller's isometric embedding. 
\end{abstract}


%
%
%
\section{Introduction}

\subsection{The Fisher-Rao normal manifold}

Let $\bbP(d)$ denote the set of symmetric positive-definite matrices, a convex regular cone, and let $\calN(d)=\{N(\mu,\Sigma) \st (\mu,\Sigma)\in\Lambda(d)=\bbR^d\times \bbP(d)\}$ denote the set of $d$-variate normal  distributions, MultiVariate Normals or  MVNs for short, also called Gaussian distributions.
A MVN distribution $N(\mu,\Sigma)$ has probability density function on the support $\bbR^d$:
$$
p_{\lambda=(\mu,\Sigma)}(x)= (2\pi)^{^\frac{d}{2}} |\Sigma|^{-\frac{1}{2}} \exp\left(-\frac{1}{2}(x-\mu)^\top\Sigma^{-1}(x-\mu)\right)
 ,\quad x\in\bbR^d.
$$

The statistical model $\calN(d)$ is of dimension $m=\dim(\Lambda(d))=d+\frac{d(d+1)}{2}=\frac{d(d+3)}{2}$ since it is identifiable, i.e.,
there is a one-to-one correspondence between $\lambda\in\Lambda(d)$ and $N(\mu,\Sigma)\in\calN(d)$ (i.e., $\lambda \leftrightarrow p_\lambda(x)$). 
The statistical model $\calN(d)$ is said regular since the second order derivatives $\frac{\partial^2 p_{\lambda}}{\partial\lambda_i\partial\lambda_j}$ and third order derivatives $\frac{\partial^3 p_{\lambda}}{\partial\lambda_i\partial\lambda_j\partial\lambda_k}$ are smooth functions (defining the metric and cubic tensors in information geometry~\cite{IG-2016}), and the set of first order partial derivatives $\{\frac{\partial p_{\lambda}}{\partial\lambda_1},\ldots,\frac{\partial p_{\lambda}}{\partial\lambda_1}\}$ are linearly independent.

Let $\Cov(X)$ denote the covariance of $X$ or variance when $X$ is scalar. 
The Fisher information matrix (FIM) is the symmetric semi-positive definite matrix:
$$
I(\lambda)=\Cov[\nabla\log p_{\lambda}(x)]\succeq 0.
$$
For regular statistical models $\{p_\lambda\}$, the FIM  is positive-definite: $I(\lambda)\gg 0$. 

\begin{Remark}
An example of non-regular statistical model is the set  $\calD=\{\delta_\theta(x)\st\theta\in\bbR\}$ of Dirac distributions where 
$\delta_\theta(x)=\left\{\begin{array}{ll} 1, & x=\theta,\cr 0, & x\not=\theta\end{array}\right.$.
Indeed, we have the following FIM of a Dirac distribution $\delta_\theta$: $I(\theta)=\Var[\delta_\theta]=E[X^2]-E[X]^2=0$.
The family of Dirac distributions is an example of non-regular model where the support of the distributions 
depends on the model parameter (i.e., $\supp(\delta_\theta)=\theta$).
Mathematical statistics of non-regular models are considered in~\cite{akahira2012non}.  
\end{Remark}

The FIM is covariant under reparameterization of the statistical model.
That is, $\theta(\lambda)$ be a new parameterization of the MVNs. Then we have
$$
I_\theta(\lambda)=\left(\frac{\partial\lambda}{\partial\theta}\right)^\top\, I_\lambda(\lambda(\theta)) \, \left(\frac{\partial\lambda}{\partial\theta}\right).
$$
For example, we may parameterize univariate normal distributions by $\lambda=(\mu,\sigma^2)$ or $\theta=(\mu,\sigma)$.
We obtain the following Fisher information matrices for these parameterizations:
$$
I_\lambda(\lambda(\mu,\sigma))=\mattwotwo{\frac{1}{\sigma^2}}{0}{0}{\frac{1}{2\sigma^4}}, \quad
I_\theta(\theta(\mu,\sigma))=\mattwotwo{\frac{1}{\sigma^2}}{0}{0}{\frac{1}{2\sigma^2}}.
$$
In higher dimensions, parameterization $\lambda$ corresponds to the parameterization $(\mu,\Sigma)$ 
while parameterization $\theta=(\mu,L)$ where $\Sigma=LL^\top$ is the unique Cholesky decomposition with $L\in\GL(d)$, the group of invertible $d\times d$ matrices. Another useful parameterization for optimization is the 
log-Cholesky parameterization~\cite{lin2019riemannian} ($\eta=(\mu,\log\sigma^2)\in\bbR^2$ for univariate normals) which ensures that gradient descent stay in the domain. The Fisher information matrix with respect to the log-Cholesky parameterization is 
$I_\eta(\eta(\mu,\sigma))=\mattwotwo{\frac{1}{\sigma^2}}{0}{0}{2}$ with $\eta(\mu,\sigma)\in\bbR^2$.

Since the statistical model $\calN(d)$ is identifiable and regular, the Fisher information matrix can be written equivalently using the first two Bartlett identities as
\begin{eqnarray}
I(\mu,\Sigma)=\Cov[\nabla\log p_{(\mu,\Sigma)}] &=& E\left[\nabla\log p_{(\mu,\Sigma)}\nabla\log p_{(\mu,\Sigma)}^\top\right],\\
&=&-E\left[\nabla^2\log p_{(\mu,\Sigma)}\right]. \label{eq:FIM2}
\end{eqnarray}

For multivariate distributions parameterized by a $m$-dimensional vector 
$$
\theta=(\theta_1,\ldots,\theta_d,\theta_{d+1},\ldots,\theta_{m})\in\bbR^m,
$$
with $\mu=(\theta_1,\ldots,\theta_d)$ and $\Sigma(\theta)=\vech(\theta_{d+1},\ldots,\theta_{m})$ (inverse half-vectorization of matrices), we have~\cite{Skovgaard-1984,malago2015information,herntier2022transversality}:
$$
I(\theta)=[I_{ij}(\theta)],\quad 
I_{ij}(\theta)=
\left(\frac{\partial\mu}{\partial\theta_i}\right)^\top\Sigma^{-1}\frac{\partial\mu}{\partial\theta_j}
+\frac{1}{2}\tr\left(\Sigma^{-1}\frac{\partial\mu}{\partial\theta_i}\Sigma^{-1}\frac{\partial\mu}{\partial\theta_j} \right).
$$

By equipping the regular statistical model $\calN(d)$ with the Fisher information metric  
$$
g^\Fisher_\calN(\mu,\Sigma)=\Cov[\nabla\log p_{(\mu,\Sigma)}(x)]
$$
we get a Riemannian manifold $\calM=\calM_{\calN}$ called the Fisher-Rao Gaussian~\cite{Skovgaard-1984}.
The induced Riemannian geodesic distance $\rho_\calN(\cdot,\cdot)$   is called the Rao distance~\cite{AtkinsonRao-1981} or the Fisher-Rao distance~\cite{Rao-1945,nielsen2013cramer,chen2021upper}:
$$
\rho_\calN(N(\lambda_1),N(\lambda_2))=\inf_c \left\{\Length(c) \st c(0)=p_{\lambda_1}, c(1)=p_{\lambda_2} \right\},
$$ 
where the Riemannian length of a smooth curve $c(t)$ is defined by
$$
\Length(c)=\int_0^1 \underbrace{\sqrt{\inner{\dot c(t)}{\dot c(t)}}_{c(t)}}_{\ds_\calN(t)} \dt.
$$
The minimizing curve $\gamma_\calN(p_{\lambda_1},p_{\lambda_2};t)$ is called the Fisher geodesic:
It is also an auto-parallel curve for the Levi-Civita connection $\nabla^\Fisher$ induced by the Fisher metric $g^\Fisher$.
See~\cite{IG-2016} for details.

\begin{Remark}
If we consider the Riemannian manifold $(\calM,\beta g)$ for $\beta>0$ then the length element is scaled by $\sqrt{\beta}$: $\ds_{\beta g}=\beta \ds_g$. It follows that the length of a curve $c$ is $\Length_{\beta g}(c)=\sqrt{\beta}\Length_{g}(c)$.
However, the geodesics are the same: $\gamma_{\beta g}(p_1,p_2;t)=\gamma_{g}(p_1,p_2;t)$ with $\gamma_{g}(p_1,p_2;0)=p_1$ and 
$\gamma_{g}(p_1,p_2;1)=p_2$.
\end{Remark}

Historically, Hotelling~\cite{Hotelling-1930} first used this Fisher Riemannian geodesic distance in the late 1920's.
From the viewpoint of information geometry~\cite{IG-2016}, the Fisher metric is the unique Markov invariant metric up to rescaling~\cite{cencov2000statistical,bauer2016uniqueness,fujiwara2022hommage} and the Fisher-Rao distance has been used to design statistical hypothesis testing~\cite{burbea1989rao,SDPMVN-1990,rios1992rao,park1994distances}, to measure the distance between the prior and posterior distributions in Bayesian statistics~\cite{gruber2008some}, clustering~\cite{strapasson2016clustering,le2019quantization},  in signal processing~\cite{said2015texture,legrand2016evaluating,halder2018gradient,collas2022riemannian}, and in deep learning~\cite{liang2019fisher} among others.

The squared line element  
$$
\ds^2_\calN(\mu,\Sigma)=g_{(\mu,\Sigma)}((\dmu,\dSigma),(\dmu,\dSigma))=\vectortwo{\dmu}{\dSigma}^\top\, I(\mu,\Sigma)\, \vectortwo{\dmu}{\dSigma}
$$ 
induced by the Fisher information metric  of the normal family is 
\begin{equation}
\ds^2_\calN(\mu,\Sigma) = \dmu^\top \Sigma^{-1} \dmu + \frac{1}{2}\tr\left(\left(\Sigma^{-1}\dSigma\right)^2\right).
\end{equation}

\begin{Remark}
The family $\calN(d)$ of normal distributions forms an exponential family~\cite{IG-MVN-1999}: 
$$
\calN(d)=\left\{p_{\theta(\lambda)}=\exp\left(\inner{\theta_v(\mu)}{x}+\inner{\theta_M(\Sigma)}{xx^\top}-F_\calN(\theta_v,\theta_M) \right)\right\},
$$
with $\theta(\lambda)=(\theta_v=(\Sigma^{-1}\mu,\theta_M=\frac{1}{2}\Sigma^{-1})$ the natural parameters and log-partition/cumulant function
$$
F_\calN(\theta) = \frac{1}{2}\left(d\log\pi-\log|\theta_M|+\frac{1}{2}\theta_v^\top\theta_M^{-1}\theta_v\right).
$$
The matrix inner product is $\inner{M_1}{M_2}=\tr(M_1M_2^\top)$.
Using Eq.~\ref{eq:FIM2}, it follows that the MVN FIM is 
$I_\theta(\theta)=-E[\nabla^2\log p_{\theta}]=\nabla^2 F(\theta)$.
As an exponential family~\cite{IG-2016}, we also have $I_\theta(\theta)=E[t(x)]$, where $t(x)=(x,xx^\top)$ is the sufficient statistic.
Thus the Fisher metric is a Hessian metric~\cite{shima2007geometry}.
Let $F_\calN(\theta_v,\theta_M)=F_v(\theta_v)+F_M(\theta_M)$ with 
$F_v(\theta_v)=\frac{1}{2}\left(d\log\pi+\frac{1}{2}\theta_v^\top\theta_M^{-1}\theta_v\right)$
and
$F_M(\theta_M)=-\frac{1}{2} \log|\theta_M|$.
We have
$$
I(\theta(\lambda))=\nabla^2 F_\calN(\theta(\mu,\Sigma))=\mattwotwo{\Sigma^{-1}}{0}{0}{\frac{1}{2}\nabla^2_{\theta_M}\log|\frac{1}{2}\Sigma^{-1}|}.
$$
Therefore $\ds^2_\calN(\mu,\Sigma) = \ds^2_{v} + \ds^2_M$ with
$\ds^2_v(\mu)=\dmu^\top \Sigma^{-1} \dmu$ and 
$\ds^2_M(\Sigma)=\frac{1}{2}\tr\left(\left(\Sigma^{-1}\dSigma\right)^2\right)$.
Let us note in passing that $\nabla^2_{\theta_M}\log|\theta_M|$ is a fourth order tensor~\cite{soen2021variance}.
\end{Remark}

The family $\calN(d)$ can also be considered as an elliptical family~\cite{SDPElliptical-2002}, thus highlighting the affine-invariance property of the Fisher information metric.
That is, the Fisher metric is invariant with respect to affine transformations~\cite{burbea1984informative}:
Let $(a,A)$ be an element of the affine group $\Aff(d)$ with $a\in\bbR^d$ and $A\in\GL(d)$.
The group identity element is $e=(0,I)$ and the group operations are $(a_1,A_1).(a_2,A_2)=(a_1+A_1a_2,A_1A_2)$ and $(a,A)^{-1}=(-A^{-1}a,A^{-1})$).
Then we have

\begin{Property}[Fisher-Rao affine-invariance]\label{prop:aifr}
For all $A\in\GL(d), a\in\bbR^d$, we have $\rho_{\calN}(N(A\mu_1+a,A\Sigma_1A^\top),N(A\mu_2+a,A\Sigma_2A^\top))=\rho_{\calN}(N(\mu_1,\Sigma_1),N(\mu_2,\Sigma_2))$.
\end{Property}

It follows that we have
\begin{eqnarray*}
\rho_{\calN}(N(\mu_1,\Sigma_1),N(\mu_2,\Sigma_2)) &=&
\rho_{\calN}(N_\std,N(\Sigma_1^{-\frac{1}{2}}(\mu_2-\mu_1),\Sigma_1^{-\frac{1}{2}}\Sigma_2\Sigma_1^{-\frac{1}{2}})),\\
&=&
\rho_{\calN}(N(\Sigma_2^{-\frac{1}{2}}(\mu_1-\mu_2),\Sigma_2^{-\frac{1}{2}}\Sigma_1\Sigma_2^{-\frac{1}{2}}),N_\std),
\end{eqnarray*}
where $N_\std=N(0,I)$ is the standard $d$-variate distribution.
The family of normal distributions can be obtained from the standard normal distribution by the action of the affine group $\Aff(d)$:
$$
N(\mu,\Sigma)=(\mu,\Sigma^{\frac{1}{2}}).N_\std=N((\mu,\Sigma^{\frac{1}{2}}).(0,I)).
$$

\subsection{Fisher-Rao distance between normal distributions}
In general, the Fisher-Rao distance $\rho_\calN(N_1,N_2)$ between two multivariate normal distributions $N_1$ and $N_2$ is not known in closed-form~\cite{Eriksen-1986,FRGeodesicEll-1997,imai2011remarks,inoue2015group}, and several lower and upper bounds~\cite{strapasson2015bounds}, and numerical techniques like costly and numerically instable geodesic shooting~\cite{MVNGeodesicShooting-2014,GeodesicShooting-2016,barbaresco2019souriau} have been investigated. See~\cite{FRMVNReview-2020} for a recent review.

Two difficulties to calculate the Fisher-Rao distance are 
\begin{itemize}
\item  to know explicitly the expression of the Riemannian Fisher-Rao geodesic $\gamma_\calN^{\FR}$ and
\item  to integrate in closed-form the length element $\ds_\calN$ along this Riemannian geodesic.
\end{itemize}

Note that the Fisher-Rao geodesics~\cite{IG-2016} $\gamma_\calN^{\mathrm{FR}}(t)$ are parameterized by constant speed (i.e., $\dot{\mu}(t)=\dot{\mu}(0)$ and $\dot\Sigma(t)=\dot\Sigma(0)$, proportional to arc length parameterization).
However, in several special cases, the Fisher-Rao distance between normal distributions belonging to restricted subsets of $\calN$ is known.

Three such prominent cases are 
\begin{itemize}
\item  when the normal distributions are univariate ($d=1$),  
\item  when we consider the set
$\calN_\mu=\{N(\mu,\Sigma)\st \Sigma\in\calP(d)\}\subset\calM_{\calN}$ of normal distributions share the same mean $\mu$ (with the embedded submanifold $\calS_\mu\in\calM$), and
\item  when  we consider the set
$\calN_\Sigma=\{N(\mu,\Sigma)\st \Sigma\in\calP(d)\}\subset\calN$ of normal distributions share the same covariance matrix $\Sigma$ 
(with the corresponding embedded submanifold $\calS_\Sigma\in\calM$)
\end{itemize}

Let us report the formula of the Fisher-Rao distance in these three cases:

\begin{itemize}
\item
In the univariate case $\calN(1)$, the Fisher-Rao distance between $N_1=N(\mu_1,\sigma_1^2)$ and $N_2=N(\mu_2,\sigma_2^2)$ can be derived from the hyperbolic distance expressed in the Poincar\'e upper space:
\begin{equation}\label{eq:FR1D}
\rho_{\calN}(N(\mu_1,\sigma_1^2),N(\mu_2,\sigma_2^2))=\sqrt{2}\log \left(\frac{1+\Delta(\mu_1,\sigma_1;\mu_2,\sigma_2)}{1-\Delta(\mu_1,\sigma_1;\mu_2,\sigma_2)}\right)
,
\end{equation}
with  
\begin{equation}
\Delta(a,b;c,d)=\sqrt{\frac{(c-a)^2+2(d-b)^2}{(c-a)^2+2(d+b)^2}}.
\end{equation}
Figure~\ref{fig:normal1d} displays four univariate normal distributions with their pairwise geodesics and Fisher-Rao distances.

\begin{figure}
\centering
\includegraphics[width=0.5\textwidth]{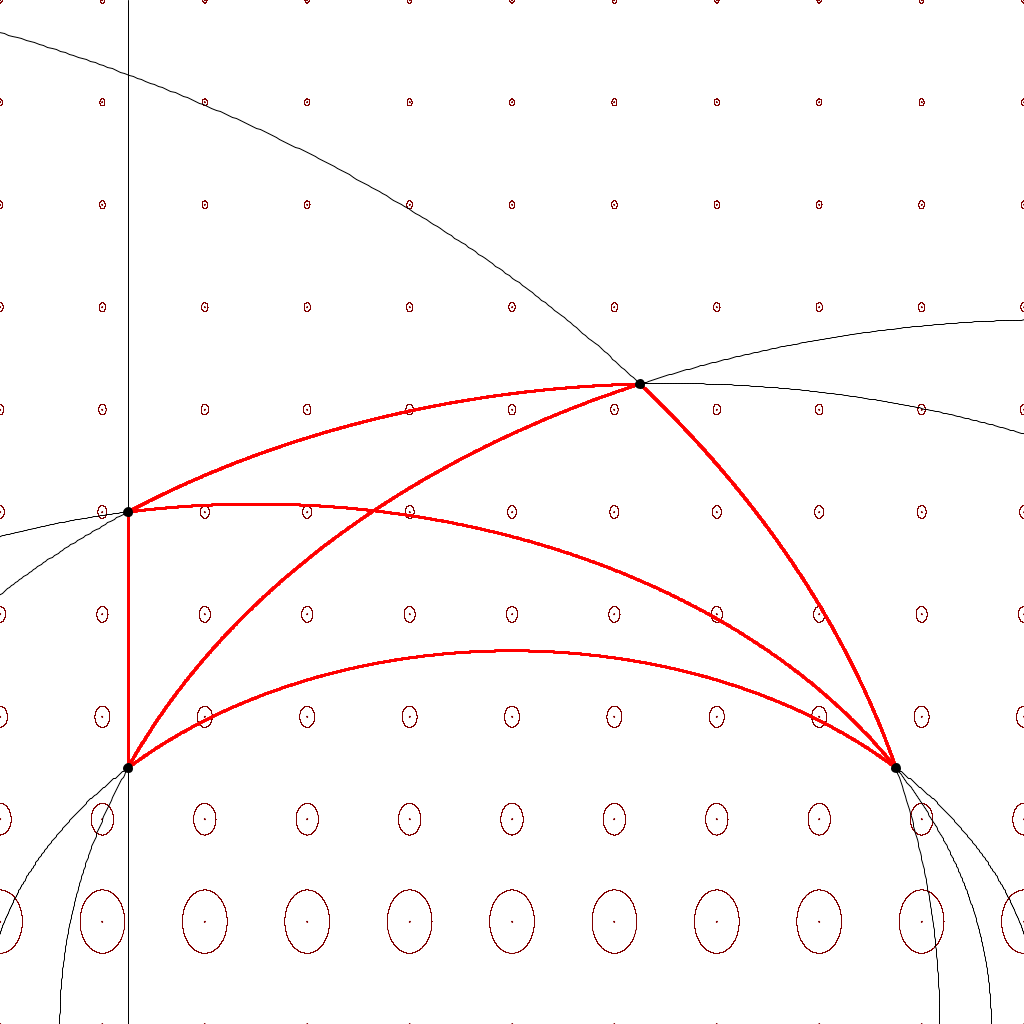}
\caption{Four univariate normal distributions 
$N_1=N(0,1)$, $N_2=N(3,1)$, $N_3=N(2,2.5)$, $N_4=N(0,2)$, and their pairwise full geodesics and geodesics linking them.
The Fisher-Rao distances are $\rho_\calN(N_1,N_2)=2.6124...$, $\rho_\calN(N_3,N_4)=0.9317...$,
$\rho_\calN(N_1,N_4)=0.9803...$, $\rho_\calN(N_2,N_3)=1.4225...$
$\rho_\calN(N_2,N_4)=2.1362...$, and $\rho_\calN(N_1,N_3)=1.7334...$.
 \label{fig:normal1d}}
\end{figure}


\item In the second case, the Rao distance between $N_1=N(\mu,\Sigma_1)$ and $N_2=N(\mu,\Sigma_2)$ has been reported in~\cite{siegel2014symplectic,DoubleCone-James-1973,Skovgaard-1984,wells2020fisher}:
\begin{eqnarray}
\rho_{\calN_\mu}(N_1,N_2)&=& \sqrt{\frac{1}{2} \sum_{i=1}^d \log^2 \lambda_i(\Sigma_1^{-1}\Sigma_2)},\\ \label{eq:RaoCMVN}
&=&\rho_{\calP}(\Sigma_1,\Sigma_2),
\end{eqnarray}
where $\lambda_i(M)$ denotes the $i$-th largest eigenvalue of matrix $M$.
Let us notice that  $\rho_{\calN_\mu}((\mu,\Sigma_1),(\mu,\Sigma_2))=\rho_{\calN_\mu}((\mu,\Sigma_1^{-1}),(\mu,\Sigma_2^{-1}))$
 since $\lambda_i(\Sigma_2^{-1}\Sigma_1)=\frac{1}{\lambda_i(\Sigma_1^{-1}\Sigma_2)}$ and 
 $\log^2 \lambda_i(\Sigma_2^{-1}\Sigma_1)=(-\log \lambda_i(\Sigma_1^{-1}\Sigma_2))^2=\log^2 \lambda_i(\Sigma_1^{-1}\Sigma_2)$.
Also, matrix $\Sigma_1^{-1}\Sigma_2$ may not be SPD: We may consider the SPD matrix $\Sigma_1^{-\frac{1}{2}}\Sigma_2\Sigma_1^{-\frac{1}{2}}$ which is SPD and such that $\lambda_i(\Sigma_1^{-1}\Sigma_2)=\lambda_i(\Sigma_1^{-\frac{1}{2}}\Sigma_2\Sigma_1^{-\frac{1}{2}})$.
The Rao distance of Eq.~\ref{eq:RaoCMVN} can be equivalently written~\cite{calvo1991explicit} as
$$
\rho_{\calN_\mu}(N_1,N_2)=\frac{1}{\sqrt{2}}\, \|\log\left(\Sigma_1^{-\frac{1}{2}}\Sigma_2\Sigma_1^{-\frac{1}{2}}\right)\|.
$$
This metric distance was rediscovered and analyzed in~\cite{forstner2003metric}.
Let $\rho_\SPD(P_1,P_2)=\sqrt{  \sum_{i=1}^d \log^2 \lambda_i(P_1^{-1}P_2)}$ so that $\rho_{\calN_\mu}(N(\mu,P_1),N(\mu,P_2))=\frac{1}{\sqrt{2}}\, \rho_\SPD(P_1,P_2)$.

This Riemannian SPD distance $\rho_\SPD$ enjoys the following  invariance properties:
\begin{itemize} 
\item Invariance by congruence transformation:
\begin{equation}\label{eq:spdconginvar}
\forall X\in\GL(d), \rho_\SPD(XP_1X^\top,XP_2X^\top)=\rho_\SPD(P_1,P_2),
\end{equation}
\item Invariance by inversion:
$$
\forall P_1,P_2\in\bbP(d), \rho(P_1^{-1},P_2^{-1})=\rho_\SPD(P_1,P_2).
$$
Let $P_1=L_1L_1^\top$ be the unique Cholesky decomposition. 
Then apply the congruence invariance for $X=L_1^{-1}$:
\begin{equation}
\rho_\SPD(P_1,P_2)=\rho_\SPD(L_1^{-1}P_1(L_1^{-1})^\top,L_1^{-1}P_2(L_1^{-1})^\top)=\rho_\SPD(I,L_1^{-1}P_2(L_1^{-1})^\top).
\end{equation}
We can also consider the factorization $P_1=S_1S_1$ where $S_1=P_1^{\frac{1}{2}}$ is the unique symmetric square root matrix.
Then we have
$$
\rho_\SPD(P_1,P_2)=\rho_\SPD(S_1^{-1}P_1(S_1^{-1})^\top,S_1^{-1}P_2(S_1^{-1})^\top)=\rho_\SPD(I,S_1^{-1}P_2(S_1^{-1})^\top).
$$

\begin{Remark}
In practice the covariance matrices usually need to be estimated from samples (sample covariance matrices) before measuring the distances between them.
In large dimensions, these approximations suffer severe errors and better consistent estimates of statistical distances based on random matrix theory (RMT) have been proposed in~\cite{FisherRaoSPD-2019}.
\end{Remark}

\item The Rao distance between $N_1=N(\mu_1,\Sigma)$ and $N_2=N(\mu_2,\Sigma)$ has been reported in closed-form~\cite{FRMVNReview-2020} (Proposition~3). Their method is detailed in the Appendix~\ref{sec:BFRsamecovar}.
We  consider the following simple  scheme based on the inverse $\Sigma^{-\frac{1}{2}}$ of the symmetric square root factorization of $\Sigma=\Sigma^{\frac{1}{2}}\Sigma^{\frac{1}{2}}$ (and $(\Sigma^{-\frac{1}{2}})^\top=\Sigma^{-\frac{1}{2}}$). Let us use the affine invariance property of the Fisher-Rao distance under the transformation $\Sigma^{-\frac{1}{2}}$ and then apply affine invariance under translation:
\begin{eqnarray*}
\rho_\calN(N(\mu_1,\Sigma),N(\mu_2,\Sigma)) &=&
 \rho_\calN(N(\Sigma^{-\frac{1}{2}}\mu_1,\Sigma^{-\frac{1}{2}}\Sigma\Sigma^{-\frac{1}{2}}),N(\Sigma^{-\frac{1}{2}}\mu_2,\Sigma^{-\frac{1}{2}}\Sigma\Sigma^{-\frac{1}{2}})),\\
&=& \rho_\calN(N(0,I),N(\Sigma^{-\frac{1}{2}}(\mu_2-\mu_1),I)),\\
&=& \rho_\calN(N(0,1),N(\|\Sigma^{-\frac{1}{2}}(\mu_2-\mu_1)\|_2,1)).
\end{eqnarray*}
The right-hand side Fisher-Rao distance is computed from Eq.~\ref{eq:FR1D} and justified by the method ~\cite{FRMVNReview-2020} (Proposition~3) described in the Appendix~\ref{sec:BFRsamecovar}.
Section~\ref{sec:FRsamecovar} shall report a simpler closed-form formula by proving that the Fisher-Rao distance between $N(\mu_1,\Sigma)$ and $N(\mu_2,\Sigma)$ is a scalar function of their Mahalanobis distance~\cite{mahalanobis1936generalised} using an algebraic method.
\end{itemize}

\subsection{Fisher-Rao distance:  Totally vs non-totally geodesic submanifolds}\label{sec:submfd}
Consider $\calN'=\{N(\lambda) \st\lambda'\in\Lambda'\}\subset\calN$ a statistical submodel of the MVN statistical model $\calN$.
Using the Fisher information matrix $I_{\lambda'}(\lambda')$, we get the intrinsic Fisher-Rao manifold $\calM'=\calM_{\calN'}$.
We may also consider $\calM'$ to be an embedded submanifold of $\calM$. 
 Let us write $\calS'=\calS_{\calN'}\subset\calM$ the embedded submanifold.

A totally geodesic submanifold $\calS'\subset\calM$ is such that the geodesics $\gamma_{\calM'}(N_1',N_2';t)$ fully stay in $\calM'$ for any pair of points $N_1',N_2'\in\calN'$.
For example, the submanifold $\calM_\mu=\{N(\mu,\Sigma)\st \Sigma\in\bbP(d)\}\subset\calM$ of MVNs with fixed mean $\mu$ is a totally geodesic submanifold~\cite{godinho2012introduction} of $\calM$ but the submanifold $\calM_\Sigma=\{N(\mu,\Sigma)\st \mu\in\bbR^d\}\subset\calM$ of MVNs sharing the same  
covariance matrix $\Sigma$ is not totally geodesic. 
When an embedded submanifold $\calS\subset\calM$ is totally geodesic, we always 
have $\rho_{\calM}(N_1,N_2)=\rho_{\calS}(N_1,N_2)$.
Thus we have $\rho_\calN(N(\mu,\Sigma_1),N(\mu,\Sigma_2))=\rho_\SPD(\Sigma_1,\Sigma_2)$.
However, when an embedded submanifold $\calS\subset\calM$ is not totally geodesic, we have $\rho_{\calM}(N_1,N_2)\leq \rho_{\calS}(N_1,N_2)$ because the Riemannian geodesic length in $\calS$ is necessarily longer or equal than the Riemannian geodesic length in $\calM$.
The merit to consider submanifolds is to be able to calculate in closed form the Fisher-Rao distance which may then provide an upper bound on the Fisher-Rao distance for the full statistical model.
For example, consider $N_1=N(\mu_1,\Sigma)$ and $N_2=N(\mu_2,\Sigma)$ in $\calM_\Sigma$, a non-totally geodesic submanifold.
The Rao distance between $N_1$ and $N_2$ in $\calM$ is upper bounded by the Riemannian distance in $\calM_\Sigma$ 
(with line element $\ds_\Sigma^2=\dmu^\top\Sigma^{-1}\dmu$) which corresponds to the Mahalanobis distance~\cite{mahalanobis1936generalised,AtkinsonRao-1981} $\Delta_\Sigma(\mu_1,\mu_2)$:
\begin{equation}
\rho_{\calM_\mu}(N_1,N_2)\leq  \Delta_\Sigma(\mu_1,\mu_2):=\sqrt{(\mu_2-\mu_1)^\top \Sigma^{-1} (\mu_2-\mu_1)}.
\end{equation}
\end{itemize}
The Mahalanobis distance can be interpreted as the Euclidean distance $D_E(p,q)=\Delta_I(p,q)=\sqrt{(p-q)^\top (p-q)}$ (where $I$ denotes the identity matrix) after an affine transformation: Let $\Sigma=LL^\top=U^\top U$ be the Cholesky decomposition of $\Sigma\gg 0$ with $L$ a lower triangular matrix or $U=L^\top$ an upper triangular matrix. Then we have 
\begin{eqnarray*}
\Delta_\Sigma(\mu_1,\mu_2)&=&\sqrt{(\mu_2-\mu_1)^\top (L^\top)^{-1} L^{-1} (\mu_2-\mu_1)},\\
&=&\|\Sigma^{-\frac{1}{2}}(\mu_2-\mu_1)\|_2,\\
&=&\Delta_I(L^{-1}\mu_1,L^{-1}\mu_2)=D_E(L^{-1}\mu_1,L^{-1}\mu_2),
\end{eqnarray*}
where $\|\cdot\|_2$ denotes the vector $\ell_2$-norm.

The Rao distance $\rho_\Sigma$ of Eq.~\ref{eq:FRSigma} between two MVNs with fixed covariance matrix emanates from the property that the submanifold 
$\calM_{[v],\Sigma}=\{N(a v,\Sigma) \st a\in\bbR\}$ is totally geodesic~\cite{strapasson2016totally}.

Let us emphasize that for a submanifold $\calS\subset\calM$ to be totally geodesic or not depend on the underlying metric in $\calM$. 
A same subset $\calN'\subset\calN$ with $\calN$ equipped with two different metric $g_1$ and $g_2$ can be totally geodesic wrt. $g_1$ and non-totally geodesic wrt. $g_2$. See Remark~\ref{rk:secondco} for such an example.

In general, using   the triangle inequality of the Riemannian metric distance $\rho_\calN$, we can upper bound $\rho_\calN(N_1,N_2)$ with $N_1=(\mu_1,\Sigma_1)$ and $N_1=(\mu_2,\Sigma_2)$  as follows:
\begin{eqnarray*}
\rho_\calN(N_1,N_2)&\leq& \rho_{\calM_{\mu_1}}(N_1,N_{12})+\rho_{\calM_{\Sigma_2}}(N_{12},N_2),\\
&\leq& \rho_{\calM_{\Sigma_1}}(N_1,N_{21})+\rho_{\calM_{\mu_2}}(N_{21},N_2),
\end{eqnarray*}
where $N_{12}=(\mu_1,\Sigma_2)$ and $N_{21}=N(\mu_2,\Sigma_1)$.
See Figure~\ref{fig:totallynottotally} for an illustration.
Furthermore, since $\rho_{\calN_{\Sigma_1}}(N_1,N_{21})\leq\Delta_{\Sigma_1}(\mu_1,\mu_2)$ and 
$\rho_{\calN_{\Sigma_2}}(N_{12},N_2)\leq \Delta_{\Sigma_2}(\mu_1,\mu_2)$, we get the following upper bound on the Rao distance between MVNs:
\begin{equation}\label{eq:UBMah}
\rho_{\mathcal{N}}(N_1,N_2)\leq \rho_{\mathcal{P}}(\Sigma_1,\Sigma_2)+\min\{\Delta_{\Sigma_1}(\mu_1,\mu_2),\Delta_{\Sigma_2}(\mu_1,\mu_2)\}.
\end{equation}
See also~\cite{chen2022multisensor}. 

\begin{figure}
\centering
\includegraphics[width=\textwidth]{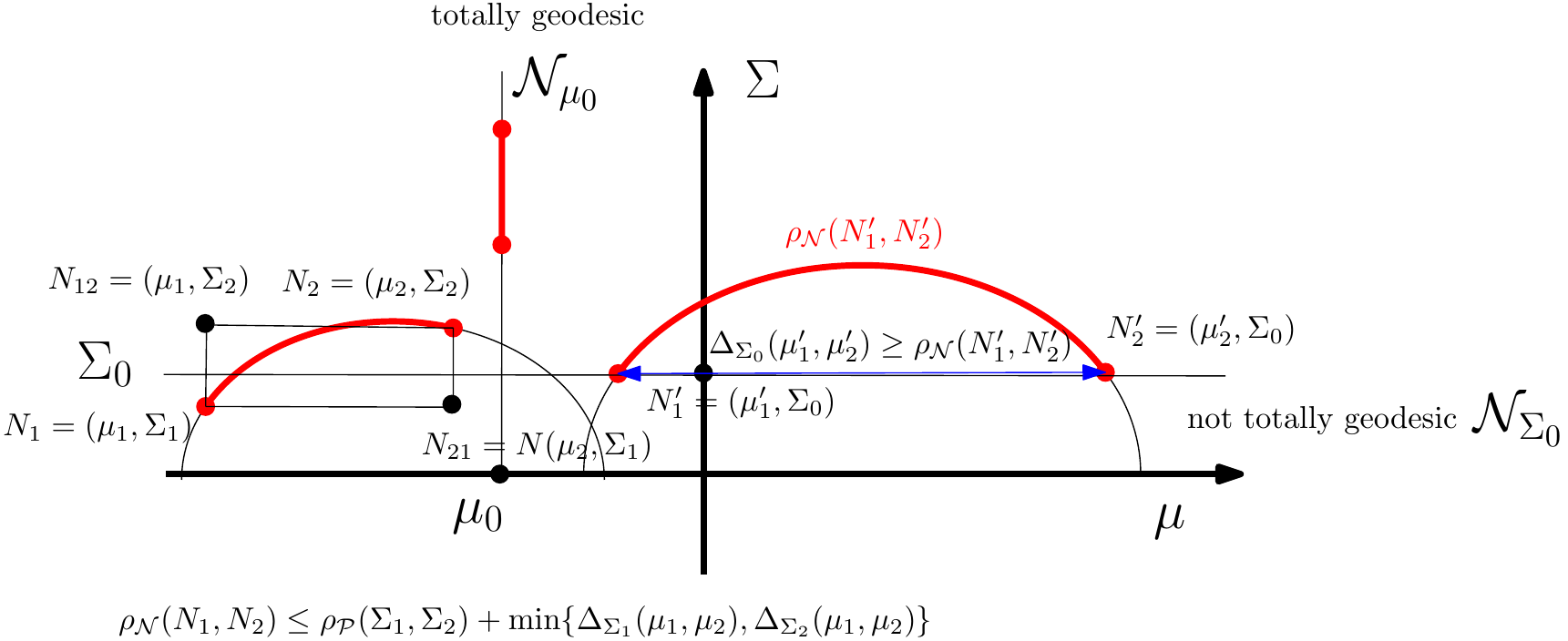}
\caption{
The submanifolds $\calN_{\Sigma}$ are not totally geodesic (i.e., $\rho_\calN(N_1',N_2')$ is upper bounded by their Mahalanobis distance) but the submanifolds $\calN_\mu$ are totally geodesic.
Using the triangle inequality of the Riemannian metric distance $\rho_\calN$, we can upper bound $\rho_\calN(N_1,N_2)$.
 \label{fig:totallynottotally}}
\end{figure}

In general, the difficulty of calculating the Fisher-Rao distance comes from the fact that 
\begin{enumerate}
\item  we do not know the Fisher-Rao geodesics with boundary value conditions (BVP) in closed form (only the geodesics with initial value conditions~\cite{calvo1991explicit}, IVP, are partially known),
\item we have to integrate the line element $\ds_\calN$ along the geodesic.
\end{enumerate}

As we shall see in \S\ref{sec:approxlength}, the above first problem is much hard to solve then the second problem which can be easily approximated by discretizing the curve. 
The lack of a closed-form formula and fast and good approximations for $\rho_{\calN}$ between MVNs is a current limiting factor for  its use in applications. 
Indeed, many applications (e.g., \cite{Sra-2016,nguyen2021geomnet}) consider the restricted case of the Rao distance between zero-centered MVNs which have closed-form (distance of Eq.~\ref{eq:RaoCMVN} in the SPD cone). 
The SPD cone  is a symmetric Hadamard manifold and its isometries have been fully studied and classified in~\cite{SPDisometries-2018} (\S4).
The Fisher-Rao geometry of zero-centered generalized MVNs was recently studied in~\cite{verdoolaege2012geometry}.

\subsection{Contributions and paper outline}
The main contribution of this paper is to propose an approximation of $\rho_{\calN}$ based on Calvo \& Oller's embedding~\cite{SDPMVN-1990}  (C\&O for short) and report its experimental performance.
First, we concisely recall C\&O's family of embeddings $f_\beta$ of $\calN(d)$ as submanifolds $\barN_\beta$ of $\calP(d+1)$ in Section~\ref{sec:CO}.
Next, we present our approximation technique in Section~\ref{sec:approxFR}  which differs from the usual geodesic shooting approach~\cite{MVNGeodesicShooting-2014}, and report experimental results.
Finally, we study some information-geometric properties~\cite{IG-2016} of the isometric embedding in \S\ref{sec:prop} like the fact that it preserves mixture geodesics (embedded C\&O submanifold is autoparallel with respect to the mixture affine connection) but not exponential geodesics.
Besides, we prove that the Fisher-Rao distance between multivariate normal distributions sharing the same covariance matrix is a scalar function of their Mahalanobis distance in \S\ref{sec:FRsamecovar} using the framework of Eaton~\cite{eaton1989group} of maximal invariants.

\subsection{A closed-form formula for the Fisher-Rao distance between normal distributions sharing the same covariance matrix}\label{sec:FRsamecovar}

Consider the Fisher-Rao distance between $N_1=(\mu_1,\Sigma)$ and $N_1=(\mu_2,\Sigma)$ for a fixed covariance matrix $\Sigma$ and the translation action $a.\mu:=\mu+a$ of the translation group $\bbR^d$ (a subgroup of the affine group). Both the Fisher-Rao distance and the Mahalanobis distance are invariant under translations:
$$
\rho_\calN((\mu_1+a,\Sigma),(\mu_2+a,\Sigma))=\rho_\calN((\mu_1,\Sigma),(\mu_2,\Sigma)),\quad
\Delta_\Sigma(\mu_1+a,\mu_2+a)=\Delta_\Sigma(\mu_1,\mu_2).
$$
To prove that $\rho_\calN((\mu_1,\Sigma),(\mu_2,\Sigma))=h_\FR(\Delta_\Sigma(\mu_1,\mu_2))$ for a scalar function $h_\FR$, we shall prove that the Mahalanobis distance is a maximal invariant and use the
 framework of maximal invariants of Eaton~\cite{eaton1989group} (Chapter~2) who proved that any other invariant function is necessarily a function of a maximal invariant, i.e., a function of the Mahalanobis distance.
The Mahalanobis distance is maximal invariant because  when 
$\Delta_\Sigma(\mu_1,\mu_2)=\Delta_\Sigma(\mu_1',\mu_2')$   there exists $a\in\bbR^d$ such that $(\mu_1+a,\mu_2+a)=(\mu_1',\mu_2')$:
Indeed, we may consider the Cholesky decomposition $\Sigma=LL^\top$ so that $\Delta_\Sigma(\mu_1,\mu_2)=\Delta_I(L^\top \mu_1,L^\top\mu_2)$.
Let $m_1=L^\top\mu_1$, $m_2=L^\top\mu_2$, $m_1'=L^\top\mu_1'$ and $m_2'=L^\top\mu_2'$.
We have to prove equivalently that when $\|m_1-m_2\|_2=\|m_1'-m_2'\|_2$ that 
there exists $a\in\bbR^d$ such that $(m_1+a,m2_+a)=(m_1',m_2')$. 
It suffices to let $a=m_1'-m_1$ and consider $m_1-m_2=m_1'-m_2'$.
Then we have $m_2+a=m_2+m_1'-m_1=m_2'-m_1'+m_1'=m_2'$.
Thus using Eaton's theorem~\cite{eaton1989group}, there exists a scalar function $h_\FR$ such that 
$\rho_\calN((\mu_1,\Sigma),(\mu_2,\Sigma))=h_\FR(\Delta_\Sigma(\mu_1,\mu_2))$.

To find explicitly the scalar function $h_\FR(\cdot)$, let us consider the univariate case of normal distributions for which the Fisher-Rao distance is given in closed form in Eq.~\ref{eq:FR1D}. In that case, the univariate Mahalanobis distance is $\Delta_{\sigma^2}(\mu_1,\mu_2)=\sqrt{(\mu_2-\mu_1)(\sigma^2){-1}(\mu_2-\mu_1)}=\frac{|\mu_2-\mu_1|}{\sigma}$ and we can write formula of Eq.~\ref{eq:FR1D} as $h_\FR(\Delta_{\sigma^2}(\mu_1,\mu_2))$ with
\begin{eqnarray}
h_\FR(u) &=& \sqrt{2}\,\log\left(\frac{\sqrt{8+u^2}+u}{\sqrt{8+u^2}-u}\right),\label{eq:h1d}\\
 &=& \sqrt{2}\,\arccosh\left(1+\frac{1}{4}u^2\right).\label{eq:h1dbis}
\end{eqnarray}
 
\begin{Proposition}\label{prop:FRsamecovar}
The Fisher-Rao distance $\rho_\calN((\mu_1,\Sigma),(\mu_2,\Sigma))$ between two MVNs with same covariance matrix is
\begin{eqnarray}\label{eq:FRsamecovar}
\rho_\calN((\mu_1,\Sigma),(\mu_2,\Sigma)) &=&\sqrt{2}\,\log\left(\frac{\sqrt{8+\Delta_\Sigma^2(\mu_1,\mu_2)}+\Delta_\Sigma(\mu_1,\mu_2)}{\sqrt{8+\Delta_\Sigma^2(\mu_1,\mu_2)}-\Delta_\Sigma(\mu_1,\mu_2)}\right)=\rho_\calN((0,1),(\Delta_\Sigma(\mu_1,\mu_2),1)),\\
&=& \sqrt{2}\, \arccosh\left(1+\frac{1}{4}\Delta_\Sigma^2(\mu_1,\mu_2)\right),\label{eq:FRh1dbis}
\end{eqnarray}
where $\Delta_\Sigma(\mu_1,\mu_2)=\sqrt{(\mu_2-\mu_1)^\top\Sigma^{-1}(\mu_2-\mu_1)}$ is the Mahalanobis distance.
\end{Proposition}

Indeed, notice that the $d$-variate Mahalanobis distance $\Delta_\Sigma(\mu_1,\mu_2)$ can be interpreted as a univariate Mahalanobis distance between the standard normal distribution $N(0,1)$ and $N(\Delta_\Sigma(\mu_1,\mu_2),1)$:
$$
\Delta_\Sigma(\mu_1,\mu_2)=\Delta_1(0,\Delta_\Sigma(\mu_1,\mu_2)).
$$
Thus we have $\rho_\calN((\mu_1,\Sigma),(\mu_2,\Sigma))=\rho_\calN((0,1),(\Delta_\Sigma(\mu_1,\mu_2),1))$,
where the right hand-side term is the univariate Fisher-Rao distance of Eq.~\ref{eq:FR1D}.
Let us notice that the square length element on $\calM_\Sigma$ is $\ds^2=\dmu^\top\Sigma^{-1}\dmu=\Delta_\Sigma^2(\mu,\mu+\dmu)$.

Let us corroborate this result by checking formula~of Eq.~\ref{prop:FRsamecovar} with two examples in the literature:
In~\cite{strapasson2015bounds} (Figure 4), we Fisher-Rao distance between
$N_1=(0,I)$ and $N_2=\left(\vectortwo{\frac{1}{2}}{\frac{1}{2}},I\right)$ is studied.
We find $\rho_\calN(N_1,N_2)=0.69994085$ in accordance to their result shown in Figure~4.
The second example is Example~1 of~\cite{FRMVNReview-2020} (p. 11) with
$N_1=\left(\vectortwo{-1}{0},\Sigma\right)$ and $N_2=\left(\vectortwo{6}{3},\Sigma\right)$ for 
$\Sigma=\mattwotwo{1.1}{0.9}{0.9}{1.1}$.
Formula of Eq.~\ref{eq:FRsamecovar} yields the Fisher-Rao distance $5.006483034546878$
in accordance with~\cite{FRMVNReview-2020} which reports $5.00648$.

Similarly, the statistical Ali-Silvey-Csisz\'ar $f$-divergences~\cite{fdiv-AliSilvey-1966,Csiszar-1967}
$$
I_f[p_{(\mu_1,\Sigma)}:p_{(\mu_2,\Sigma)}]=\int_{\bbR^d}  p_{(\mu_1,\Sigma)}(x)\, f\left(\frac{ p_{(\mu_2,\Sigma)}}{ p_{(\mu_1,\Sigma)}}\right)  \mathrm{d}x,
$$
between two MVNs sharing the same covariance matrix are increasing functions of the Mahalanobis distance because the $f$-divergences between two MVNs sharing the same covariance matrix are invariant under the action of the translation group~\cite{nielsen2022note}. Thus we have 
$I_f[p_{(\mu_1,\Sigma}:p_{(\mu_2,\Sigma)}]=h_f(\Delta_\Sigma(\mu_1,\mu_2))$. 
Since $\Delta_\Sigma(\mu_1,\mu_2)=\Delta_1(0,\Delta_\Sigma(\mu_1,\mu_2))$, we thus have
$$
I_f[p_{(\mu_1,\Sigma}:p_{(\mu_2,\Sigma)}]=h_f(\Delta_1(0,\Delta_\Sigma(\mu_1,\mu_2))=I_f[p_{(0,1}:p_{(\Delta_\Sigma(\mu_1,\mu_2),1)}],
$$
where the right-hand side $f$-divergence is between univariate normal distributions.
See Table~2 of~\cite{nielsen2022note} for some explicit functions $h_f$.

\section{Calvo and Oller's family of diffeomorphic embeddings}\label{sec:CO}
Calvo and Oller~\cite{SDPMVN-1990,SDPElliptical-2002} noticed that we can be can embed the space of normal distributions in $\calP(d+1)$ by  using the following   mapping:
\begin{equation}
f_{\beta}(N)=f_{\beta}(\mu,\Sigma)=
  \mattwotwo{\Sigma+\beta\mu\mu^\top}{\beta\mu}{\beta\mu^\top}{\beta}\in\calP(d+1),
\end{equation}
where $\beta\in\bbR_{>0}$ and $N=N(\mu,\Sigma)$.
Notice that since the dimension of $\calP(d+1)$ is $\frac{(d+1)(d+2)}{2}$, we only use $\frac{(d+1)(d+2)}{2}-\frac{d(d+3)}{2}=1$ extra dimension for embedding $\calN(d)$ into $\calP(d+1)$. 
By foliating $\bbP=\bbR_{>0}\times \bbP_c$ where $\bbP_c=\{P\in\bbP\st |P|=c\}$ denotes the subsets of $\bbP$ with determinant $c$, we get the following Riemannian Calvo \& Oller metric on the SPD cone:
\begin{eqnarray*}
\ds^2_\CO &=& \frac{1}{2}\tr\left(\left(f^{-1}(\mu,\Sigma)\mathrm{d}f(\mu,\Sigma)\right)^2\right),\\
&=& \frac{1}{2}\left(\frac{\dbeta}{\beta}\right)^2 + \beta\dmu^\top \Sigma^{-1}\dmu+\frac{1}{2}\tr\left(\left(\Sigma^{-1}\dSigma\right)^2\right).
\end{eqnarray*}

Let 
$$
\barN_\beta(d) = \left\{\barP=f_{\beta}(\mu,\Sigma) \st (\mu,\Sigma)\in \calN(d)=\bbR^d\times\calP(d)\right\}
$$
 denote the submanifold of $\calP(d+1)$ of codimension $1$, and $\barN=\barN_1$ (i.e., $\beta=1$). 
The family of mappings $f_{\beta}$ provides diffeomorphisms between $\calN(d)$ and $\barN_\beta(d)$. 
Let $f_{\beta}^{-1}(\barP)=(\mu_\barP,\Sigma_\barP)$ denote the inverse mapping for $\barP\in\barN_\beta(d)$, and let $f=f_1$ (i.e., $\beta=1$): $f(\mu,\Sigma)=\mattwotwo{\Sigma+\mu\mu^\top}{\mu}{\mu^\top}{1}$.

By equipping the cone $\calP(d+1)$ by the trace metric~\cite{moakher2011riemannian,MIG-2013,EllipticIsometrySPD-2021} (also called the affine invariant Riemannian metric, AIRM) scaled by $\frac{1}{2}$: 
$$
g_P^\trace(P_1,P_2):=\tr(P^{-1}P_1P^{-1}P_2)
$$ 
 (yielding the squared line element $\ds_\calP^2=\frac{1}{2}\tr((P\,\dP)^2)$),
Calvo and Oller~\cite{SDPMVN-1990} proved that $\barN(d)$ is isometric to $\calN(d)$ (i.e., the Riemannian metric of $\calP(d+1)$ restricted to $\calN(d)$ coincides with the Riemannian metric of $\calN(d)$ induced by $f$) but $\barN(d)$ is not totally geodesic (i.e., the geodesics $\gamma_\calP(\barP_1,\barP_2;t)$ for $\barP_1=f(N_1),\barP_2=f(N_2)\in\barN(d)$ leaves the embedded normal submanifold $\barN(d))$.
Note that $g_P^\trace$ can be interpreted as the Fisher metric for the family $\calN_0$ of $0$-centered normal distributions.
Thus we have $(\calN(d),g^\Fisher)\hookrightarrow (\calP(d+1),g^\trace)$, and the following diagram between parameter spaces and corresponding distributions:
$$
\begin{array}{ccc}
\calN(d) & \hookrightarrow  &\calN_0(d+1)\\
\updownarrow & & \updownarrow\\
\Lambda(d) & \hookrightarrow & \bbP(d+1)
\end{array}
$$

\begin{Remark}
The trace metric was first studied by Siegel~\cite{siegel2014symplectic,nielsen2020siegel} using the wider scope of complex symmetric matrices with positive-definite imaginary parts generalizing the Poincar\'e upper half-plane (see Appendix~\ref{sec:Siegel}).
\end{Remark}

We omit to specify the dimensions and write for short $\calN$, $\barN$ and $\calP$  when clear from context.
Thus C\&O proposed to use the embedding $f=f_{1}$ to give a lower bound $\rho_\CO$ of the Fisher-Rao distance $\rho_\calN$ between normals:
\begin{equation}
\mathrm{LC}_\CO:\quad \rho_{\calN}(N_1,N_2)\geq \rho_{\CO}(\underbrace{f(\mu_1,\Sigma_1)}_{\barP_1},\underbrace{f(\mu_2,\Sigma_2)}_{\barP_2})=
\sqrt{\frac{1}{2}\sum_{i=1}^{d+1} \log^2 \lambda_i(\barP_1^{-1}\barP_2)}.
\end{equation}

We let $\rho_{\CO}(N_1,N_2)=\rho_\CO(f(N_1),f(N_2))$.
The $\rho_\CO$ distance is invariant under affine transformations like the Fisher-Rao distance of Property~\ref{prop:aifr}:

\begin{Property}[affine-invariance of C\&O distance~\cite{SDPMVN-1990}]\label{prop:aico}
For all $A\in\GL(d), a\in\bbR^d$,
we have
$\rho_{\CO}((A\mu_1+a,A\Sigma_1 A^\top),(A\mu_2+a,A\Sigma_2 A^\top)) = \rho_{\CO}(N(\mu_1,\Sigma_1),N(\mu_2,\Sigma_2))$.
\end{Property}

When $\Sigma_1=\Sigma_2=\Sigma$, we have $|\barP_1|=|\barP_2|=|\Sigma|$.
Since the Riemannian geodesics $\gamma_\bbP(P_1,P_2;t)$ in the SPD cone are given by $\gamma_\bbP(P_1,P_2;t)=P_1^{\frac{1}{2}}(P_1^{-\frac{1}{2}}P_2P_1^{-\frac{1}{2}})^t P_1^{\frac{1}{2}}$~\cite{RieMinimax-2013} (also written $\gamma_\SPD(P_1,P_2;t)$), we have
$|\gamma_\bbP(P_1,P_2;t)|=|\Sigma|$. 
Although the submanifold $\bbP_c=\{P\in\bbP\st |P|=c\}$ is totally geodesic with respect to the trace metric, it is not totally geodesic with respect to $\frac{1}{2}\tr((\barP\mathrm{d}\barP)^2)$.
Thus although $\gamma_\bbP(P_1,P_2)\in\barN$, it does not correspond to the embedded MVN geodesics with respect to the Fisher metric.
The C\&O distance between two MVNs $N(\mu_1,\Sigma)$ and $N(\mu_2,\Sigma)$ sharing the same covariance matrix~\cite{SDPMVN-1990} is
\begin{equation}\label{eq:cosamecovar}
\rho_{\CO}(N(\mu_1,\Sigma),N(\mu_2,\Sigma))=\arccosh\left(1+\frac{1}{2}\Delta_\Sigma^2(\mu_1,\mu_2)\right),
\end{equation}
where $\arccosh(x):=\log(x+\sqrt{x^2-1})$ for $x\geq 1$ and $\Delta_\Sigma(\mu_1,\mu_2)$ is the Mahalanobis distance between $N(\mu_1,\Sigma)$ and $N(\mu_2,\Sigma)$. In that case,  we thus have $\rho_{\CO}(N(\mu_1,\Sigma),N(\mu_2,\Sigma))=h_\CO(\Delta_\Sigma(\mu_1,\mu_2))$ where
$h_\CO(u)=\arccosh\left(1+\frac{1}{2}u^2\right)$ is a strictly monotone increasing function.
Let us note in passing that in~\cite{SDPMVN-1990} (Corollary, page 230) there is a confusing or typographic error since the distance is reported as $\arccosh\left(1+\frac{1}{2}d_M(\mu_1,\mu_2)\right)$ where $d_M$ denotes ``Mahalanobis distance''~\cite{mahalanobis1936generalised}. So either $d_M=\Delta_\Sigma^2$, Mahalanobis $D^2$-distance, or there is a missing square in the equation of the Corollary page 230.
To get a flavor of how good is the approximation of the C\&O distance, we may consider the same-covariance case where we have both closed-form solutions for $\rho_{\calN}$ (Eq.~\ref{eq:FRh1dbis}) and $\rho_{\CO}$ (Eq.~\ref{eq:cosamecovar}).
Figure~\ref{fig:qualityFRCO} plots the two functions $h_\CO$ and $h_\FR$ (with $h_\CO(u)\leq h_\FR(u)\leq u$ for $u\in[0,\infty)$).

\begin{figure}
\centering
\includegraphics[width=0.65\textwidth]{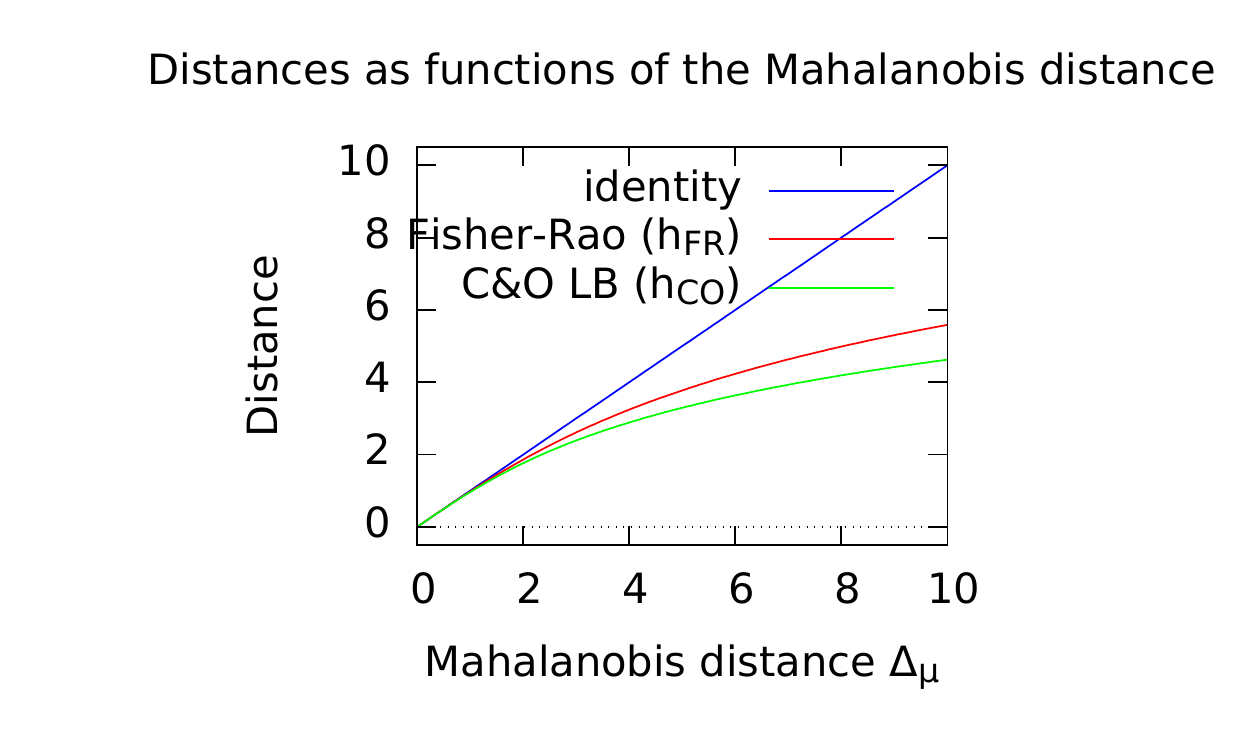}
\caption{Quality of the C\&O lower bound compared to the exact Fisher-Rao distance in the case of $N_1,N_2\in\calM_\Sigma$ (MVNs sharing the same covariance matrix $\Sigma$). We have $\rho_\CO\leq \rho_\calN\leq \Delta_\Sigma$. \label{fig:qualityFRCO}}
\end{figure}

Let us  remark that similarly all $f$-divergences between $N_1=(\mu_1,\Sigma)$ and $N_2=(\mu_2,\Sigma)$ are scalar functions of their Mahalanobis distance $\Delta_\Sigma(\mu_1,\mu_2)$ too, see~\cite{nielsen2022note}.

The C\&O distance $\rho_{\CO}$ is a metric distance that has been used in many applications ranging from computer vision~\cite{ceolin2012computing,wang2017g2denet,nguyen2021geomnet,miyamoto2022fisher} to signal/sensor processing, statistics~\cite{kurtek2015bayesian,marti2016optimal}, machine learning~\cite{tang2018information,DA-FisherRaoMVN-2015,li2016local,liang2019fisher,picot2022adversarial,collas2022use} and analogical reasoning~\cite{murena2018opening}.

\begin{Remark}\label{rk:secondco}
In a second paper, Calvo and Oller~\cite{SDPElliptical-2002} noticed that we can be can embed normal distributions in $\calP(d+1)$ by the following  more general  mapping (Lemma 3.1~\cite{SDPElliptical-2002}):
\begin{equation}
g_{\alpha,\beta,\gamma}(\mu,\Sigma)=
|\Sigma|^{\alpha} \mattwotwo{\Sigma+\beta\gamma^2\mu\mu^\top}{\beta\gamma\mu}{\beta\gamma\mu^\top}{\beta}\in\calP(d+1),
\end{equation}
where $\alpha\in\bbR$, $\beta\in\bbR_{>0}$ and $\gamma\in\bbR$.
It is show in~\cite{SDPElliptical-2002} that the induced length element is
$$
\ds^2_{\alpha,\beta,\gamma}=\frac{1}{2}\left(
\alpha ((d+1)+2\alpha)\tr^2(\Sigma^{-1}\dSigma)+\tr((\Sigma^{-1}\dSigma)^2)
+2\beta\gamma^2\dmu^\top\Sigma^{-1}\dmu+2\alpha \tr(\Sigma^{-1}\dSigma)\frac{\dbeta}{\beta}+
\left(\frac{\dbeta}{\beta}\right)^2
\right).
$$
When $\gamma=\beta=1$, we have
$$
\ds^2_\alpha=\frac{1}{2}\left(
\alpha ((d+1)+2\alpha)\tr^2(\Sigma^{-1}\dSigma)+\tr((\Sigma^{-1}\dSigma)^2)
+2\beta\gamma^2\dmu^\top\Sigma^{-1}\dmu
\right).
$$
Thus to cancel the term $\tr^2(\Sigma^{-1}\dSigma)$, we may either choose $\alpha=0$ or $\alpha=-\frac{2}{1+d}$.

In some applications~\cite{popovic2022measure}, the embedding 
\begin{equation}\label{eq:g}
g_{-\frac{1}{d+1},1,1}(\mu,\Sigma)=|\Sigma|^{-\frac{1}{d+1}} \mattwotwo{\Sigma+\mu\mu^\top}{\mu}{\mu^\top}{1}:=\hat{f}(\mu,\Sigma),
\end{equation}
is used to ensure that $\left|g_{-\frac{1}{d+1},1,1}(\mu,\Sigma)\right|=1$.
That is normal distributions are embedded diffeomorphically into the submanifold of positive-definite matrices with unit determinant (also called SSPD, acronym of Special SPD).
In~\cite{SDPElliptical-2002}, C\&O showed that there exists a second isometric embedding of the Fisher-Rao Gaussian manifold $\calN(d)$ into a submanifold of the cone $\calP(d+1)$: $f_\SSPD(\mu,\Sigma)=|\Sigma|^{-\frac{2}{d+1}} \mattwotwo{\Sigma+\mu\mu^\top}{\mu}{\mu^\top}{1}$.
Let $\hat{P}=f_\SSPD(\mu,\Sigma)$. This  mapping can be understood as taking the elliptic isometry ${P}\mapsto {|P|}^{-\frac{2}{d+1}} {P}$ of $P\in\calP(d+1)$~\cite{EllipticIsometrySPD-2021} since $|\Sigma|=|\bar{P}(\mu,\Sigma)|$ (see proof in Proposition~\ref{prop:KL}). 
It follows that
$$
\rho_\CO(N_1,N_2)=\rho_\calP(\barP_1,\barP_2)=\rho_\calP(\hat{P}_1,\hat{P}_2) \leq \rho_\calN(N_1,N_2).
$$
Similarly, we could have mapped ${P}\mapsto P^{-1}$ to get another isometric embedding.
See the four types of elliptic isometric of the SPD cone described in~\cite{EllipticIsometrySPD-2021}.
Finally, let us remark that the SSPD submanifold is totally geodesic with respect to the trace metric but not with respect to the C\&O metric.
\end{Remark}

The multivariate Gaussian manifold $\calN(d)$ can also be embedded into the SPD cone $\calP(d+1)$ as a Riemannian symmetric space~\cite{lovric2000multivariate,globke2021information}  by  $f_\SSPD$: $\hat\calP=\{f_\SSPD(N)\subset\calP(d+1)\ :\ N\in\calN(d)\}$.
We have $\hat\calP\cong\SL(d+1)/\SO(d+1)$~\cite{lovric2000multivariate,fernandes2000fisher,fernandes2003geometric} (and textbook~\cite{bridson2013metric}, Part II Chapter 10), and the symmetric space $\SL(d+1)/\SO(d+1)$ can be embedded with the Killing Riemannian metric instead of the Fisher information metric:
$$
g^\Killing(N_1,N_2)=\kappa_\Killing\, \left( \mu_1^\top\Sigma^{-1}\mu_2+\frac{1}{2}\tr\left(\Sigma^{-1}\Sigma_1\Sigma^{-1}\Sigma_2\right)
-\frac{1}{2(d+1)}\tr\left(\Sigma^{-1}\Sigma_1\right)\tr\left(\Sigma^{-1}\Sigma_2\right)\right), 
$$
where $\kappa_\Killing>0$ is a predetermined constant (e.g. $1$).
The length element of the Killing metric is
$$
\ds_\SS^2=\kappa_\Killing\, \left(\frac{1}{2}\dmu^\top \Sigma^{-1} \dmu +\tr\left(\left(\Sigma^{-1}\dSigma\right)^2\right)-\frac{1}{2}\tr^2\left(\Sigma^{-1}\dSigma\right)\right).
$$
When we consider $\calN_\Sigma$, we may choose $\kappa_\Killing=2$ so that the Killing metric coincides with the Fisher information metric.
The induced Killing distance~\cite{lovric2000multivariate} is available in closed form:
\begin{equation}\label{eq:KillingDistance}
\rho_\Killing(N_1,N_2)=\sqrt {\kappa_\Killing\, \sum_{i=1}^{d+1} \log^2 \lambda_i\left(\hat{L}_1^{-1}\hat{P}_2\left(\hat{L}_1^{-1}\right)^\top\right)},
\end{equation}
where $\hat{L}_1$ is the unique lower triangular matrix obtained from the Cholesky decomposition of $\hat{P}_1=f_\SSPD(N_1)=\hat{L}_1\hat{L}_1^\top$.  
Note that $\hat{L}_1^{-1}\hat{P}_2\left(\hat{L}_1^{-1}\right)^\top\in\calP(d+1)$ and $|\hat{L}_1|$, i.e., $\hat{L}_1\in\SL(d+1)$.

When $N_1=(\mu_1,\Sigma)$ and $N_2=(\mu_2,\Sigma)$ ($N_1,N_2\in\calN_\Sigma$), we have~\cite{lovric2000multivariate}
$$
\rho_\Killing(N_1,N_2)=\sqrt{2\kappa_\Killing}\arccosh\left(1+\frac{1}{2}\Delta_\Sigma^2(\mu_1,\mu_2)\right),
$$
where $\Delta_\Sigma^2$ is the squared Mahalanobis distance.
Thus $\rho_\Killing(N_1,N_2)=h_\Killing(\Delta_\Sigma(\mu_1,\mu_2))$ where $h_\Killing(u)=\sqrt{2\kappa_\Killing}\arccosh\left(1+\frac{1}{2}u^2\right)$.

When $N_1=(\mu,\Sigma_1)$ and $N_2=(\mu,\Sigma_2)$ ($N_1,N_2\in\calN_\mu$), we have~\cite{lovric2000multivariate}:
$$
\rho_\Killing(N_1,N_2)= \sqrt{
\kappa_\Killing
\left(
\sum_{i=1}^d \log^2\lambda_i\left({L}_1^{-1}{P}_2\left({L}_1^{-1}\right)^\top\right)
-\frac{1}{(d+1)^2}
\left(\sum_{i=1}^d \log\lambda_i\left({L}_1^{-1}{P}_2\left({L}_1^{-1}\right)^\top\right)\right)
\right)
}.
$$


See Example~\ref{ex1entropy}. Let us emphasize that the Killing distance is not the Fisher-Rao distance but is available in closed-form as an alternative metric distance between MVNs.

\section{Approximating the Fisher-Rao distance}\label{sec:approxFR}

\subsection{Approximating length of curves}\label{sec:approxlength}

Rao's distance~\cite{micchelli2005rao} is a Riemannian geodesic distance 
$$\rho_{\calN}(p_{\lambda_1},p_{\lambda_2})=\inf_c \left\{\Length(c) \st c(0)=p_{\lambda_1}, c(1)=p_{\lambda_2} \right\},$$ 
where 
$$
\Length(c)=\int_0^1 \underbrace{\sqrt{\inner{\dot c(t)}{\dot c(t)}}_{c(t)}}_{\ds_\calN(t)} \dt.
$$

We can approximate the Rao distance $\rho_\calN(N_1,N_2)$ by discretizing regularly  any smooth curve $c(t)$ joining $N_1$ ($t=0$) to $N_2$ ($t=1$):
$$
\rho_\calN(N_1,N_2)\leq \frac{1}{T} \sum_{i=1}^{T-1} \rho_\calN\left(c\left(\frac{i}{T}\right),c\left(\frac{i+1}{T}\right)\right),$$
with equality holding iff $c(t)=\gamma_\calN(N_1,N_2;t)$ is the Riemannian geodesic induced by the Fisher information metric.

When $T$ is sufficiently large, the normal distributions 
$c\left(\frac{i}{T}\right)$ and $c\left(\frac{i+1}{T}\right)$ are close to each other, and we can approximate 
$\rho_\calN\left(c\left(\frac{i}{T}\right),c\left(\frac{i+1}{T}\right)\right)$ 
by $\sqrt{D_J\left[c\left(\frac{i}{T}\right),c\left(\frac{i+1}{T}\right)\right]}$, where $D_J[N_1,N_2]=D_\KL[N_1,N_2]+D_\KL[N_2,N_1]$ is Jeffreys divergence, and $D_\KL$ is the Kullback-Leibler divergence:
$$
D_\KL[p_{(\mu_1,\Sigma_1)}:p_{(\mu_2,\Sigma_2)}]
=\frac{1}{2}\left(
\tr(\Sigma_2^{-1}\Sigma_1)+\Delta\mu^\top\Sigma_2^{-1}\Delta\mu-d+\log\frac{|\Sigma_2|}{|\Sigma_1|}
\right).
$$
Thus the costly determinant computations cancel each others in Jeffreys divergence (i.e., $\log\frac{|\Sigma_2|}{|\Sigma_1|}+\log\frac{|\Sigma_1|}{|\Sigma_2|}=0$) and we have:
$$
D_J[p_{(\mu_1,\Sigma_1)}:p_{(\mu_2,\Sigma_2)}]=\tr\left(\frac{\Sigma_2^{-1}\Sigma_1+\Sigma_1^{-1}\Sigma_2}{2}-I\right)
+\Delta\mu^\top\frac{\Sigma_1^{-1}+\Sigma_2^{-1}}{2}\Delta\mu.
$$
Figure~\ref{fig:method} summarizes our method to approximate the Fisher-Rao geodesic distance.

\begin{figure}
\centering
\includegraphics[width=0.6\textwidth]{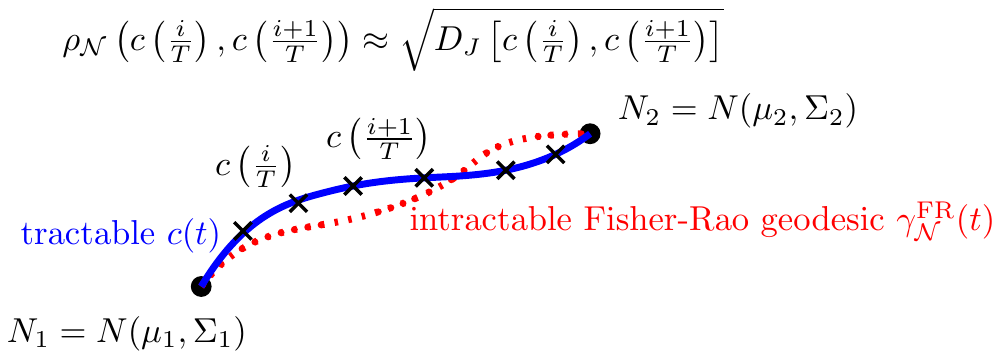}
\caption{Approximating the Fisher-Rao geodesic distance $\rho_\calN(N_1,N_2)$: 
The Fisher-Rao geodesic $\gamma_\calN^{\mathrm{FR}}$ is not known in closed-form.
 We consider a tractable curve $c(t)$, discretize $c(t)$ at $T+1$ points $c(\frac{i}{T})$ with $c(0)=N_1$ and $c(1)=N_2$, and approximate $\rho_\calN\left(c\left(\frac{i}{T}\right),c\left(\frac{i+1}{T}\right)\right)$ by $D_J\left[c\left(\frac{i}{T}\right),c\left(\frac{i+1}{T}\right)\right]$. Considering different tractable curves $c(t)$ yield different approximations.
 \label{fig:method}}
\end{figure}

In general, it holds that $I_f[p:q]\approx \frac{f''(1)}{2}\ds^2_\Fisher$ between infinitesimally close distributions $p$ and $q$ ($\ds\approx \sqrt{\frac{2\,I_f[p:q]}{f''(1)}}$), where $I_f[\cdot:\cdot]$ denotes a $f$-divergence~\cite{IG-2016}.
The Jeffreys divergence is a $f$-divergence obtained for $f_J(u)=-\log u+u\log u$ with $f_J''(1)=2$.
It is thus interesting to find low computational cost $f$-divergences between multivariate normal distributions in order to approximate the infinitesimal length element $\ds$. 
Note that $f$-divergences between MVNs are invariant under the action of the affine group~\cite{nielsen2022note}.
Thus for infinitesimally close distributions $p$ and $q$, this informally explains that $\ds_\Fisher$ is invariant under the action of the affine group.

Although the definite integral of the length element along the Fisher-Rao geodesic $\gamma_\calN^{\mathrm{FR}}$ is not known in closed form (i.e., Fisher-Rao distance), the integral of the squared length element along the mixture geodesic $\gamma_\calN^m$ 
and exponential geodesic $\gamma_\calN^e$ coincide with Jeffreys divergence~\cite{IG-2016}:

\begin{Property}[\cite{IG-2016}]
We have $$
D_J[p_{\lambda_1},p_{\lambda_2}]=\int_0^1 \ds_\calN^2(\gamma^m_\calN(p_{\lambda_1},p_{\lambda_2};t)\dt =
\int_0^1 \ds_\calN^2(\gamma^e_\calN(p_{\lambda_1},p_{\lambda_2};t)\dt.
$$
\end{Property}

\begin{proof}
Let us report a proof of this remarkable fact  in the general setting of Bregman manifolds (proof omitted in~\cite{IG-2016}).
Since $D_J[p_{\lambda_1},p_{\lambda_2}]=D_\KL[p_{\lambda_1},p_{\lambda_2}]+D_\KL[p_{\lambda_2},p_{\lambda_1}]$
and $D_\KL[p_{\lambda_1},p_{\lambda_2}]=B_F(\theta(\lambda_2):\theta(\lambda_1))$, where $B_F$ denotes the Bregman divergence induced by the cumulant function of the multivariate normals and $\theta(\lambda)$ is the natural parameter corresponding to $\lambda$, we have
\begin{eqnarray*}
D_J[p_{\lambda_1},p_{\lambda_2}]&=&B_F(\theta_1:\theta_2)+B_F(\theta_2:\theta_1),\\
&=& S_F(\theta_1;\theta_2)=(\theta_2-\theta_1)^\top (\eta_2-\eta_1)=S_{F^*}(\eta_1;\eta_2),
\end{eqnarray*}
where $\eta=\nabla F(\theta)$ and $\theta=\nabla F^*(\eta)$ denote the dual parameterizations obtained by the Legendre-Fenchel convex conjugate  $F^*(\eta)$ of $F(\theta)$.
The proof is based on the first-order and second-order directional derivatives.
The first-order directional derivative $\nabla_u F(\theta)$ with respect to vector $u$ is defined by 
$$
\nabla_u F(\theta)=\lim_{t\rightarrow 0} \frac{F(\theta+tv)-F(\theta)}{t}=v^\top \nabla F(\theta).
$$
 
The second-order directional derivatives $\nabla_{u,v}^2 F(\theta)$ is
\begin{eqnarray*}
\nabla_{u,v}^2 F(\theta) &=& \nabla_{u} \nabla_v F(\theta),\\
 &=& \lim_{t\rightarrow 0} \frac{v^\top \nabla F(\theta+tu)-v^\top\nabla F(\theta)}{t},\\
&=& u^\top \nabla^2 F(\theta) v.
\end{eqnarray*}

Now consider the squared length element $\ds^2(\gamma(t))$ on the primal geodesic $\gamma(t)$ expressed using the primal coordinate system $\theta$:
$\ds^2(\gamma(t))=\dtheta(t)^\top \nabla^2F(\theta(t)) \dtheta(t)$ with $\theta(\gamma(t))=\theta_1+t(\theta_2-\theta_1)$ and $\dtheta(t)=\theta_2-\theta_1$.
Let us express the $\ds^2(\gamma(t))$   using the second-order directional derivative:
$$
\ds^2(\gamma(t))=\nabla^2_{\theta_2-\theta_1}  F(\theta(t)).
$$
Thus we have $\int_0^1 \ds^2(\gamma(t))\dt=[\nabla_{\theta_2-\theta_1}  F(\theta(t))]_0^1$,
where the first-order directional derivative is $\nabla_{\theta_2-\theta_1}  F(\theta(t))=(\theta_2-\theta_1)^\top \nabla F(\theta(t))$.
Therefore we get $\int_0^1 \ds^2(\gamma(t))\dt=(\theta_2-\theta_1)^\top (\nabla F(\theta_2)-\nabla F(\theta_1))=S_F(\theta_1;\theta_2)$.

Similarly, we express the squared length element $\ds^2(\gamma^*(t))$ using the dual coordinate system $\eta$ as the second-order directional derivative of $F^*(\eta(t))$ with $\eta(\gamma^*(t))=\eta_1+t(\eta_2-\eta_1)$:
$$
\ds^2(\gamma^*(t))=\nabla^2_{\eta_2-\eta_1}  F^*(\eta(t)).
$$
Therefore, we have  $\int_0^1 \ds^2(\gamma^*(t))\dt=[\nabla_{\eta_2-\eta_1}  F^*(\eta(t))]_0^1=S_{F^*}(\eta_1;\eta2)$.
Since $S_{F^*}(\eta_1;\eta_2)=S_F(\theta_1;\theta_2)$, we conclude that
$$
S_F(\theta_1;\theta_2)=\int_0^1 \ds^2(\gamma(t))\dt=\int_0^1 \ds^2(\gamma^*(t))\dt
$$

In 1D, both pregeodesics $\gamma(t)$ and $\gamma^*(t)$ coincide. We have $\ds^2(t)=(\theta_2-\theta_1)^2 f''(\theta(t))=(\eta_2-\eta_1){f^*}''(\eta(t))$ so that we check that $S_F(\theta_1;\theta_2)=\int_0^1 \ds^2(\gamma(t))\dt=(\theta_2-\theta_1)[f'(\theta(t))]_0^1=(\eta_2-\eta_1)[{f^*}'(\eta(t))]_0^1=(\eta_2-\eta_1)(\theta_2-\theta_2)$.
\end{proof}

It follows the following property:

\begin{Property}[Fisher-Rao upper bound]\label{prop:UBJ}
The Fisher-Rao distance between   normal distributions is upper bounded by the square root of the Jeffreys divergence: $\rho_\calN(N_1,N_2) \leq\sqrt{D_J(N_1,N_2)}$.
\end{Property}

\begin{proof}
Consider the Cauchy-Schwarz inequality for positive functions $f(t)$ and $g(t)$:
 $\int_0^1 f(t)g(t)\dt\leq\sqrt{(\int_0^1 f(t)^2\dt)(\int_0^1 g(t)^2\dt)}$), and let $f(t)=\ds_\calN(\gamma^c_\calN(p_{\lambda_1},p_{\lambda_2};t)$ and $g(t)=1$. 
Then we get: 
$$
\left(\int_0^1 \ds_\calN(\gamma^c_\calN(p_{\lambda_1},p_{\lambda_2};t)\dt\right)^2
\leq \left(\int_0^1 \ds_\calN^2(\gamma^c_\calN(p_{\lambda_1},p_{\lambda_2};t)\dt\right) \left(\int_0^1 1^2 \dt\right)
$$
Furthermore since by definition of $\gamma_\calN^{\mathrm{FR}}$, we have 
$$
\int_0^1 \ds_\calN(\gamma^c_\calN(p_{\lambda_1},p_{\lambda_2};t)\dt\geq \int_0^1 \ds_\calN(\gamma^{\mathrm{FR}}_\calN(p_{\lambda_1},p_{\lambda_2};t)\dt=:\rho_\calN(N_1,N_2),
$$
it follows for $c=e$ (i.e., $e$-geodesic $\gamma_\calN^e$), we have:
$$
\rho_\calN(N_1,N_2)^2 \leq \int_0^1 \ds_\calN^2(\gamma^e_\calN(p_{\lambda_1},p_{\lambda_2};t)\dt = D_J(N_1,N_2).
$$
Thus we have $\rho_\calN(N_1,N_2) \leq\sqrt{D_J(N_1,N_2)}$.

Note that in Riemannian geometry, a curve $\gamma$ minimizes the energy $E(\gamma)=\int_0^1 |\dot\gamma(t)|^2\dt$ if it minimizes the length $L(\gamma)=\int_0^1 \|\dot\gamma(t)\|\dt$ and $\|\dot\gamma(t)\|$ is constant. Using Cauchy-Schwartz inequality, we can show that $L(\gamma)\leq E(\gamma)$.
\end{proof}

Note that this upper bound is tight at infinitesimal scale (i.e., when $N_2=N_1+\mathrm{d}N$).

For any smooth curve $c(t)$, we thus approximate $\rho_\calN$ by
\begin{equation}\label{eq:approx}
\boxed{\tilde\rho_\calN^c(N_1,N_2) := \frac{1}{T} \sum_{i=1}^{T-1} \sqrt{D_J\left[c\left(\frac{i}{T}\right),c\left(\frac{i+1}{T}\right)\right]}.}
\end{equation}

For example, we may consider the following  curves on $\calN$ which admit closed-form parameterizations in $t\in[0,1]$:
\begin{itemize}
	\item linear interpolation 
	$c_\lambda(t)=t (\mu_1,\Sigma_1)+(1-t)(\mu_2,\Sigma_2)$ 
	between $(\mu_1,\Sigma_1)$ and $(\mu_2,\Sigma_2)$,
	
	\item the mixture geodesic~\cite{nielsen2019jensen} $c_m(t)=\gamma^m_\calN(N_1,N_2;t)=(\mu_t^m,\Sigma_t^m)$ with
$\mu_t^m =\bar\mu_t$ and $\Sigma_t^m = \bar\Sigma_t+t\mu_1\mu_1^\top+(1-t)\mu_2\mu_2^\top-\bar\mu_t\bar\mu_t^\top$
where   $\bar\mu_t=t\mu_1+(1-t)\mu_2$ and $\bar\Sigma_t=t\Sigma_1+(1-t)\Sigma_2$,

\item the exponential geodesic~\cite{nielsen2019jensen} $c_e(t)=\gamma_\calN^e(N_1,N_2;t)=(\mu_t^e,\Sigma_t^e)$ with $\mu_t^e = \bar\Sigma_t^H (t\Sigma_1^{-1}\mu_1+(1-t)\Sigma_2^{-1}\mu_2)$
and $\Sigma_t^e =\bar\Sigma^H_t$
where $\bar\Sigma^H_t=(t\Sigma_1^{-1}+(1-t)\Sigma_2^{-1})^{-1}$ is the matrix harmonic mean,
\item the curve $c_{em}(t)=\frac{1}{2}\left(\gamma_\calN^e(N_1,N_2;t)+\gamma^m_\calN(N_1,N_2;t)\right)$ which is obtained by averaging the mixture geodesic with the exponential geodesic.
\end{itemize}

Let us denote by $\tilde\rho^{\lambda}_\calN=\tilde\rho^{c_\lambda}_\calN$,
 $\tilde\rho^{m}_\calN=\tilde\rho^{c_m}_\calN$, 
 $\tilde\rho^{e}_\calN=\tilde\rho^{c_e}_\calN$ and 
 $\tilde\rho^{em}_\calN=\tilde\rho^{c_{em}}_\calN$ the approximations obtained by these curves following from Eq.~\ref{eq:approx}.
Figure~\ref{fig:vizemgeo2d} visualizes the exponential and mixture geodesics between two bivariate normal distributions.
When $T$ is sufficiently large, the approximated distances $\tilde\rho^{x}$ are close to the length of curve $x$, and we may thus consider several curves $c_i$ and report the smallest Fisher-Rao distance approximations obtained: $\rho_\calN(N_1,N_2)\approx \min_i \tilde\rho_\calN^{c_i}(N_1,N_2)$. 

\begin{figure}
\centering
\includegraphics[width=0.85\textwidth]{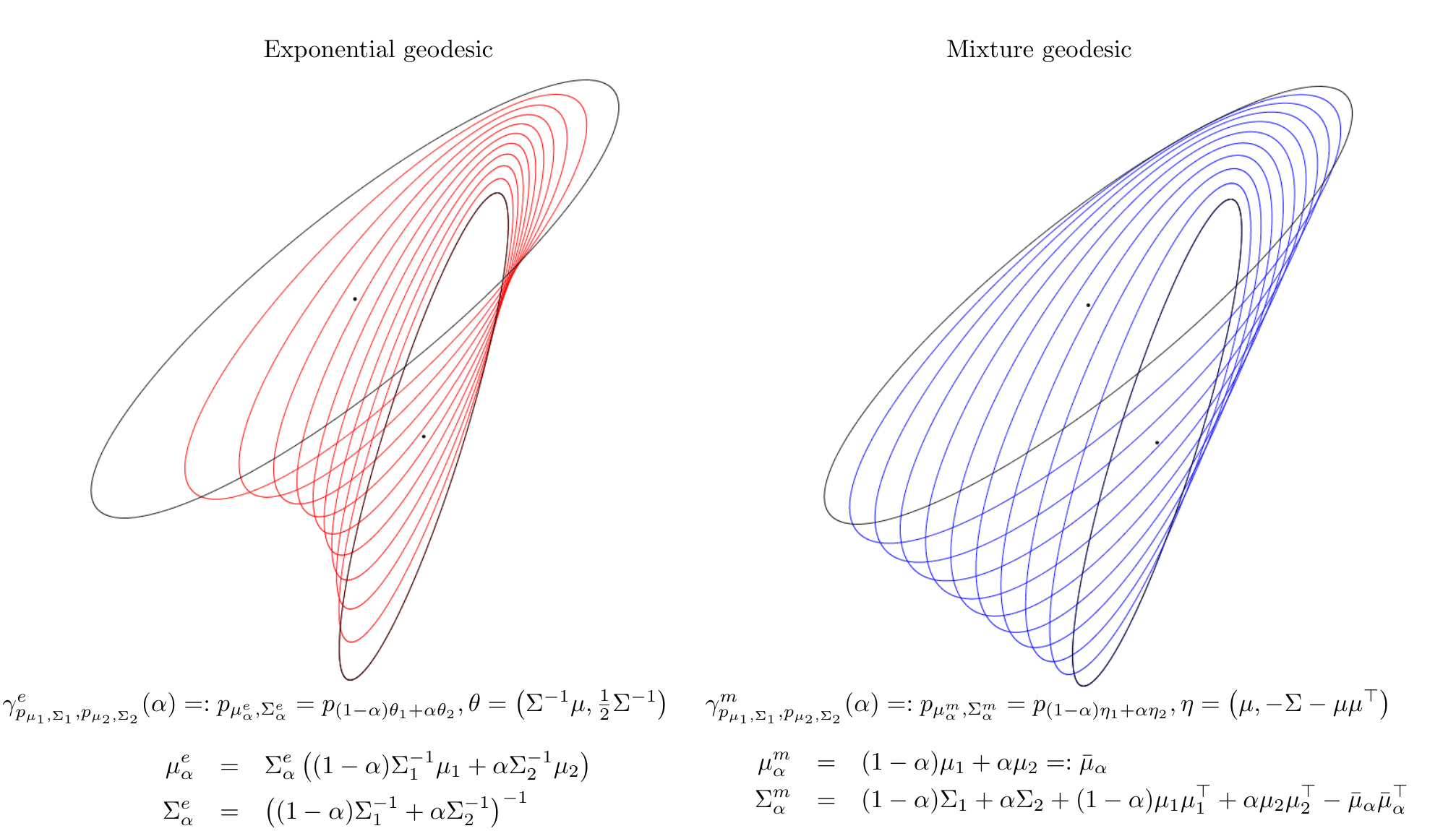} 
\caption{Visualizing the exponential and mixture geodesics between two bivariate normal distributions.\label{fig:vizemgeo2d}}
\end{figure}

Note that we consider the regular spacing for approximating a curve length and do not optimize the position of the sample points on the curve.
Indeed, as $T\rightarrow\infty$, the curve length approximation tends to the Riemannian curve length.
In other words,  we can measure approximately finely the length of any curve available with closed-form reparameterization  by increasing $T$.
Thus the key question of our method is how to best approximate the Fisher-Rao geodesic by a closed-form curve.

\begin{Remark}
In~\cite{globke2021information}, the authors consider the embedding of Eq.~\ref{eq:g} and use the Killing metric $g^\Killing$ instead of the Fisher metric defined by:
$$
g_N^\Killing(N_1,N_2)=\mu_1^\top\Sigma^{-1}\mu_2+\frac{1}{2}\tr\left(\Sigma^{-1}\Sigma_1\Sigma^{-1}\Sigma_2\right)
-\frac{1}{2(d+1)}\tr\left(\Sigma^{-1}\Sigma_1\right)\tr\left(\Sigma^{-1}\Sigma_2\right),
$$ 
where $N=(\mu,\Sigma)$, $N_1=(\mu_1,\Sigma_1)$, and $N_2=(\mu_2,\Sigma_2)$.
A Fisher geodesic defect measure of a curve $c$ is defined by 
$$
\delta(c)=\lim_{s\rightarrow \infty} \frac{1}{s}\int_0^s \|\nabla_{\dot{c}}^{g^\Fisher} \dot{c}\|_{c(t)}^\Fisher \dt,
$$
where $\nabla^{g^\Fisher}$ denotes the Levi-Civita connection induced by the Fisher metric. 
When $\delta(c)=0$ the curve is said an asymptotic geodesic of the Fisher geodesic.
It is proven that Killing geodesics at $(\mu,\Sigma)$ are asymptotic Fisher geodesics when the initial condition $c'(0)$ is orthogonal to $\calN_\mu$. 
\end{Remark}

Next, we introduce yet another curve $c_\CO(t)$ derived from Calvo \& Oller isometric mapping $f$ which experimentally behaves better when normals are not {\it too far} from each others.

\subsection{Calvo \& Oller's curve}\label{sec:cocurve}
This approximation consists in leveraging the closed-form expression of the SPD geodesics~\cite{moakher2011riemannian,RieMinimax-2013}: 
$$
\gamma_\calP(P,Q;t)=P^{\frac{1}{2}} \, \left(P^{-\frac{1}{2}}Q^{\frac{1}{2}}P^{-\frac{1}{2}}\right)^t \, P^{\frac{1}{2}},
$$ 
to approximate the Fisher-Rao normal geodesic $\gamma_\calN(N_1,N_2;t)$ as follows:
Let $\barP_1=f(N_1),\barP_2=f(N_2)\in\barN$, and consider the smooth curve 
$$
\bar c_\CO(P_1,P_2;t)=\proj_{\barN}\left(\gamma_\calP(\barP_1,\barP_2;t)\right),
$$ 
where $\proj_{\barN}(P)$ denotes the Fisher orthogonal projection of $P\in\calP(d+1)$ onto $\barN$ (Figure~\ref{fig:FRapprox}).
Thus curve $c_\CO$ is then defined as $f^{-1}(\bar c_\CO)$.
Note that the power matrix $P^t$ is $U\diag(\lambda_1^t,\ldots,\lambda_d^t) V^\top$ where 
$P=U\diag(\lambda_1^t,\ldots,\lambda_d^t) V^\top $ is the eigenvalue decomposition of $P$.

\begin{figure}
\centering
\includegraphics[width=0.85\textwidth]{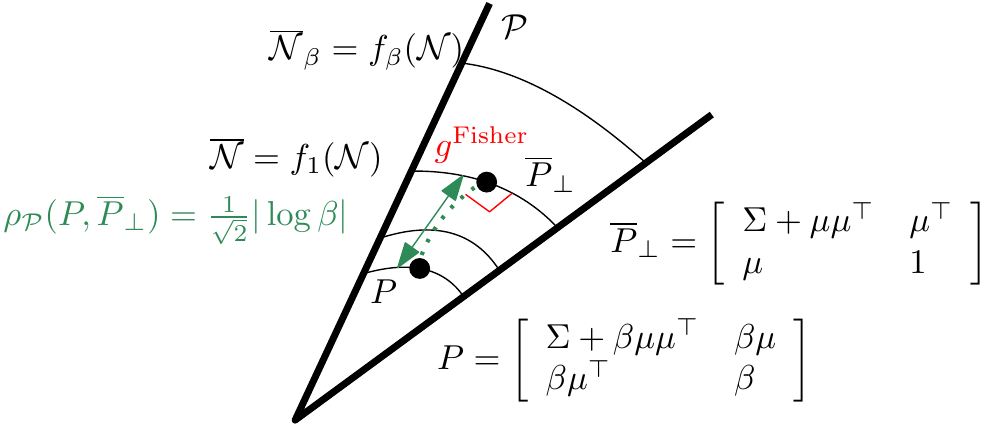}
\caption{Projecting a SPD matrix $P\in\calP$ onto $\barN=f(\calN)$.
 \label{fig:projection}}
\end{figure}
 
Let us now explain how to project $P=[P_{i,j}]\in\calP(d+1)$ onto $\barN$ based on  the analysis of the Appendix of~\cite{SDPMVN-1990} (page 239):

\begin{Proposition}[Projection of a SPD matrix onto the embedded normal submanifold $\barN$]\label{prop:proj}
Let $\beta=P_{d+1,d+1}$ and write  $P=\mattwotwo{\Sigma+\beta\mu\mu^\top}{\beta\mu}{\beta\mu^\top}{\beta}$.
Then the orthogonal projection at $P\in\calP$ onto $\barN$ is: 
$$
\barP_\perp=\proj_{\barN}(P)=\mattwotwo{\Sigma+\mu\mu^\top}{\mu^\top}{\mu}{1},
$$
and the SPD distance between $P$ and $\barP_\perp$  is $\rho_\calP(P,\barP_\perp)=\frac{1}{\sqrt{2}} |\log\beta|$.
\end{Proposition}

\begin{figure}
\centering
\begin{tabular}{cc}
\includegraphics[width=0.4\textwidth]{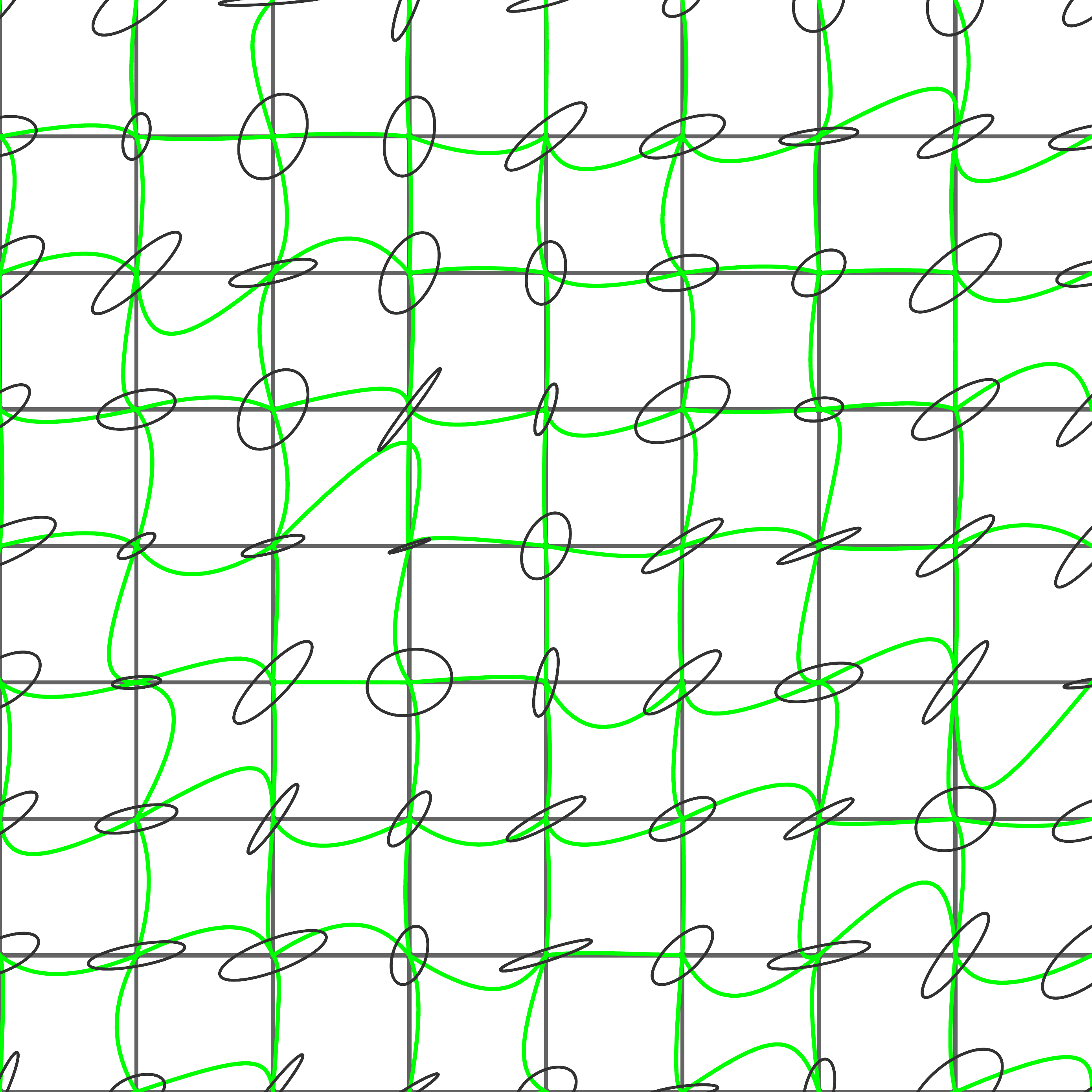} &
\includegraphics[width=0.4\textwidth]{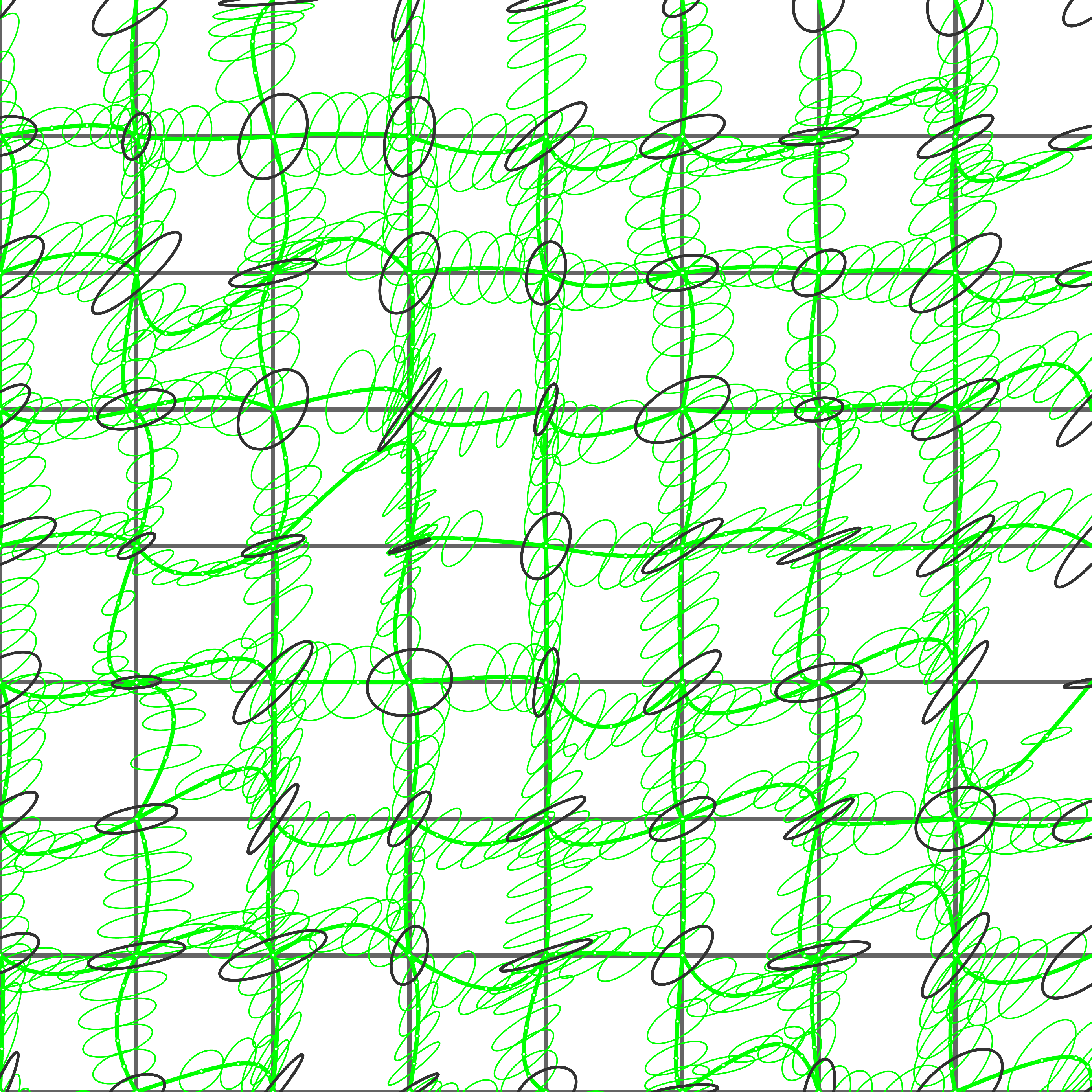} \cr
(a) & (b)
\end{tabular}
\caption{Diffusion tensor imaging on a 2D grid: (a) Ellipsoids shown at the $8\times 8$ grid locations with C\&O curves in green, and (b) some interpolated  ellipsoids are further shown along the C\&O curves.
 \label{fig:gridN}}
\end{figure}

\begin{Remark}
In Diffusion Tensor Imaging~\cite{MVNGeodesicShooting-2014} (DTI), 
the Fisher-Rao distance can be used to evaluate the distance between 3D normal distributions with means located at a 3D grid position.
We may consider $3\times 3\times 3-1=26$ neighbor graphs induced by the grid, and for each  normal $N$ of the grid, calculate the approximations of the Fisher-Rao distance of $N$ with its neighbors $N'$ as depicted in Figure~\ref{fig:gridN}. 
Then the distance between two tensors $N_1$ and $N_2$ of the 3D grid is calculated as the shortest path on the weighted graph using Dijkstra's algorithm~\cite{MVNGeodesicShooting-2014}.
\end{Remark}

Note that the Fisher-Rao projection of $N=(\mu,\Sigma)$ on a submanifold with fixed mean $\mu_0$ was recently reported in closed-form in~\cite{tang2018information}.

\begin{figure}
\centering
\includegraphics[width=\textwidth]{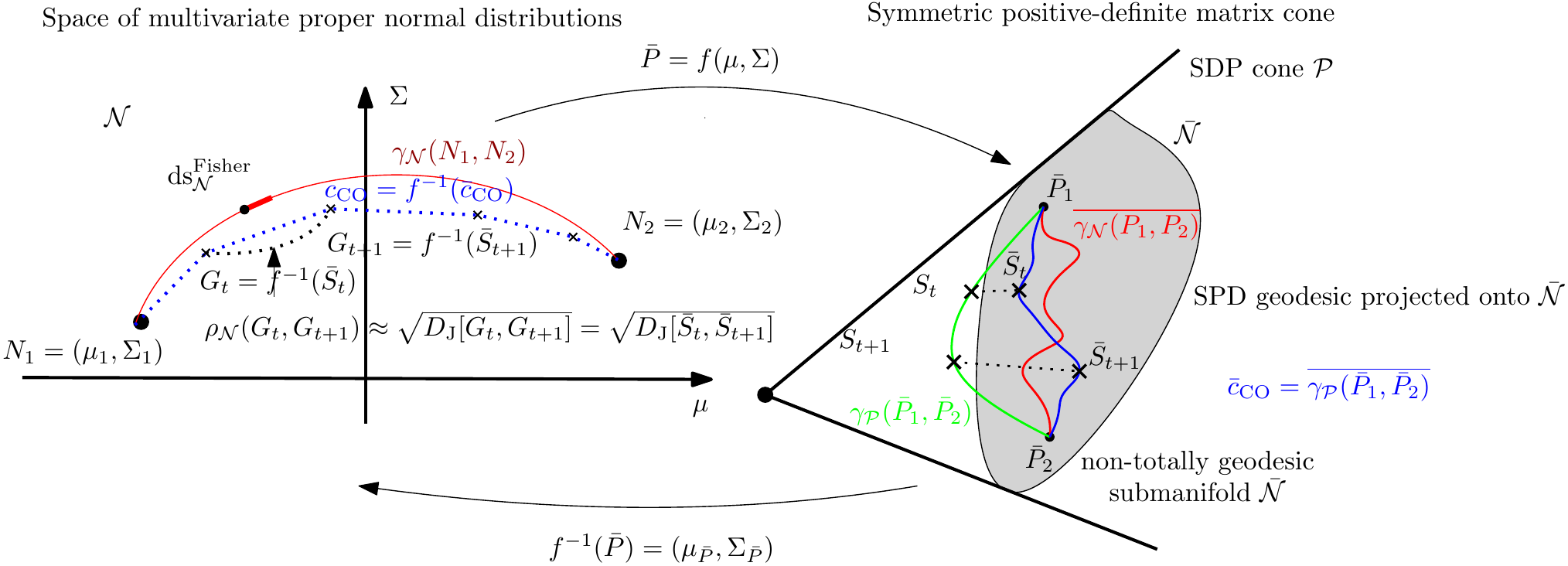}
\caption{Illustration of the approximation of the Fisher-Rao distance between two multivariate normals $N_1$ and $N_2$ (red geodesic length $\gamma_{\mathcal{N}}(N_1,N_2)$ by discretizing curve $\bar c_\CO\in\barN$ or equivalently curve $c_\CO\in\calN$.
 \label{fig:FRapprox}}
\end{figure}

Let $\bar c_\CO(t)=\bar S_t$ and $c_\CO(t)=f^{-1}(c_\CO(t))=:G_t$.
The following proposition shows that we have $D_J[\bar S_t,\bar S_{t+1}]=D_J[G_t,G_{t+1}]$.

\begin{Proposition}\label{prop:KL}
The Kullback-Leibler divergence   between 
$p_{\mu_1,\Sigma_1}$ and $p_{\mu_2,\Sigma_2}$ amounts to the KLD
between $q_{\bar P_1}=p_{0,f(\mu_1,\Sigma_1)}$ and $q_{\bar P_2}=p_{0,f(\mu_2,\Sigma_2)}$ where $\bar P_i=f(\mu_i,\Sigma_i)$:
$$
D_\KL[p_{\mu_1,\Sigma_1}:p_{\mu_2,\Sigma_2}]=D_\KL[q_{\bar P_1}:q_{\bar P_2}].
$$
\end{Proposition}

\begin{proof}
The KLD between two centered $(d+1)$-variate normals $q_{P_1}=p_{0,P_1}$ and $q_{P_2}=p_{0,P_2}$ is
$$
D_\KL[q_{P_1}:q_{P_2}]
=\frac{1}{2}\left(
\tr(P_2^{-1}P_1)-d-1+\log\frac{|P_2|}{|P_1|}
\right).$$
This divergence can be interpreted as the matrix version of the Itakura-Saito divergence~\cite{davis2006differential}.
The SPD cone equipped with $\frac{1}{2}$ of the trace metric can be interpreted as Fisher-Rao centered normal manifolds: 
$(\calN_\mu,g^\Fisher_{\calN_\mu})=(\calP,\frac{1}{2}g^\trace)$.

Since the determinant of a block matrix is
$$
\left|\mattwotwo{A}{B}{C}{D}\right|=\left|A-BD^{-1}C\right|,
$$
 we get with $D=1$: $|f(\mu,\Sigma)| = |\Sigma+\mu\mu^\top-\mu\mu^\top|=|\Sigma|$.

Let $\bar P_1=f(\mu_1,\Sigma_1)$ and 
$\bar P_2=f(\mu_2,\Sigma_2)$.
Checking $D_\KL[p_{\mu_1,\Sigma_1}:p_{\mu_2,\Sigma_2}]=D_\KL[q_{\bar P_1}:q_{\bar P_2}]$ where $q_{\barP}=p_{0,\barP}$ amounts to verify that
$$
\tr(\bar P_2^{-1}\bar P_1)=1+\tr(\Sigma_2^{-1}\Sigma_1+\Delta_\mu^\top\Sigma_2^{-1}\Delta_\mu).
$$
Indeed, using the inverse matrix  
$$f(\mu,\Sigma)^{-1}=
\mattwotwo{\Sigma^{-1}}{-\Sigma^{-1}\mu}{-\mu^\top\Sigma^{-1}}{1+\mu^\top \Sigma^{-1}\mu}
,$$
we have 
\begin{eqnarray*}
\tr(\bar P_2^{-1}\bar P_1)&=&\tr\left(
\mattwotwo{\Sigma^{-1}_2}{-\Sigma^{-1}_2\mu_2}{-\mu_2^\top\Sigma_2^{-1}}{1+\mu_2^\top \Sigma^{-1}_2\mu_2}\ \mattwotwo{\Sigma_1+\mu_1\mu_1^\top}{\mu_1}{\mu_1^\top}{1}
\right),\\
&=& 1+\tr(\Sigma_2^{-1}\Sigma_1+\Delta_\mu^\top\Sigma_2^{-1}\Delta_\mu).
\end{eqnarray*}
Thus even if the dimension of the sample spaces of $p_{\mu,\Sigma}$ and $q_{\barP=f(\mu,\Sigma)}$ differs by one, we get the same KLD by Calvo and Oller's isometric mapping $f$.
\end{proof}

This property holds for the KLD/Jeffreys divergence but not for all $f$-divergences~\cite{IG-2016} $I_f$ in general (e.g., it fails for the Hellinger divergence).

Figure~\ref{Fig:CurvesTissot} shows the various geodesics and curves used to approximate the Fisher-Rao distance  with the Fisher metric shown using Tissot indicatrices.
\begin{figure}
\centering
\begin{tabular}{ccc}
\fbox{\includegraphics[width=0.3\textwidth]{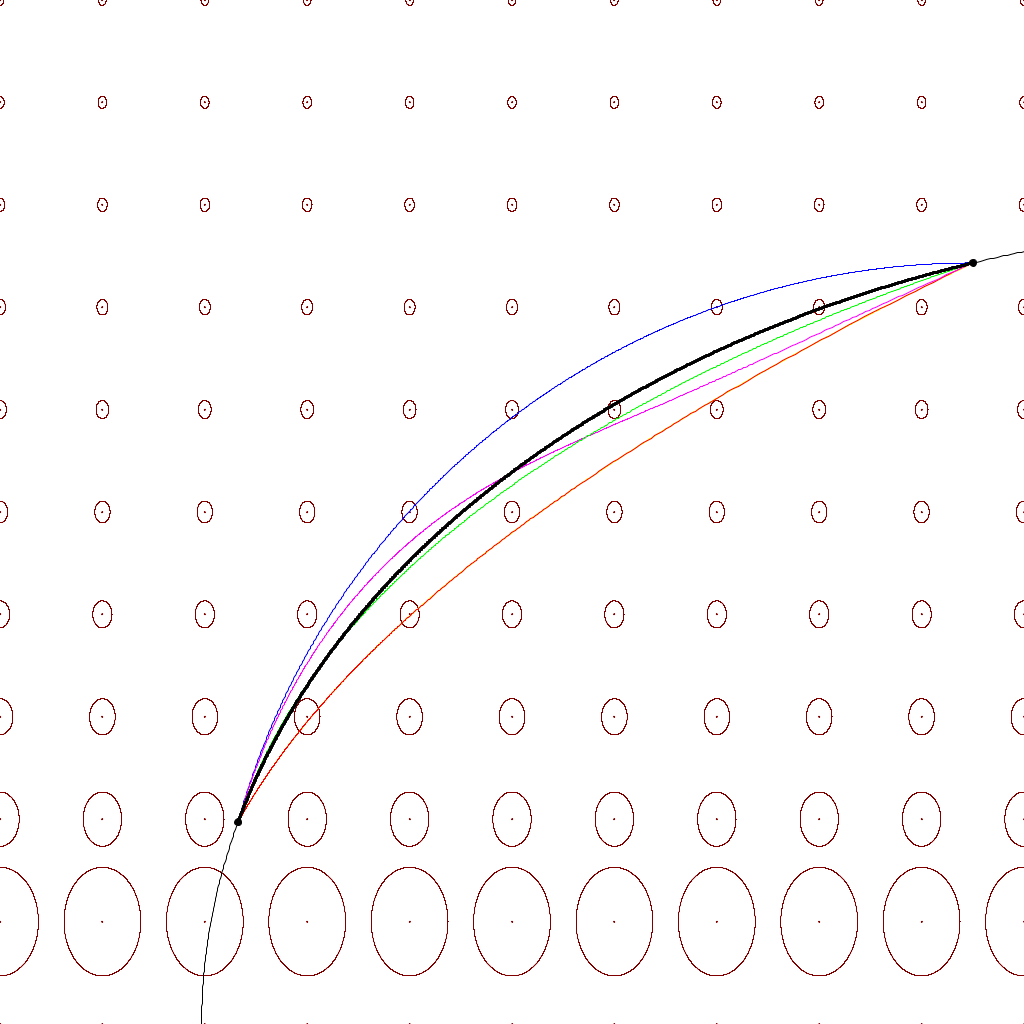}} &
\fbox{\includegraphics[width=0.3\textwidth]{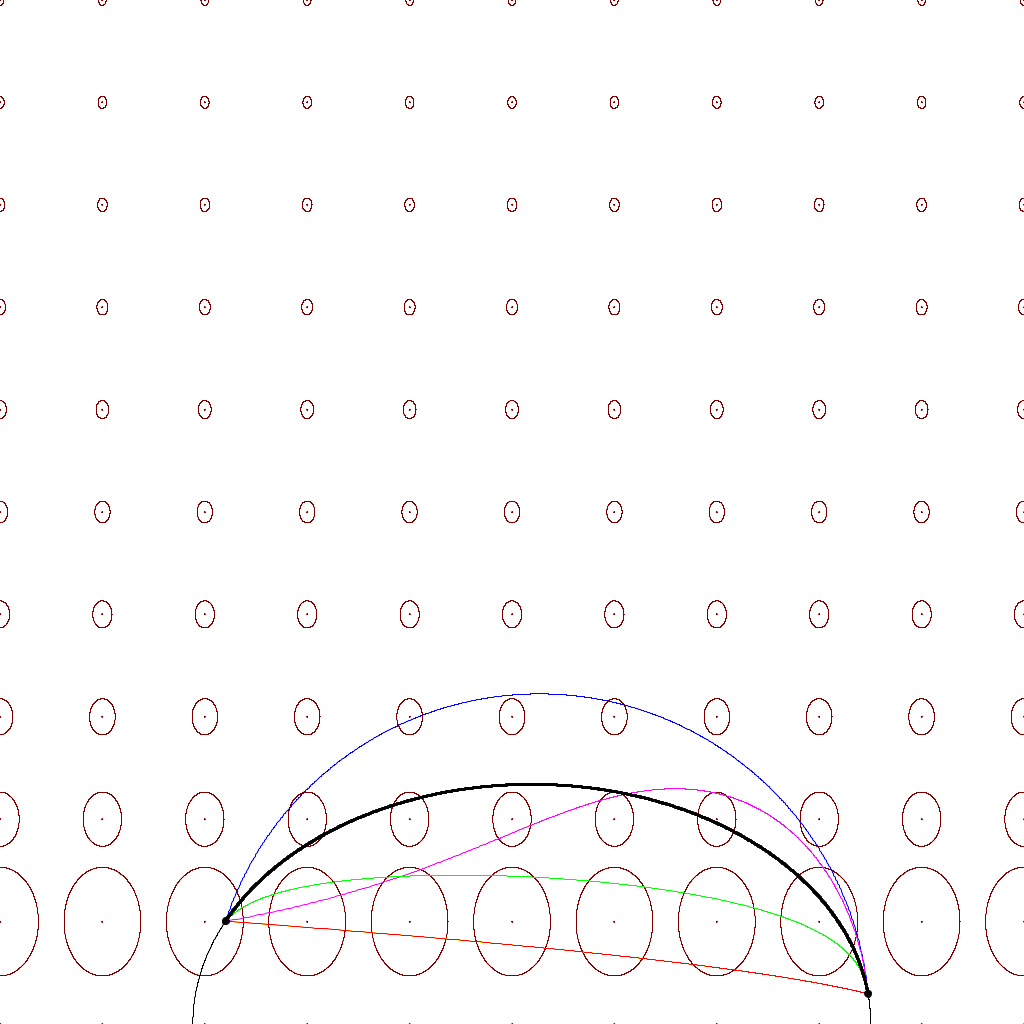}} &
\fbox{\includegraphics[width=0.3\textwidth]{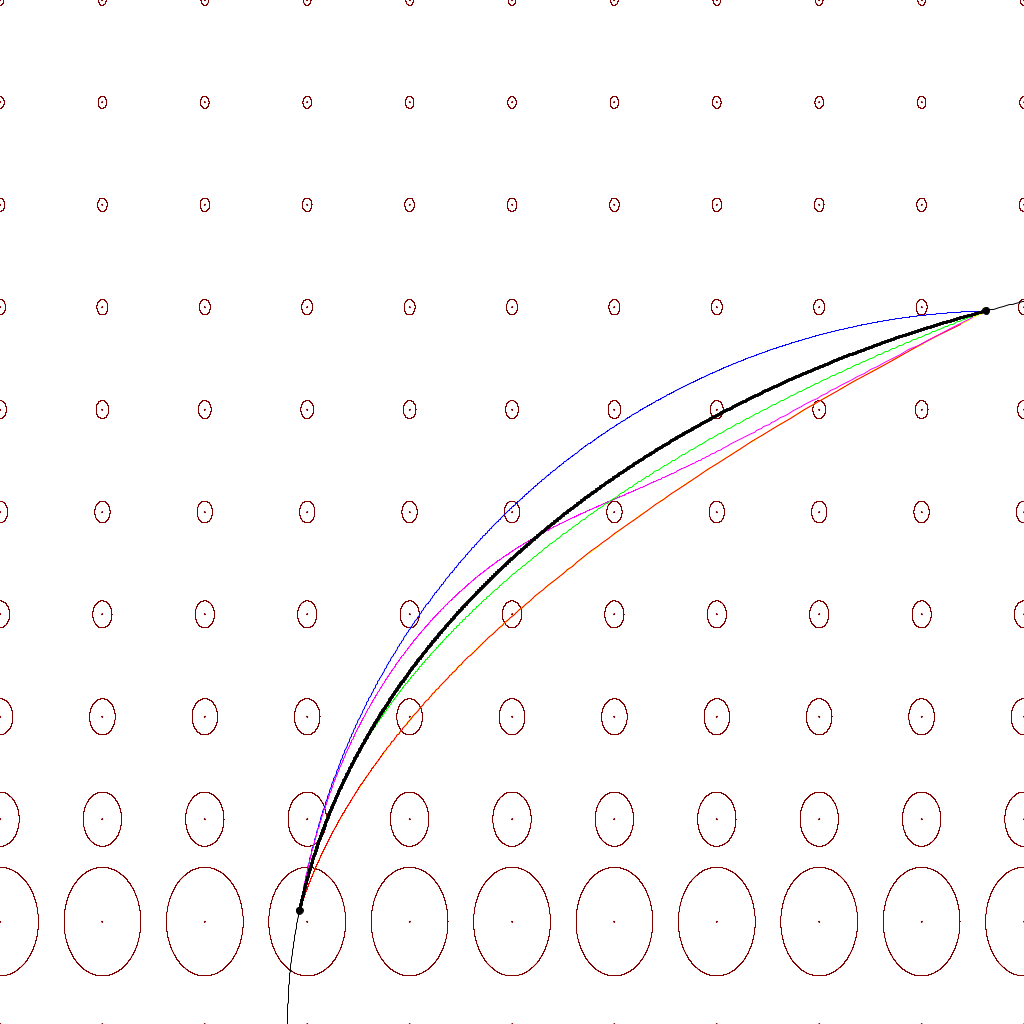}}
\end{tabular}
\caption{Geodesics and curves used to approximate the Fisher-Rao distance with the Fisher metric shown using Tissot's indicatrices:
exponential geodesic (red), mixture geodesic (blue), mid exponential-mixture curve (purple), projected CO curve (green) and target Fisher-Rao geodesic (black). (Visualization in the parameter space of normal distributions.)}\label{Fig:CurvesTissot}
\end{figure}

\begin{figure}
\centering
 
\includegraphics[width=0.4\textwidth]{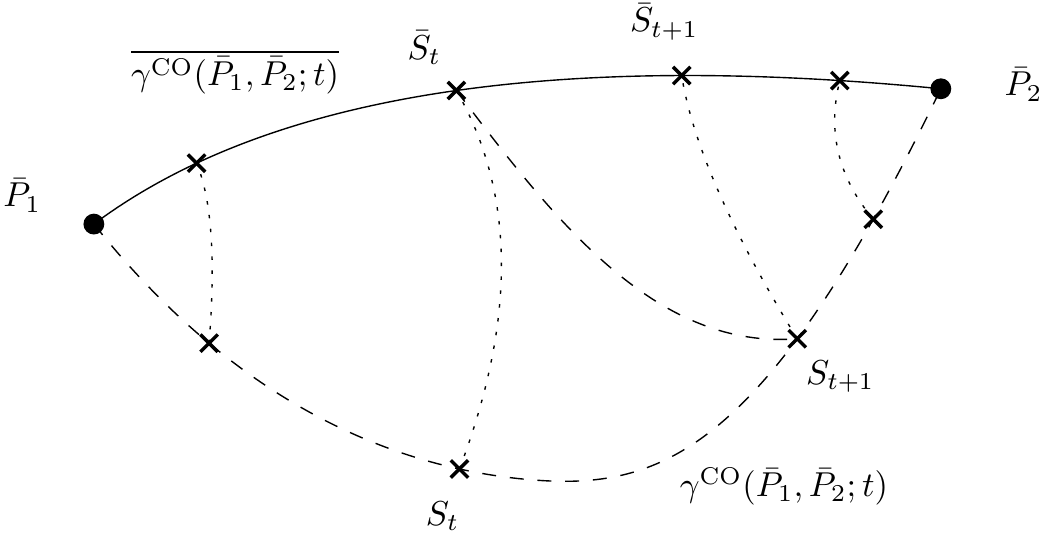}  

\caption{Bounding $\rho_\calN(\bar S_t,\bar S_{t+1})$ using the triangular inequality of $\rho_\calP$ in the SPD cone $\calP(d+1)$.
 \label{fig:approxCO}}
\end{figure}

Note that the introduction of parameter $\beta$ is related to the foliation of the SPD cone $\calP$ by $\{f_{\beta}(\calN) \st \beta>0\}$: $\calP(d+1)=\bbR_{>0}\times f_{\beta}(\calN)$. See Figure~\ref{fig:projection}.
Thus we may define how good the projected C\&O curve is to the Fisher-Rao geodesic by measuring the average distance between points on $\gamma_\calP(\barP_1,\barP_2;t)$ and their projections $\overline{\gamma_\calP(\barP_1,\barP_2;t)}^\perp$ onto $\barN$:
$$
\delta^\CO(P_1,P_2)=\int_0^1 \rho_\calP(\gamma_\calP(\barP_1,\barP_2;t),\overline{\gamma_\calP(\barP_1,\barP_2;t)}^\perp)\, \dt.
$$
In practice, we evaluate this integral at the sampling points $S_t$:
\begin{equation}\label{eq:avgprojdist}
\delta^\CO(P_1,P_2) \approx \delta^\CO_T(P_1,P_2)=\frac{1}{T}\sum_{i=1}^T  \rho_\calP(S_t,\bar{S}_t),
\end{equation}
where $S_t=\gamma_\calP(\barP_1,\barP_2;t)$ and $\bar{S}_t=\gamma_\calP(\barP_1,\barP_2;t)^\perp$.
We checked experimentally (see Section~\ref{sec:exp}) 
that for close by normals $N_1$ and $N_1$, we have $\delta^\CO(\bar{N}_1,\bar{N}_2)$ small, and that when $N_1$ gets further separated from $N_2$, the average projection error $\delta^\CO(\bar{N}_1,\bar{N}_2)$ increases.
Thus $\delta^\CO_T(P_1,P_2)$ is a good measure of the precision of our Fisher-Rao distance approximation.

\begin{Property}
We have $\rho_{\barN}(\bar{S}_t,\bar{S}_{t+1})\leq \rho_\calP(\bar S_t,S_t)+\rho_\calP(S_t,S_{t+1})+\rho_\calP(S_{t+1},\bar S_{t+1})$.
\end{Property}

\begin{proof}
The proof consists in applying twice the triangle inequality of metric distance $\rho_\calP$:
\begin{eqnarray*}
\rho_{\barN}(\bar{S}_t,\bar{S}_{t+1}) &\leq&  \rho_\calP(\bar{S}_t,S_{t+1})+\rho_\calP(S_{t+1},\bar S_{t+1}),\\
&\leq &  \rho_\calP(\bar{S}_t,S_{t}) + \rho_\calP(S_t,S_{t+1}) + \rho_\calP(S_{t+1},\bar S_{t+1}).
\end{eqnarray*}
See Figure~\ref{fig:approxCO}.
\end{proof}

\begin{Property}
We have 
$\rho_\calN(N_1,N_2)\leq \rho_\calN^\CO(N_1,N_2) \leq  \rho_\calN(N_1,N_2) + 2 \delta^\CO_T(\bar P_1,\bar P_2)$.
\end{Property}

\begin{proof}
At infinitesimal scale, we have
$$
\ds_{\calN}(\bar S_t) \leq \ds_{\calP}(S_t) + 2\rho_\calP(S_t,\bar S_t). 
$$

Taking the integral, we get
$$
\rho_\calN(N_1,N_2)\leq \rho_\calP(\bar P_1,\bar P_2) + 2 \delta^\CO_T(\bar P_1,\bar P_2)
$$
Since $\rho_\calP(P_1,P_2)\leq \rho_\calN(N_1,N_2)$, we have
$$
\rho_\calN(N_1,N_2)\leq \rho_\calN^\CO(N_1,N_2) \leq  \rho_\calN(N_1,N_2) + 2 \delta^\CO_T(\bar P_1,\bar P_2).
$$
\end{proof}

\begin{Example}\label{ex1entropy}
Let us consider Example~1 of~\cite{FRMVNReview-2020} (p. 11):
$$
N_1=\left(\vectortwo{-1}{0},\Sigma\right),\quad
N_2=\left(\vectortwo{6}{3},\Sigma\right), \Sigma=\mattwotwo{1.1}{0.9}{0.9}{1.1}.
$$
The Fisher-Rao distance is evaluated numerically in ~\cite{FRMVNReview-2020} as $5.00648$.
We have the lower bound $\rho_\calN^\CO(N_1,N_2)=4.20447$, and the Mahalanobis distance $8.06226$ upper bounds the Fisher-Rao distance (not totally geodesic submanifold $\calN_\Sigma$).
Our projected C\&O curve discretized with $T=1000$ yields an approximation $\tilde\rho_\calN^\CO(N_1,N_2)=5.31667$.
The average projection distance $\rho_\calP(S_t,\bar S_{t})$ is $\delta^\CO_T(N_1,N_2)=0.61791$, and the maximum projected distance is $1.00685$.
We check that
$$
5.00648\approx \rho_\calN(N_1,N_2)\leq \tilde\rho_\calN^\CO(N_1,N_2)\approx 5.31667 \leq \rho_\calN(N_1,N_2) + 2 \delta^\CO_T(\bar P_1,\bar P_2)\approx 5.44028.
$$
The Killing distance obtained for $\kappa_\Killing=2$ is $\rho_\Killing(N_1,N_2)\approx 6.82028$.
Notice that geodesic shooting is time consuming compared to our approximation technique.
\end{Example}

\subsection{Some experiments}\label{sec:exp}

The KLD $D_\KL$ and Jeffreys divergence $D_J$, the Fisher-Rao distance $\rho_\calN$ and the Calvo \& Oller distance $\rho_\CO$ are all invariant under the congruence action of the affine group $\Aff(d)=\bbR^d\rtimes\GL(d)$ with group operation 
$$
(a_1,A_1)(a_2,A_2)=(a_1+A_1a_2,A_1A_2).
$$
Let $(A,a)\in\Aff(d)$, and define the action on the normal space $\calN$ as follows:
$$
(A,a).N(\mu,\Sigma)=N(A^\top\mu+a,A\Sigma A^\top).
$$
Then we have $\rho_{\calN}((A,a).N_1,(A,a).N_2)=\rho_{\calN}(N_1,N_2)$, $\rho_{\CO}((A,a).N_1,(A,a).N_2)=\rho_{\CO}(N_1,N_2)$ 
and $D_\KL[(A,a).N_1:(A,a).N_2]=D_\KL[N_1:N_2]$.
This invariance extends to our approximations $\tilde\rho^c_\calN$ (see Eq.~\ref{eq:approx}).
 
Since we have 
$$
\tilde\rho_{\calN}^c(N_1,N_2)\approx \rho_{\calN}(N_1,N_2)\geq \rho_{\CO}(N_1,N_2),
$$
the ratio $\kappa_c=\frac{\tilde\rho_{\calN}^c}{\rho_{\CO}}\geq \kappa=\frac{\tilde\rho_{\calN}^c}{\rho_{\calN}}$ gives an upper bound on the approximation factor of $\tilde\rho_{\calN}^c$ compared to the true Fisher-Rao distance $\rho_{\calN}$:
$$
\kappa_c \rho_{\calN}(N_1,N_2)\geq \kappa \rho_{\calN}(N_1,N_2)\geq \tilde\rho_{\calN}^c(N_1,N_2)\approx \rho_{\calN}(N_1,N_2) \geq \rho_{\CO}(N_1,N_2).
$$

Let us now report some numerical experiments of our approximated Fisher-Rao distances $\tilde\rho_{\calN}^x$ with $x\in\{l,m,e,\mathrm{em},\CO\}$. When normal distributions are in 1D we can exactly plot the Fisher-Rao geodesics are locate the other geodesics/curves with respect to the Fisher-Rao geodesics (Figure~\ref{fig:emgeodesic1d}).
Although that dissimilarity $\tilde\rho_{\calN}$ is positive-definite, it does not satisfy the triangular inequality of metric distances (e.g., Riemannian distances $\rho_{\calN}$ and $\rho_\CO$).

\begin{figure}
\centering
\begin{tabular}{cc}
\includegraphics[width=0.3\textwidth]{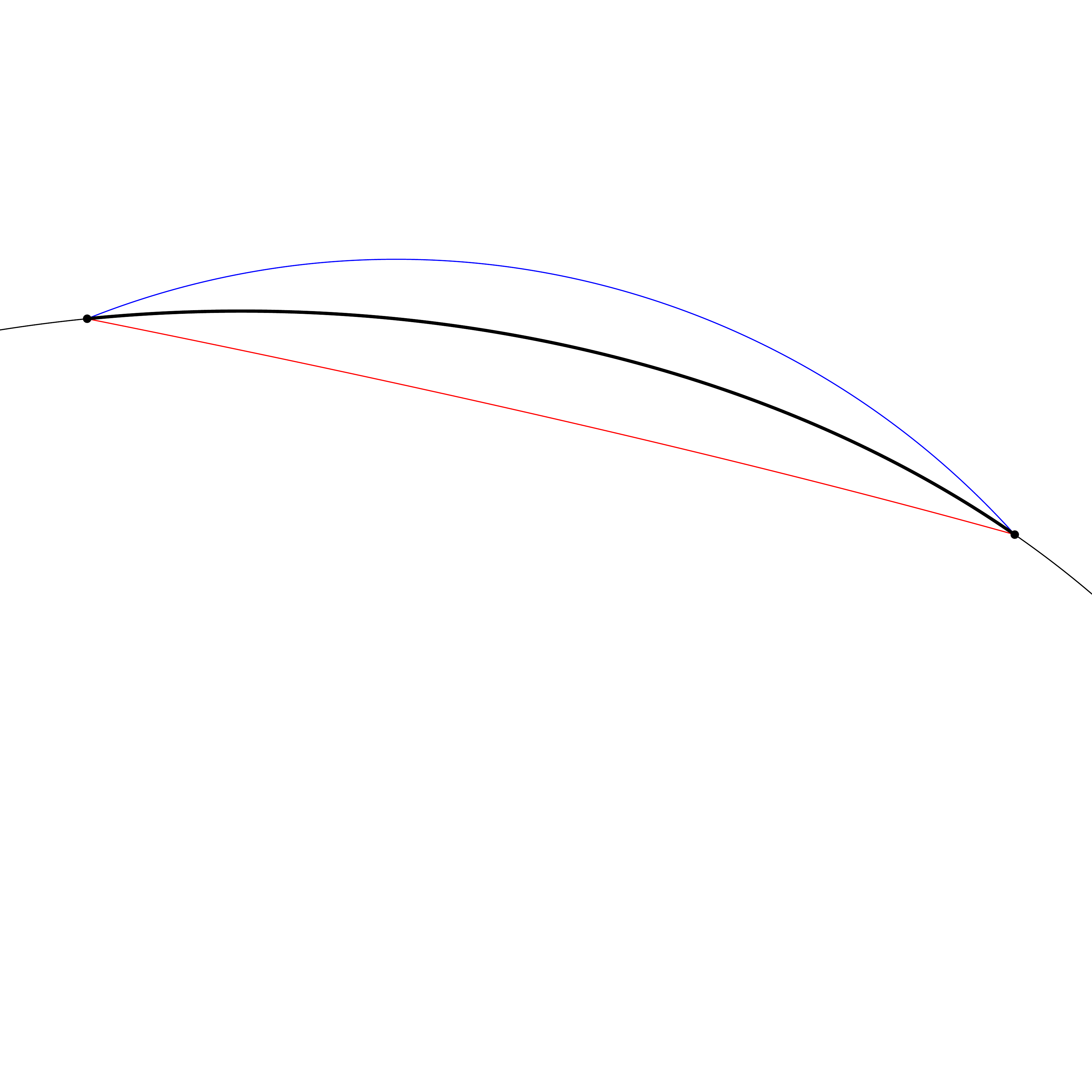} &
\includegraphics[width=0.3\textwidth]{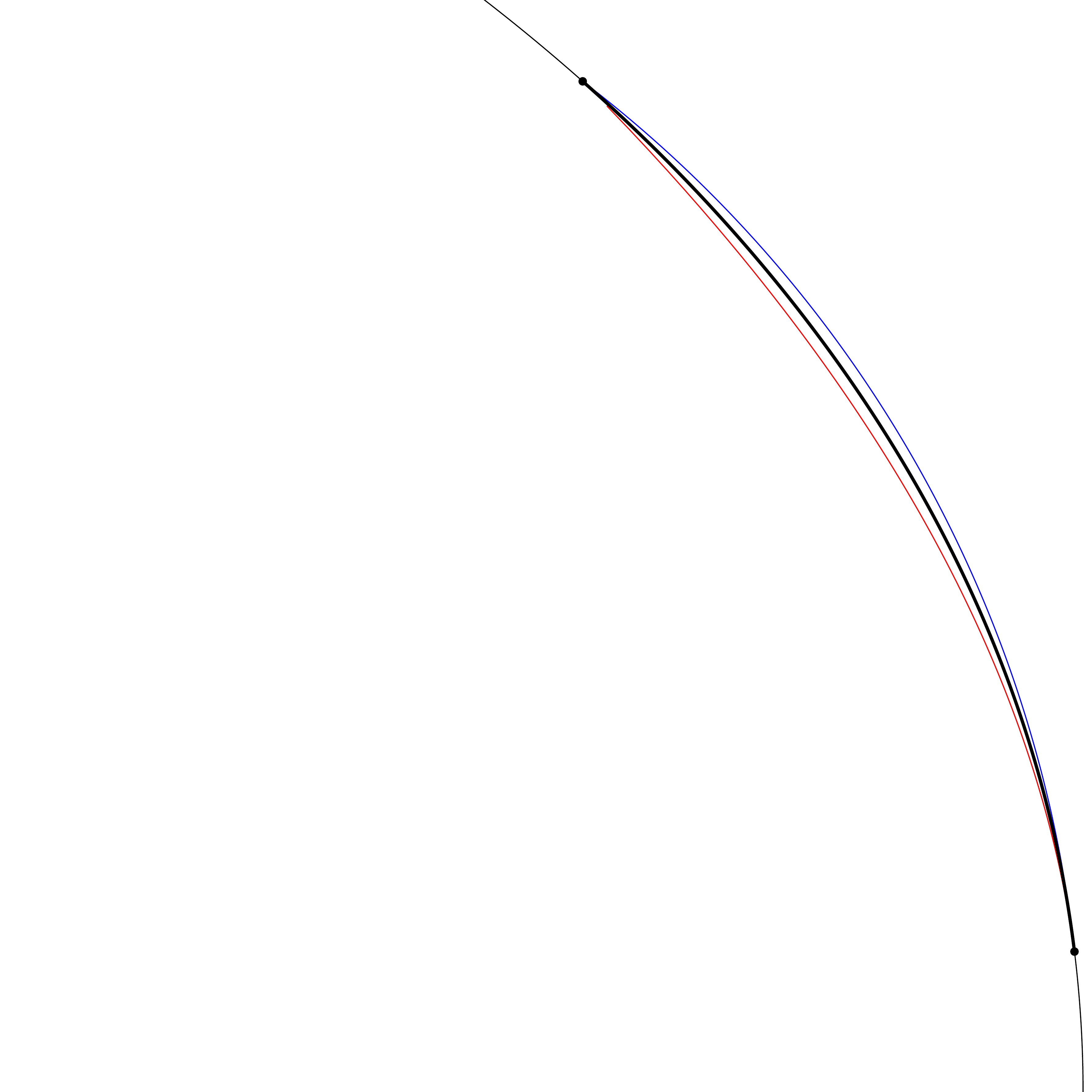}
\end{tabular}
\caption{The full Fisher-Rao geodesic is shown in grey with the geodesic arc linking two univariate normal distributions shown in black. The exponential geodesics are shown in red and are below the Fisher-Rao geodesic. The mixture geodesics are shown in blue and are above the Fisher-Rao geodesic.}\label{fig:emgeodesic1d}
\end{figure}

First, we draw multivariate normals by sampling means $\mu\sim\mathrm{Unif}(0,1)$ and sample covariance matrices $\Sigma$ as follows: We draw a lower triangular matrix $L$ with entries $L_{ij}$ iid sampled from $\mathrm{Unif}(0,1)$, and take $\Sigma=LL^\top$.
We use $T=1000$ samples on curves and repeat the experiment $1000$ times
to gather average statistics on $\kappa_c$'s of curves. Results are summarized in Table~\ref{tab1}.

\begin{table}
\caption{First set of experiments demonstrates the advantage of the $c_\CO(t)$ curve.\label{tab1}}
\centering
\begin{tabular}{l|lllll}
$d$ & $\kappa_\CO$ & $\kappa_l$ & $\kappa_e$ & $\kappa_m$ & $\kappa_{em}$ \\ \hline
1	& {\bf 1.0025}	& 	1.0414		& 1.1521		& 1.0236 & 1.0154\cr
2		& {\bf 1.0167}	& 	1.0841	& 	1.1923		& 1.0631 & 1.0416\cr
3		& {\bf 1.0182}		& 1.8997	& 	2.6072	& 	1.9965 & 1.07988\cr
4		& {\bf 1.0207} 	& 	2.0793		& 1.8080	& 	2.1687 & 1.1873\cr
5		& {\bf 1.0324}	& 	4.1207	& 	12.3804	& 	5.6170 & 4.2349\cr
\end{tabular}

\end{table}

For that scenario that the C\&O curve (either $\bar c_\CO\in\barN$ or $c_\CO\in\calN$) performs best compared to the linear interpolation curves with respect to source parameter ($l$), mixture geodesic ($m$),   exponential geodesic ($e$), or exponential-mixture mid curve ($\mathrm{em}$).
Let us point out that we sample $\gamma_\calP(\barP_1,\barP_2;\frac{i}{T})$ for $i\in\{0,\ldots, T\}$.

Strapasson, Porto and Costa~\cite{strapasson2015bounds}  (SPC)reported the following upper bound on the Fisher-Rao distance between multivariate normals:
\begin{equation}\label{prop:USPC}
\rho_\CO(N_1,N_2) \leq \rho_\calN(N_1,N_2) \leq U_\SPC(N_1,N_2)=\sqrt{2\sum_{i=1}^d \log^2
\left( \frac{\sqrt{(1+D_{ii})^2+\mu_i^2} +\sqrt{(1-D_{ii})^2+\mu_i^2} }{\sqrt{(1+D_{ii})^2+\mu_i^2} -\sqrt{(1-D_{ii})^2+\mu_i^2}}\right)},
\end{equation}
where $\Sigma=\Sigma_1^{-\frac{1}{2}}\Sigma_2\Sigma_1^{-\frac{1}{2}}$, $\Sigma=\Omega D\Omega^\top$ is the eigen decomposition, and $\mu=\Omega^\top \Sigma_1^{-\frac{1}{2}}(\mu_2-\mu_1)$. 
This upper bound performs better when the normals are well-separated and worse then the $\sqrt{D_J}$-upper bound when the normals are close to each others.

Let us compare $\rho_\CO(N_1,N_2)$ with $\rho_\calN(N_1,N_2) \approx \tilde\rho^{c_\CO}(N_1,N_2)$ and the upper bound $U(N_1,N_2)$ by averaging over $1000$ trials with $N_1$ and $N_2$ chosen randomly as before and $T=1000$. 
We have $\rho_\CO(N_1,N_2)\leq \rho_\calN(N_1,N_2) \approx \tilde\rho^{c_\CO}(N_1,N_2) \leq U(N_1,N_2)$.
Table~\ref{tab:comparebounds} shows that our Fisher-Rao approximation is close to the lower bound (and hence to the underlying true Fisher-Rao distance) and that the upper bound is about twice the lower bound for that particular scenario.
 
 \begin{table}
\caption{Comparing our Fisher-Rao approximation with the Calvo \& Oller lower bound and the Strapasson et al. upper bound.\label{tab:comparebounds}}
\centering
\begin{tabular}{l|lll}
$d$ & $\rho_\CO(N_1,N_2)$ & $\tilde\rho^{c_\CO}(N_1,N_2)$ & $U(N_1,N_2)$ \\ \hline
1 & 1.7563 & 1.8020 & 3.1654 \cr
2 & 3.2213 & 3.3194 & 6.012 \cr
3 & 4.6022 & 4.7642 & 8.7204 \cr
4 & 5.9517 & 6.1927 & 11.3990 \cr
5 & 7.156 & 7.3866 & 13.8774 \cr
\end{tabular}
\end{table}


\begin{figure}
\centering

\includegraphics[width=0.5\textwidth]{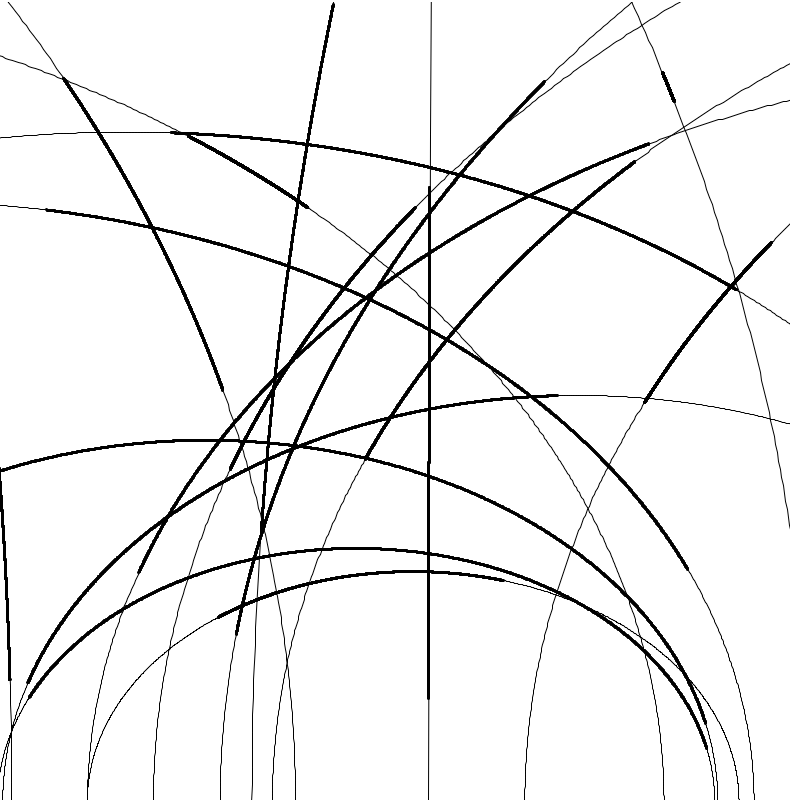}  
\caption{Visualizing some Fisher-Rao geodesics of univariate normal distributions on the Poincar\'e upper plane (semi-circles with origin on the $x$-axis and stretched by $\sqrt{2}$ on the $x$-axis). Full geodesics plotted with thin style and geodesic arcs plotted with thick style.  
 \label{fig:vizgeo1d}}
\end{figure}

The Fisher-Rao geodesics $\gamma_\calN^{\mathrm{FR}}(N_1,N_2)$ on the Fisher-Rao univariate normal manifolds are either vertical line segments when $\mu_1=\mu_2$, or semi-circle with origin on the $x$-axis and $x$-axis stretched by $\sqrt{2}$~\cite{verdoolaege2015new}:

$$
\gamma_\calN^{\mathrm{FR}}(\mu_1,\sigma_1;\mu_2,\sigma_2)=\left\{
\begin{array}{ll}
(\mu,(1-t)\sigma_1+t\sigma_2), & \mu_1=\mu_2=\mu\cr
(\sqrt{2}(c+r\cos t,r\sin t), t\in [\min\{\theta_1,\theta_2\},\max\{\theta_1,\theta_2\}], & \mu_1\not=\mu_2,
\end{array}
\right.,
$$
where
$$
c=\frac{\frac{1}{2}(\mu_2^2-\mu_1^2)+\sigma_2^2-\sigma_1^2}{\sqrt{2}(\mu_1-\mu_2)},\quad
r=\sqrt{\left(\frac{\mu_i}{\sqrt{2}}-c\right)^2+\sigma_i^2}, i\in\{1,2\},
$$
and
$$
\theta_i=\arctan\left(\frac{\sigma_i}{\frac{\mu_i}{\sqrt{2}}-c}\right), i\in\{1,2\},
$$
provided that $\theta_i\geq 0$ for $i\in\{1,2\}$ (otherwise, we let $\theta_i\leftarrow\theta_i+\pi$).
Figure~\ref{fig:vizgeo1d} displays some geodesics on the Fisher-Rao univariate normal manifold.
Figure~\ref{fig:viz1d} displays the considered geodesics and curves in  the stretched Poincar\'e upper plane of univariate normal distributions ($x$-axis is stretched by $\sqrt{2}$) (in 1D for illustration purpose).

\begin{figure}
\centering
\begin{tabular}{ccc}
\includegraphics[width=0.3\textwidth]{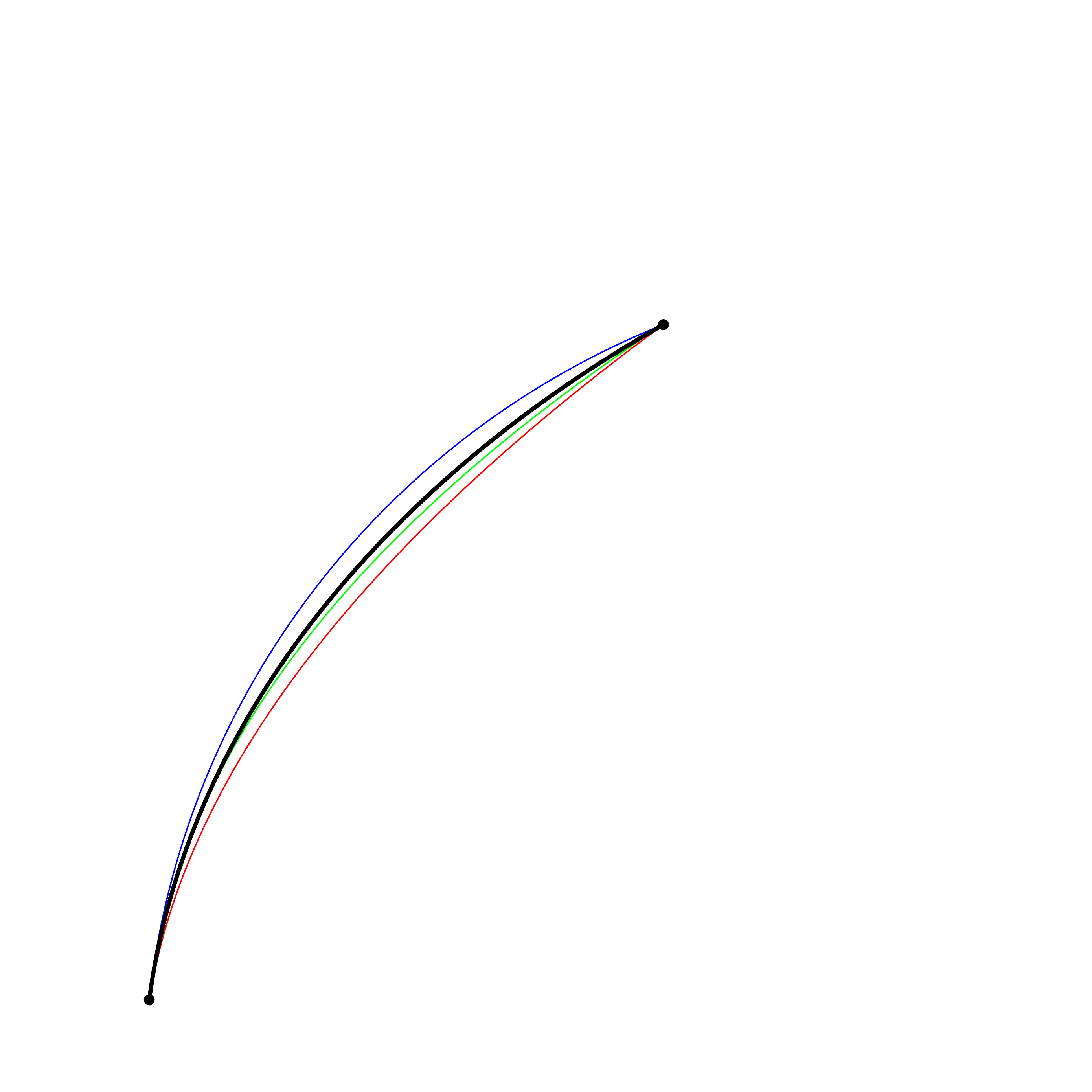} &
\includegraphics[width=0.3\textwidth]{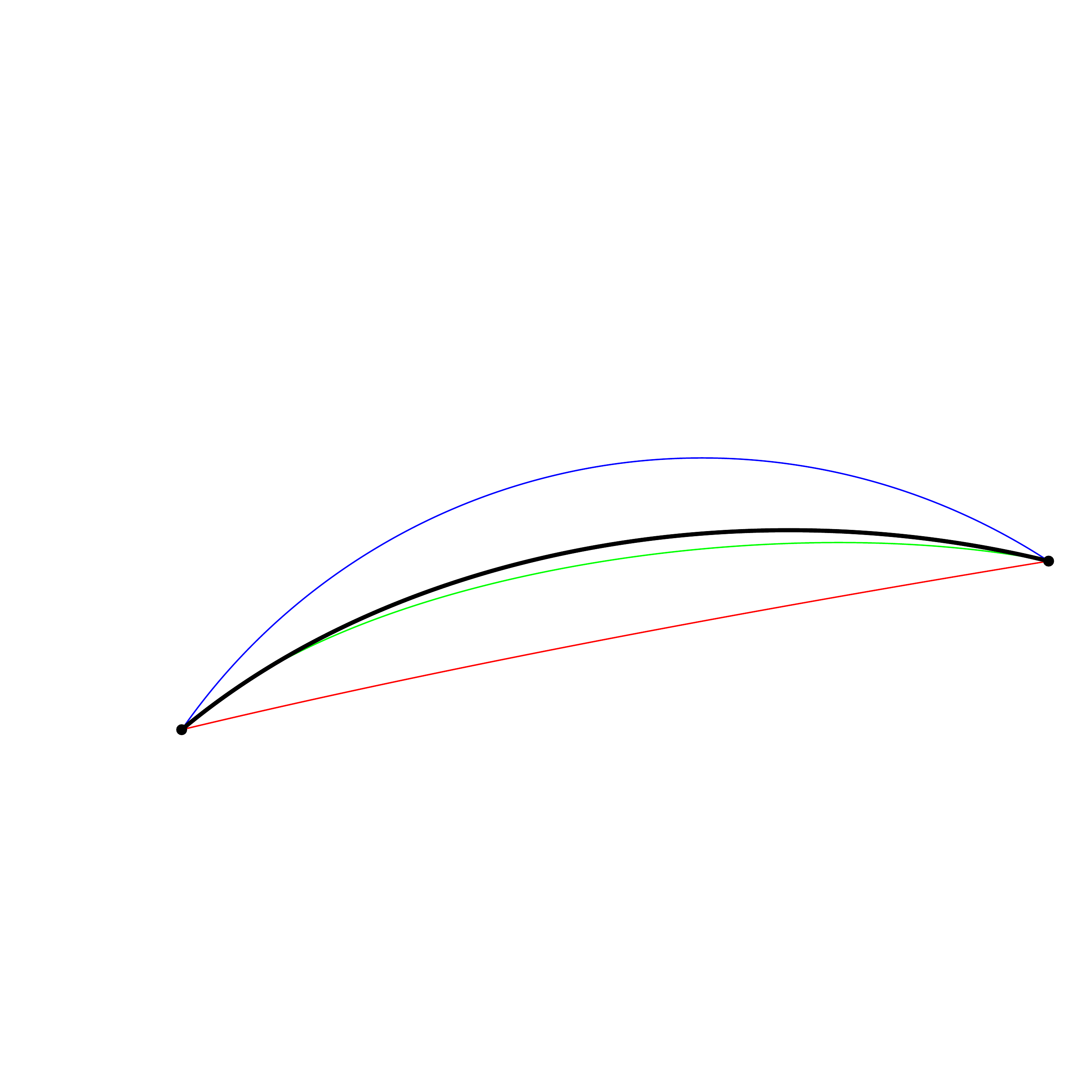}&
\includegraphics[width=0.3\textwidth]{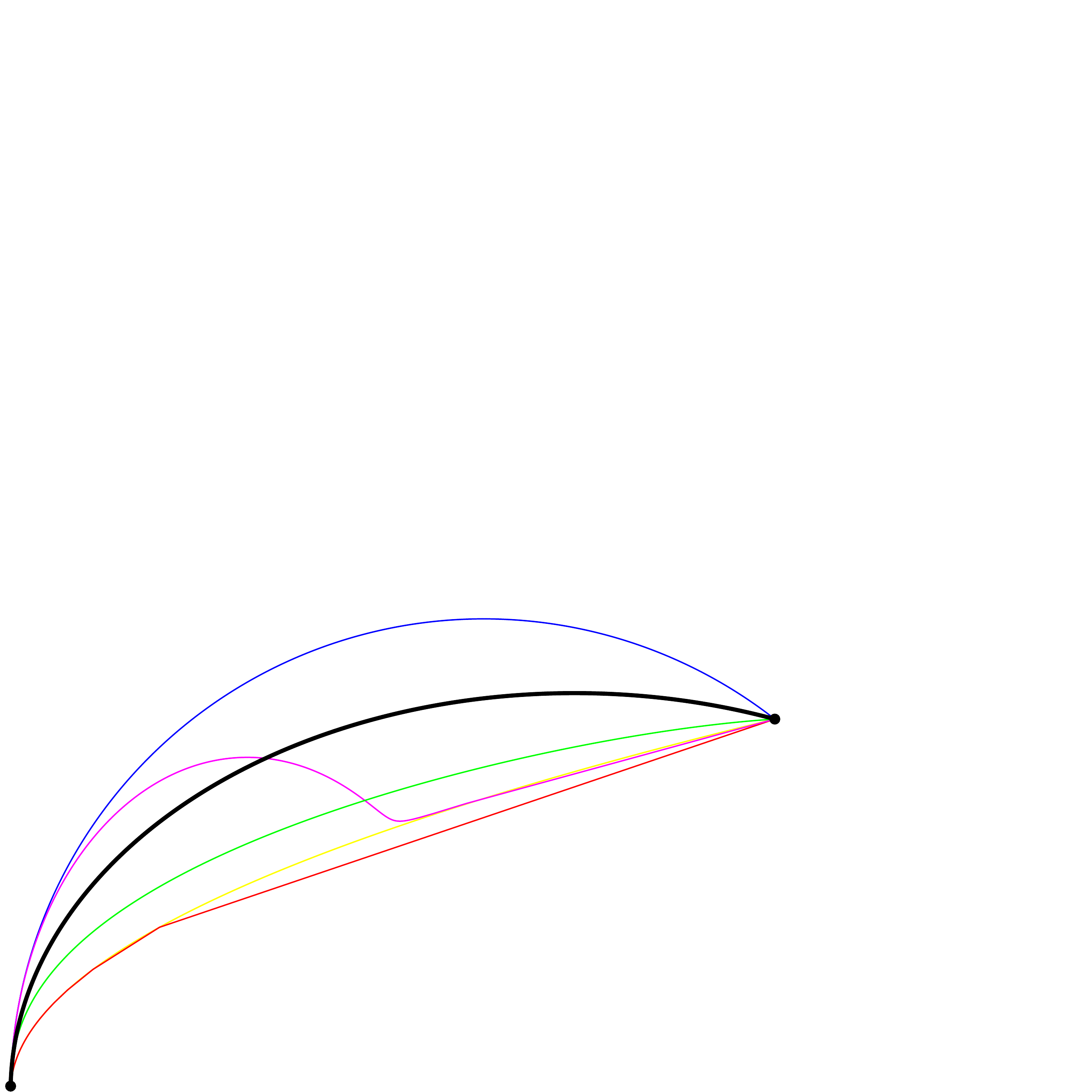} 
\end{tabular}
\caption{Visualizing the geodesics and curves in the Poincar\'e upper plane with $x$-axis stretched by $\sqrt{2}$: 
(a) and (b): Fisher-Rao geodesic (black), our projected Calvo \& Oller curve (green), the mixture geodesic (blue), and the exponential geodesic (red). (c): interpolation in ordinary parameterization $\lambda$ (yellow), mid mixture-exponential curve (purple). The range is $[-1,1]\times (0,2]$.
 \label{fig:viz1d}}
\end{figure}

Second, since the distances are invariant under the action of the affine group, we can set wlog. $N_1=(0,I)$ (standard normal distribution) and let $N_2=\diag(u_1,\ldots, u_d)$ where $u_i\sim\mathrm{Unif}(0,a)$. As normals $N_1$ and $N_2$ separate to each others, we notice experimentally that the performance of the $c_\CO$ curve degrades in the second experiment with $a=5$  (see Table~\ref{tab2}):
Indeed, the mixture geodesic works experimentally better than the C\&O curve when $d\geq 11$.

\begin{table}
\caption{Second set of experiments shows limitations of the $c_\CO(t)$ curve.}\label{tab2}
\centering
\begin{tabular}{l|llll}
$d$ & $\kappa_\CO$ & $\kappa_l$ & $\kappa_e$ & $\kappa_m$ \\ \hline
1		&   {\bf 1.0569}	& 	1.1405	& 1.139		& 1.0734\cr
5	 	&  {\bf 1.1599}		& 1.4696		& 1.5201		& 1.1819\cr
10	& 	{\bf 1.2180}	& 	1.6963	& 	1.7887		& 1.2184\cr
11		& 1.2260	& 	1.7333	& 	1.8285	& 	{\bf 1.2235}\cr
12		& 1.2301	& 1.7568		& 1.8539		& {\bf 1.2282}\cr
15		& 1.2484	& 	1.8403		& 1.9557		& {\bf 1.2367} \cr
20		& 1.2707		& 1.9519	& 	2.0851	& 	{\bf 1.2466}
\end{tabular}

\end{table}

Figure~\ref{fig:viz2D} and Figure~\ref{fig:viz2Dbis} display the various curves considered for approximating the Fisher-Rao distance between bivariate normal distributions: For a curve $c(t)$, we visualize its corresponding bivariate normal distributions $(\mu_{c(t)},\Sigma_{c(t)})$ at several increment steps $t\in [0,1]$ by plotting the ellipsoid 
$$
E_{c(t)}=\mu_{c(t)}+\left\{L^\top x, x=(\cos\theta,\sin\theta), \theta\in [0,2\pi)\right\},
$$
where $\Sigma_{c(t)}=L_{c(t)}L^\top_{c(t)}$.

\begin{figure}
\centering
\begin{tabular}{ccc}
\includegraphics[width=0.3\textwidth]{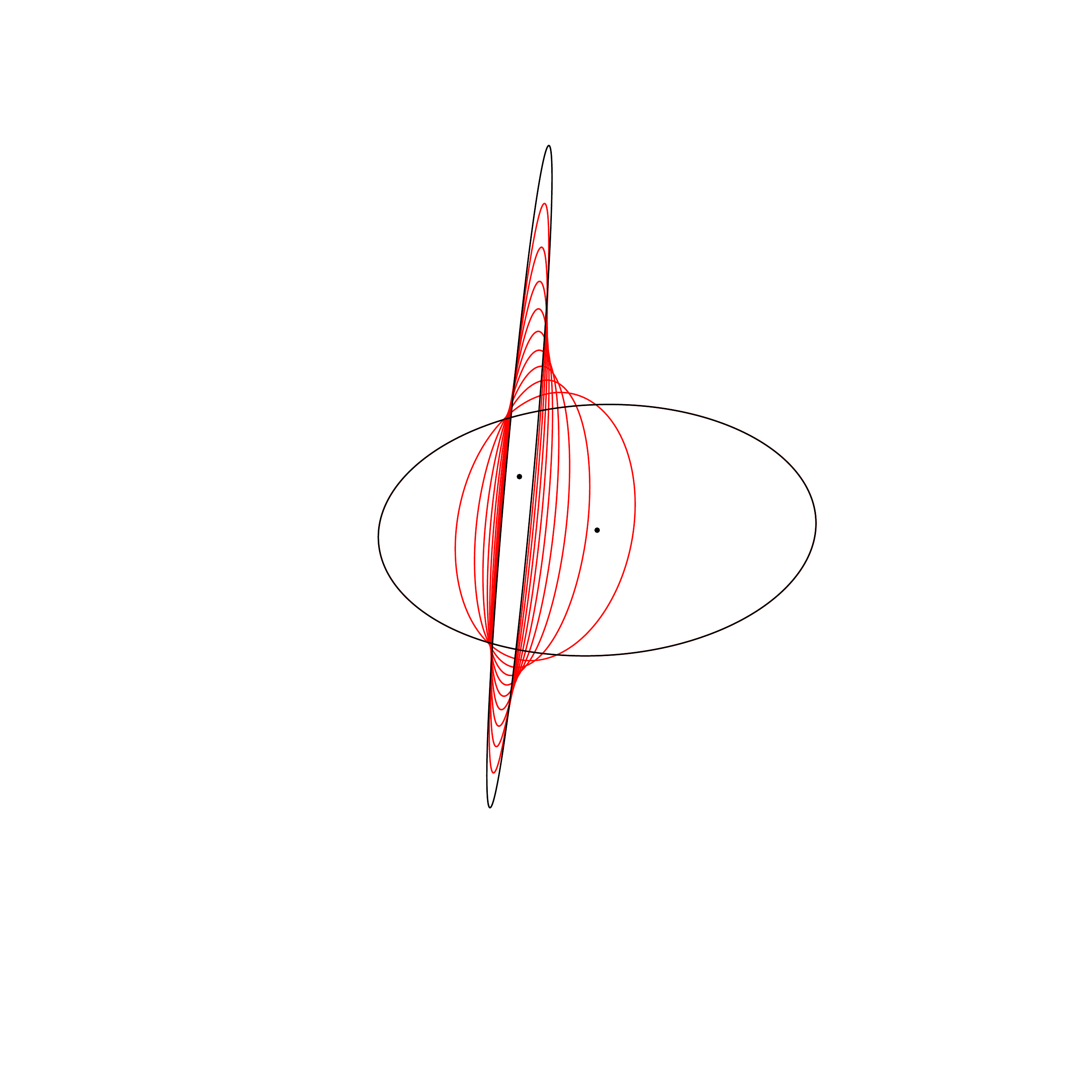} &
\includegraphics[width=0.3\textwidth]{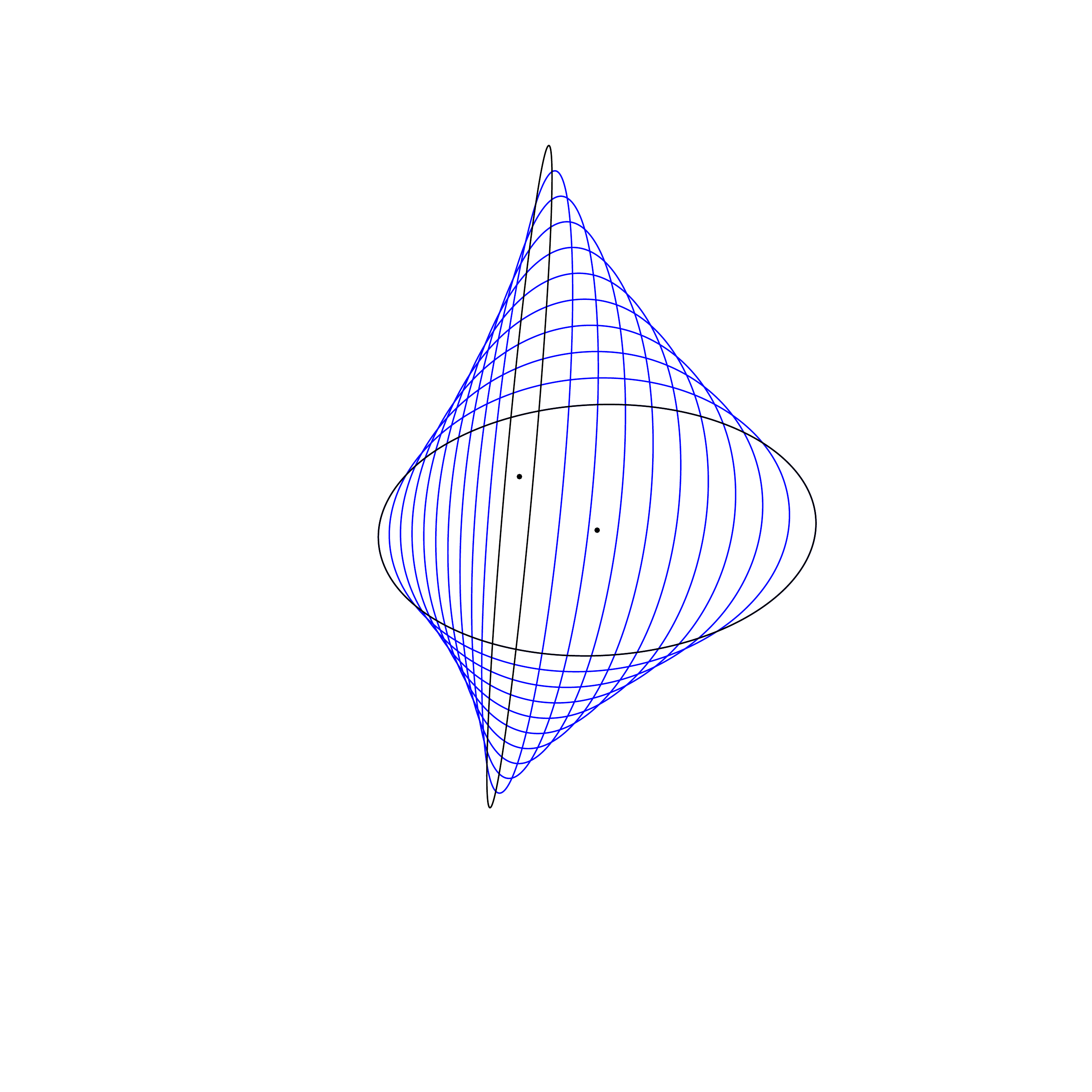}&
\includegraphics[width=0.3\textwidth]{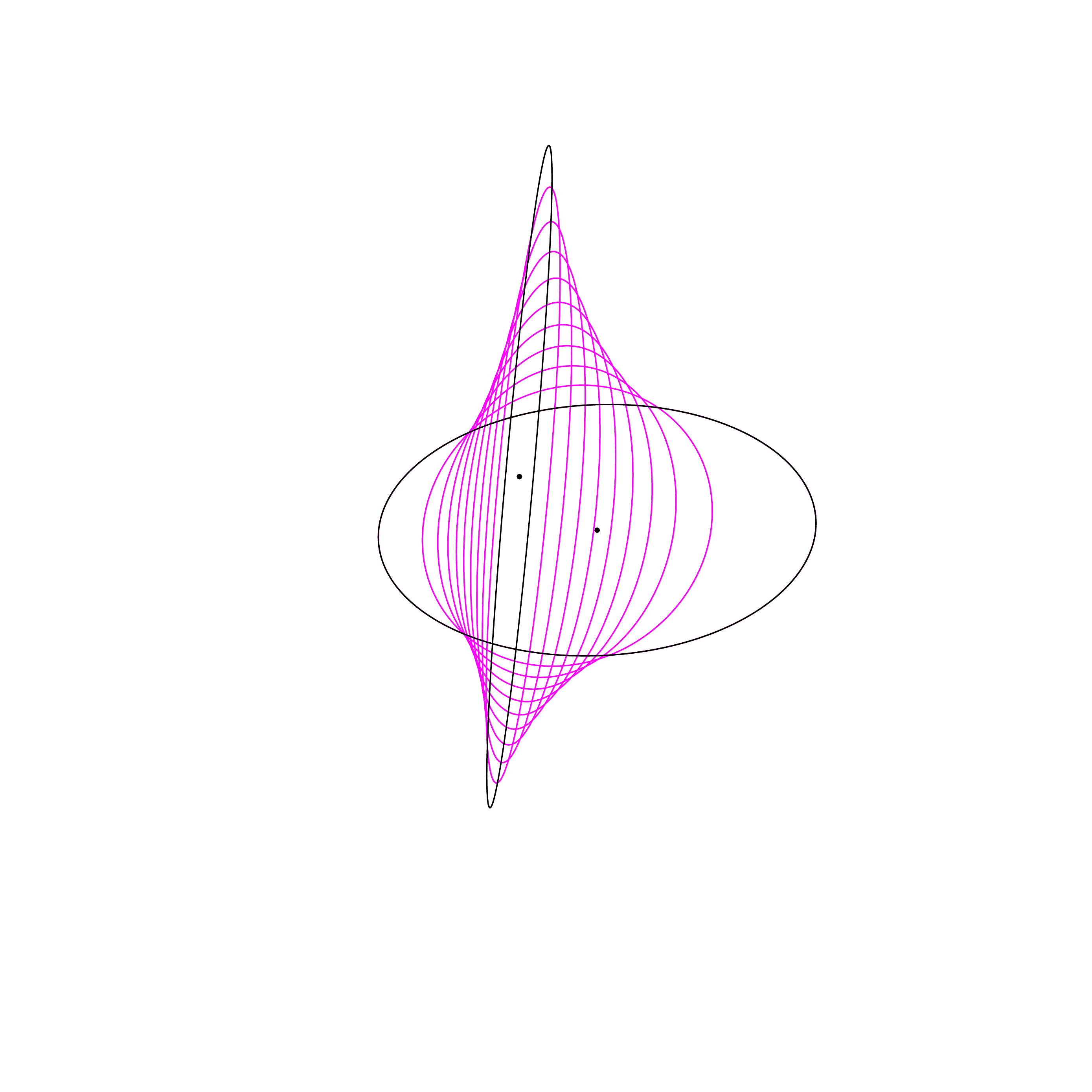} \\
(a) & (b) & (c)\cr
\includegraphics[width=0.3\textwidth]{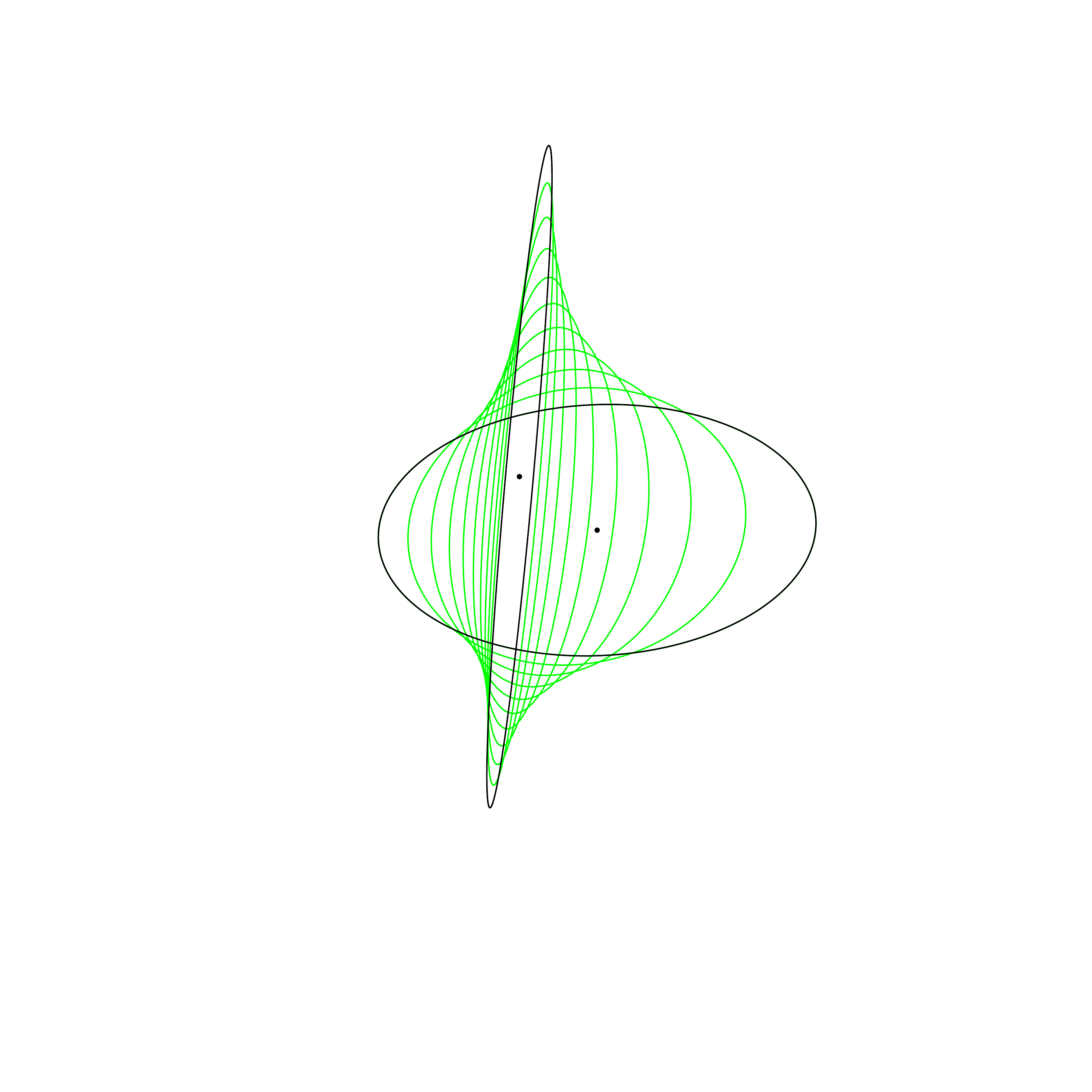} &
\includegraphics[width=0.3\textwidth]{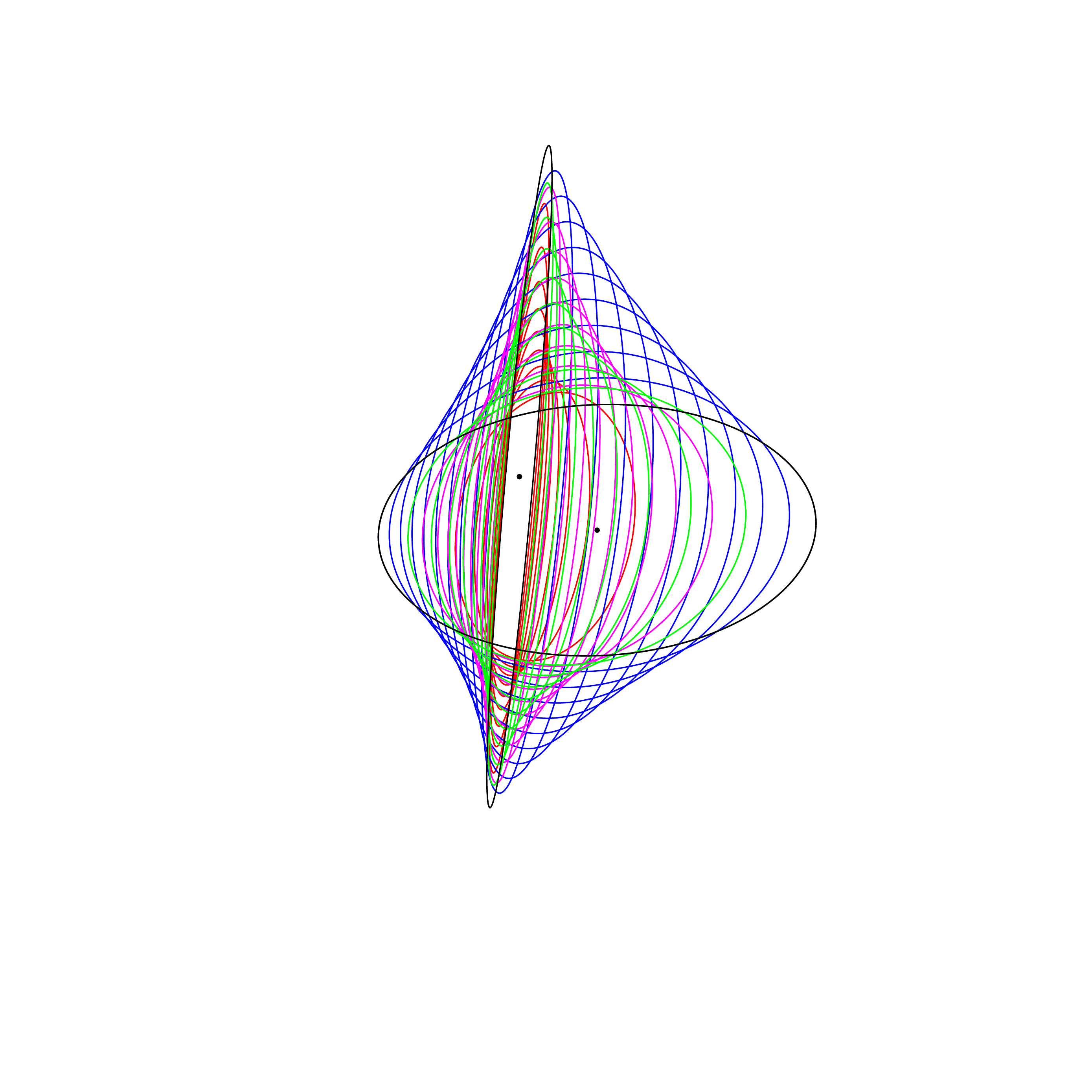}& \\
(d) &(e)
\end{tabular}
\caption{Visualizing at discrete positions (10 increment steps between $0$ and $1$) some curves used to approximate the Fisher-Rao distance between two bivariate normal distributions:
(a) exponential geodesic $c^e=\gamma_\calN^e$ (red),
(b) mixture geodesic $c^m=\gamma_\calN^m$ (blue),
(c) mid mixture-exponential curve $c^{\mathrm{em}}$ (purple),
(d) projected Calvo \& Oller curve $c^{\CO}$ (green), and
(e) All superposed curves at once. 
Visualization in the 2D sample space of bivariate normal distributions because bivariate normals are modeled as a $m=5$ dimensional point on the Fisher-Rao manifold $\calM$.
 \label{fig:viz2D}}
\end{figure}

\begin{figure}
\centering
\begin{tabular}{ccc}
\includegraphics[width=0.3\textwidth]{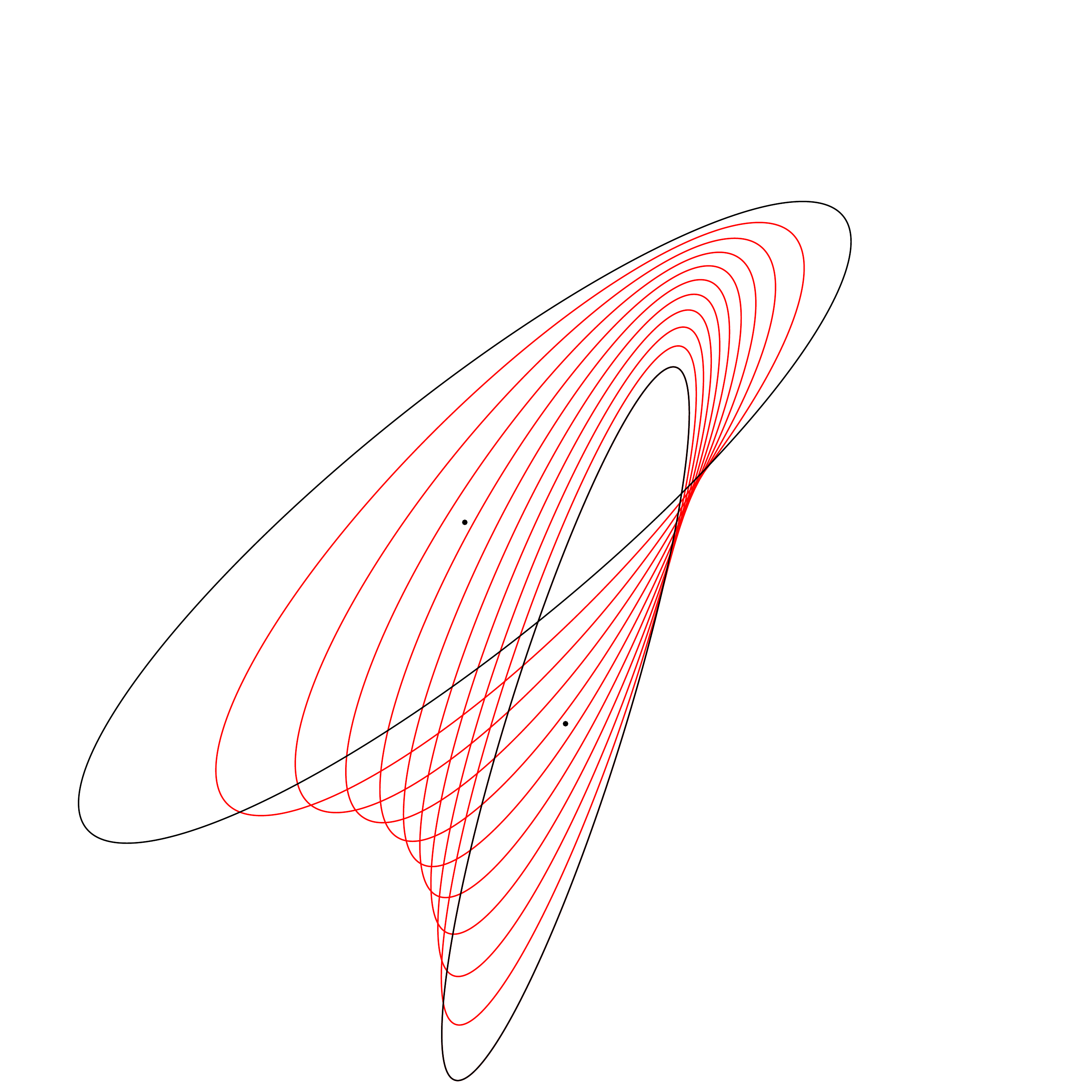} &
\includegraphics[width=0.3\textwidth]{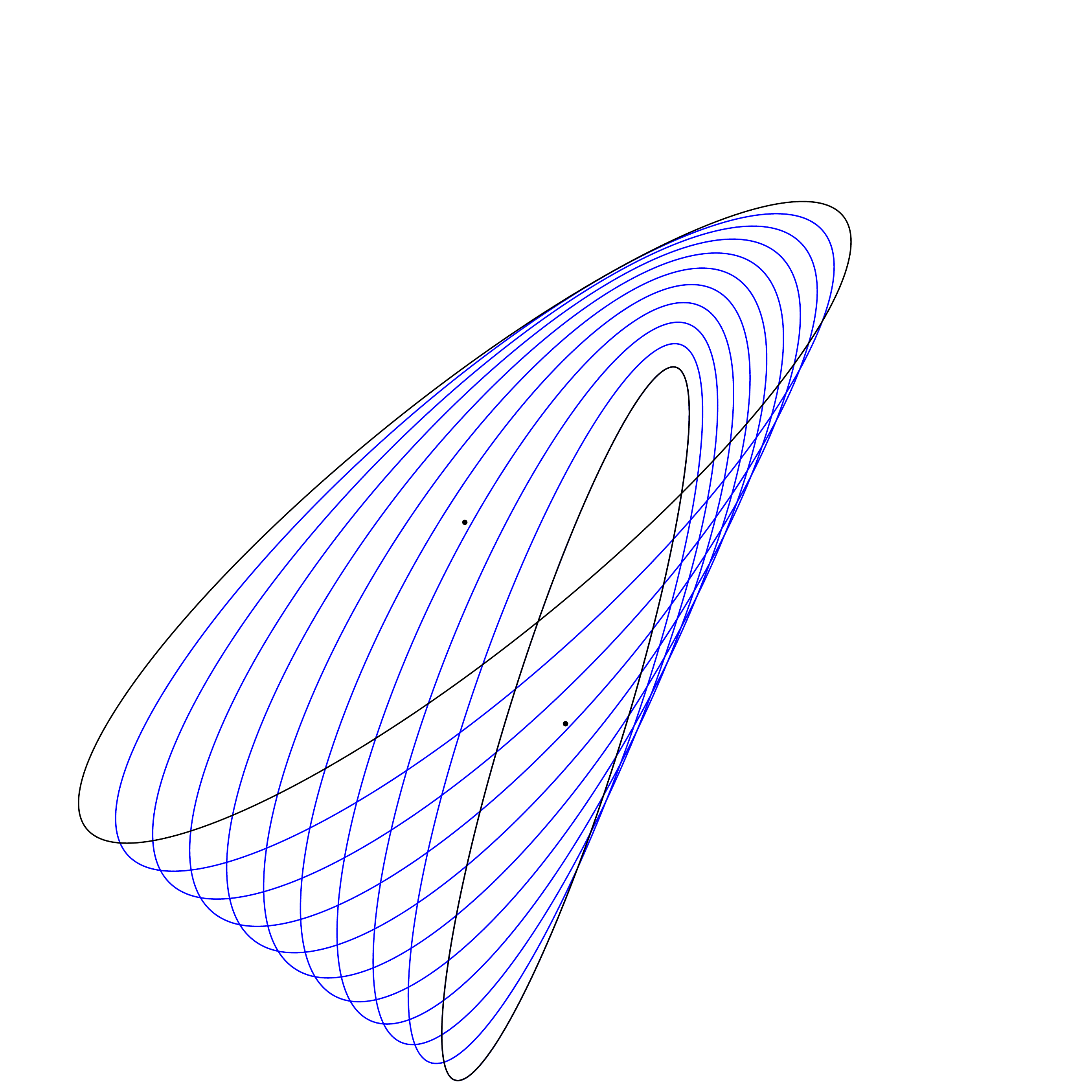}&
\includegraphics[width=0.3\textwidth]{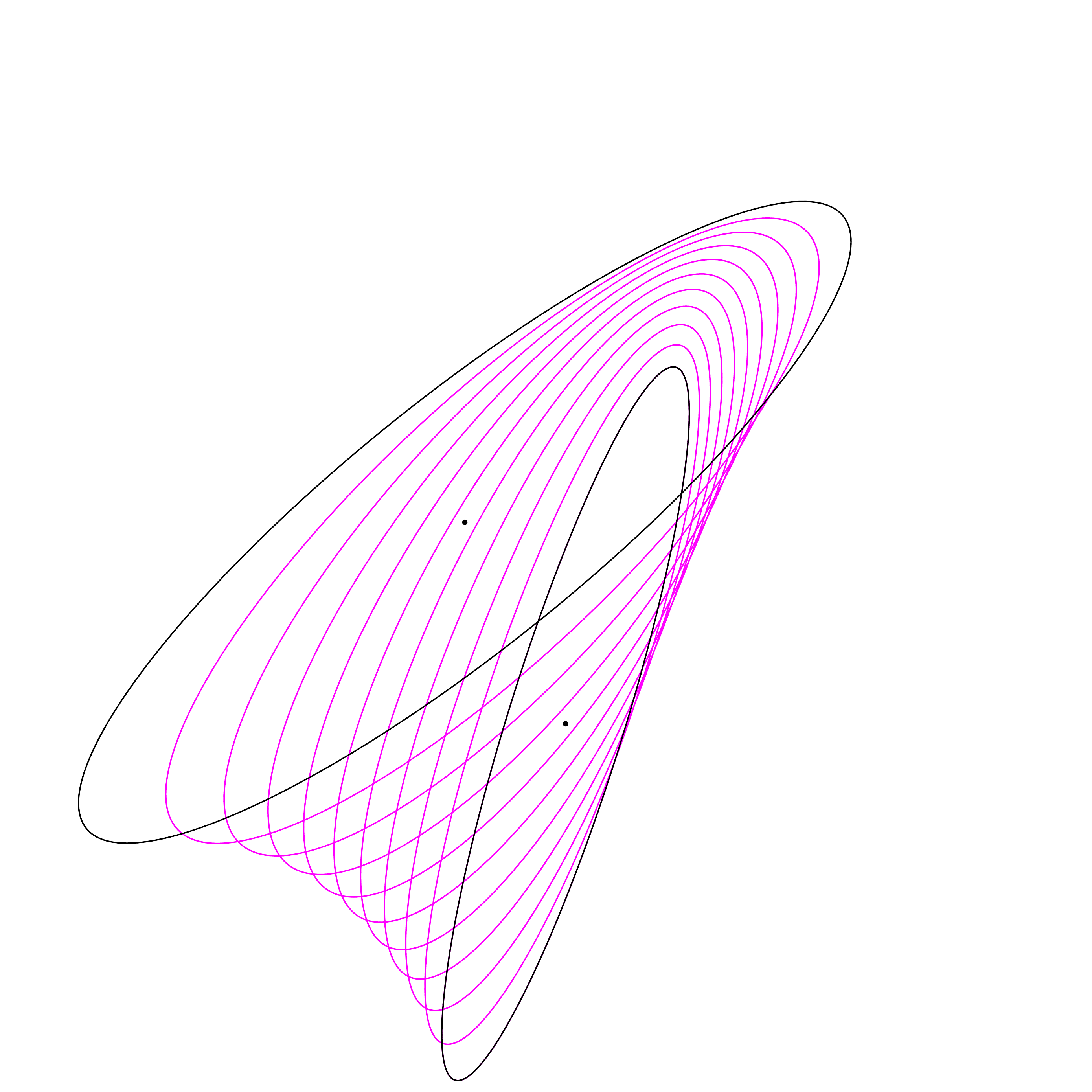} \\
(a) & (b) & (c)\cr
\includegraphics[width=0.3\textwidth]{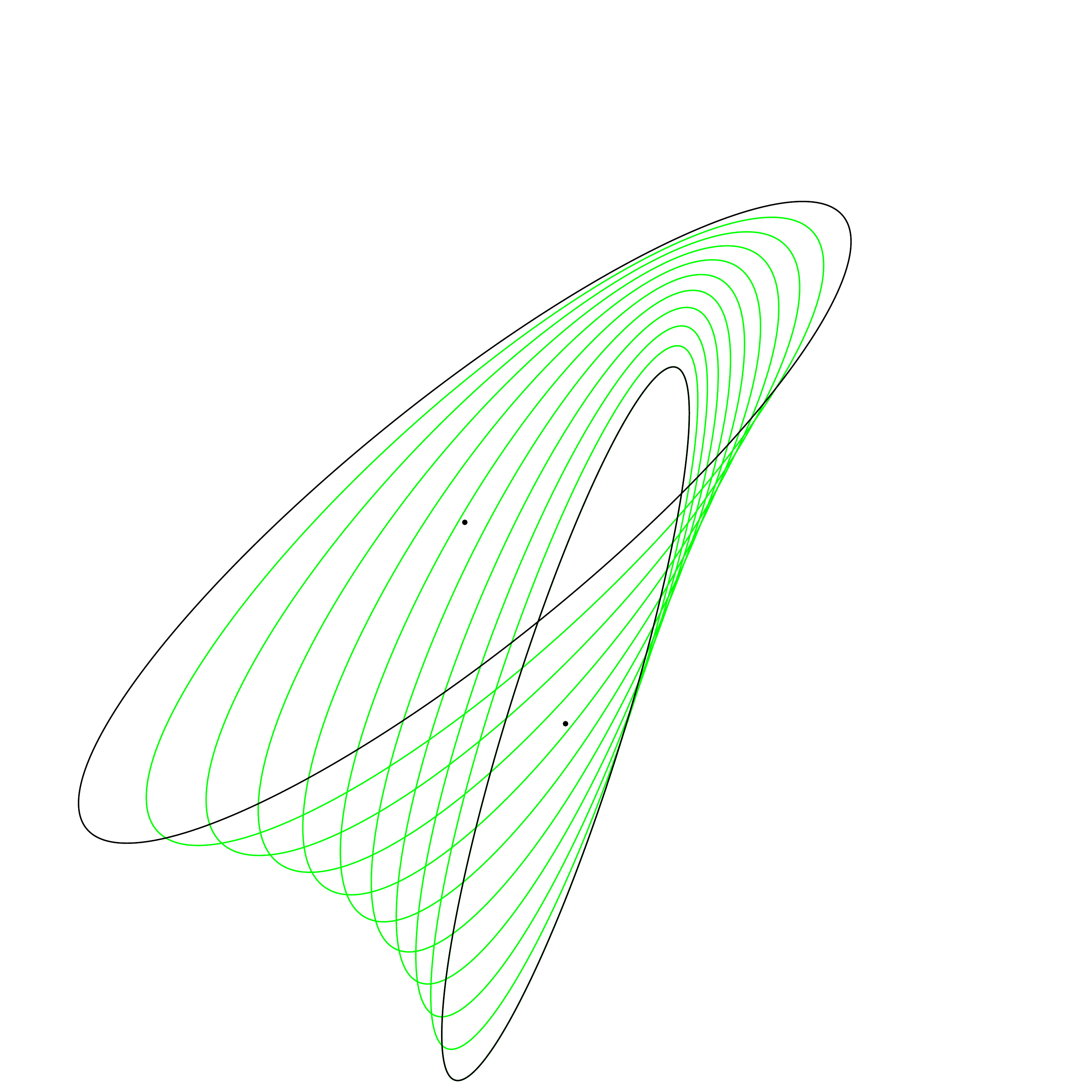} &
\includegraphics[width=0.3\textwidth]{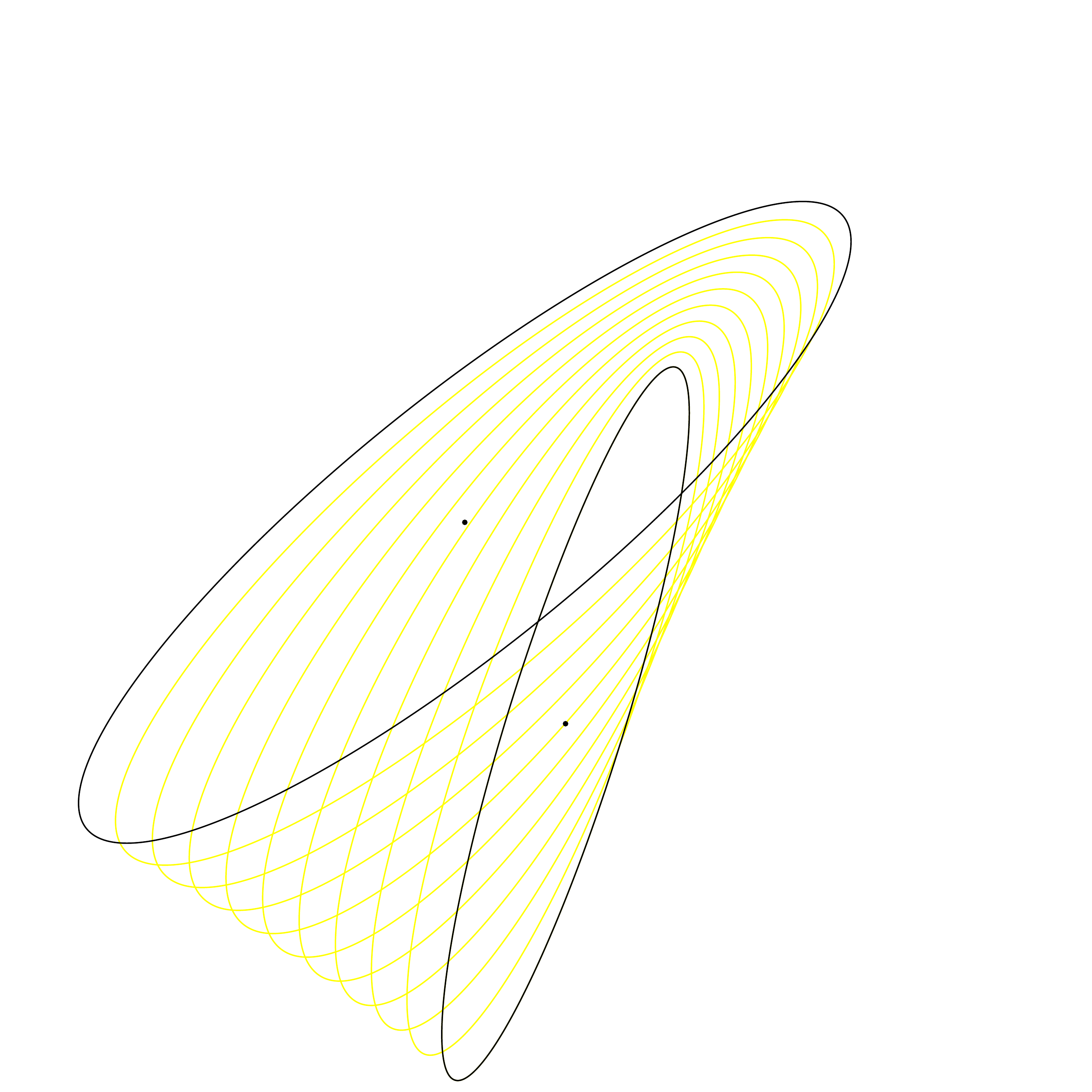} &
\includegraphics[width=0.3\textwidth]{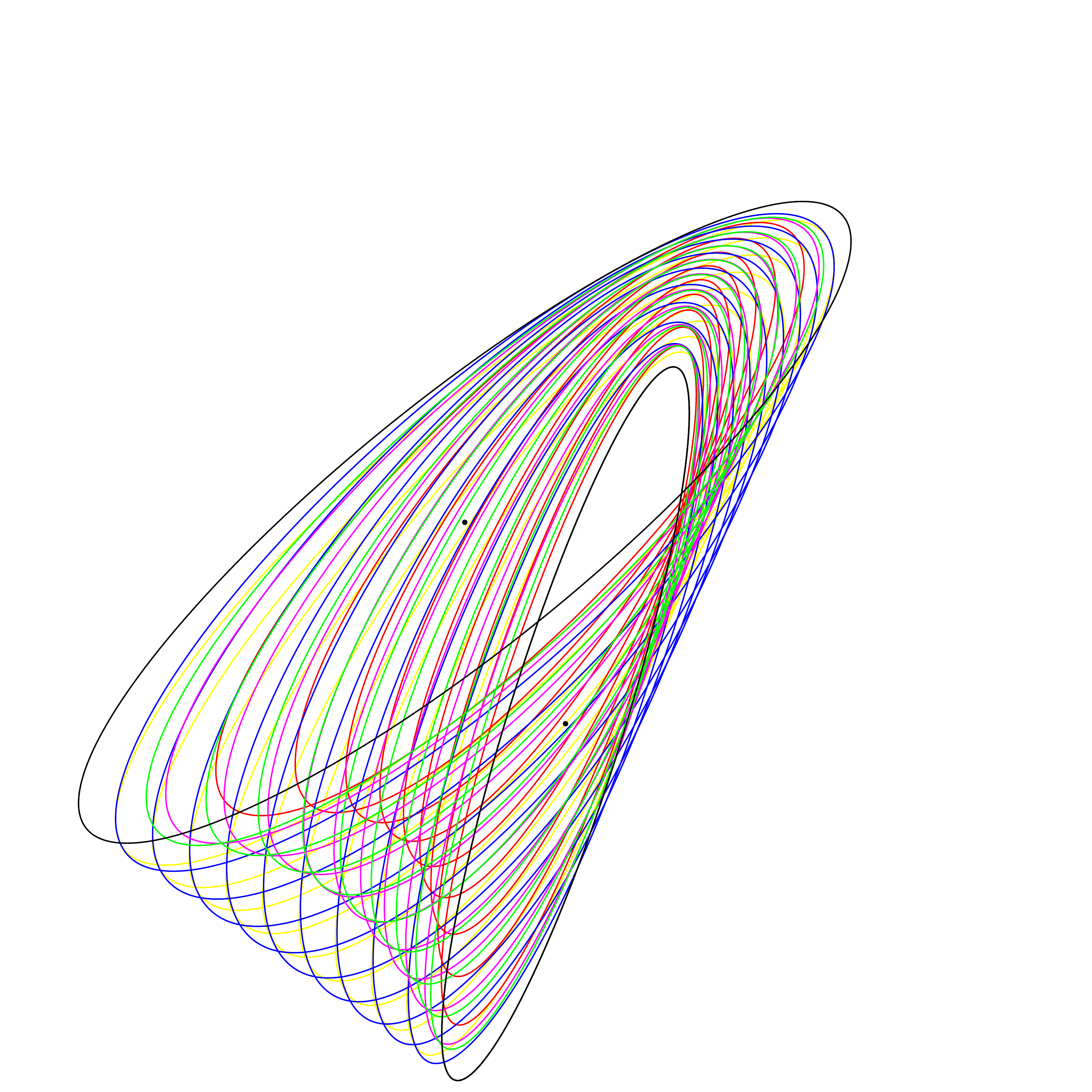} \\
(d) &(e) & (f)
\end{tabular}
\caption{Visualizing at discrete positions (10 increment steps between $0$ and $1$) some curves used to approximate the Fisher-Rao distance between two bivariate normal distributions:
(a) exponential geodesic $c^e=\gamma_\calN^e$ (red),
(b) mixture geodesic $c^m=\gamma_\calN^m$ (blue),
(c) mid mixture-exponential curve $c^{\mathrm{em}}$ (purple),
(d) projected Calvo \& Oller curve $c^{\CO}$ (green), 
(e) $c^\lambda$: ordinary linear interpolation in $\lambda$ (yellow),
and
(f) All superposed curves at once.
 \label{fig:viz2Dbis}}
\end{figure}

\begin{Example}
Let us report some numerical results for bivariate normals with $T=1000$: 

\begin{itemize}

\item We use the following example of Han and Park~\cite{MVNGeodesicShooting-2014} (Eq. 26):
$$
N_1=\left(\vectortwo{0}{0},\mattwotwo{1}{0}{0}{0.1}\right),\quad
N_2=\left(\vectortwo{1}{1},\mattwotwo{0.1}{0}{0}{1}\right).
$$
Their geodesic shooting algorithm~\cite{MVNGeodesicShooting-2014} evaluates the Fisher-Rao distance to 
$\rho_\calN(N_1,N_2) \approx \mathbf{3.1329}$ (precision $10^{-5}$).

We get:
\begin{itemize}
	\item Calvo \& Oller lower bound: $\rho_\CO(N_1,N_2)\approx \mathbf{3.0470}$,
	\item Upper bound using Eq.~\ref{eq:UBMah}: $7.92179$,
	\item SPC upper bound~(Eq.~\ref{prop:USPC}): $U_\SPC(N_1,N_2)\approx 5.4302$,
	\item $\sqrt{D_J}$ upper bound: $U_{\sqrt{J}}(N_1,N_2)\approx \mathbf{4.3704}$,
	\item $\tilde\rho_\calN^\lambda(N_1,N_2)\approx 3.4496$,
	\item $\tilde\rho_\calN^m(N_1,N_2)\approx 3.5775$,
	\item $\tilde\rho_\calN^e(N_1,N_2)\approx 3.7314$,
	\item $\tilde\rho_\calN^{\mathrm{em}}(N_1,N_2)\approx 3.1672$,
	\item $\tilde\rho_\calN^{\CO}(N_1,N_2)\approx \mathbf{3.1391}$.
\end{itemize}
In that setting, the $\sqrt{D_J}$ upper bound is better than the upper bound of Eq.~\ref{prop:USPC}, and the projected Calvo \& Oller geodesic yields the best approximation of the Fisher-Rao distance (Figure~\ref{Fig:HanPark}) with an absolute error of $0.0062$ (about $0.2\%$ relative error).
When $T=10$, we have  $\tilde\rho_\calN^{\CO}(N_1,N_2)\approx 3.1530$, when $T=100$, we get  $\tilde\rho_\calN^{\CO}(N_1,N_2)\approx 3.1136$, and when $T=500$ we obtain $\tilde\rho_\calN^{\CO}(N_1,N_2)\approx 3.1362$ (which is better than the approximation obtained for $T=1000$).
Figure~\ref{Fig:HanParkRangeT} shows the fluctuations of the approximation of the Fisher-Rao distance by the projected C\&O curve when $T$ ranges from $3$ to $100$.

\begin{figure}
\centering
\begin{tabular}{cc}
\includegraphics[width=0.3\textwidth]{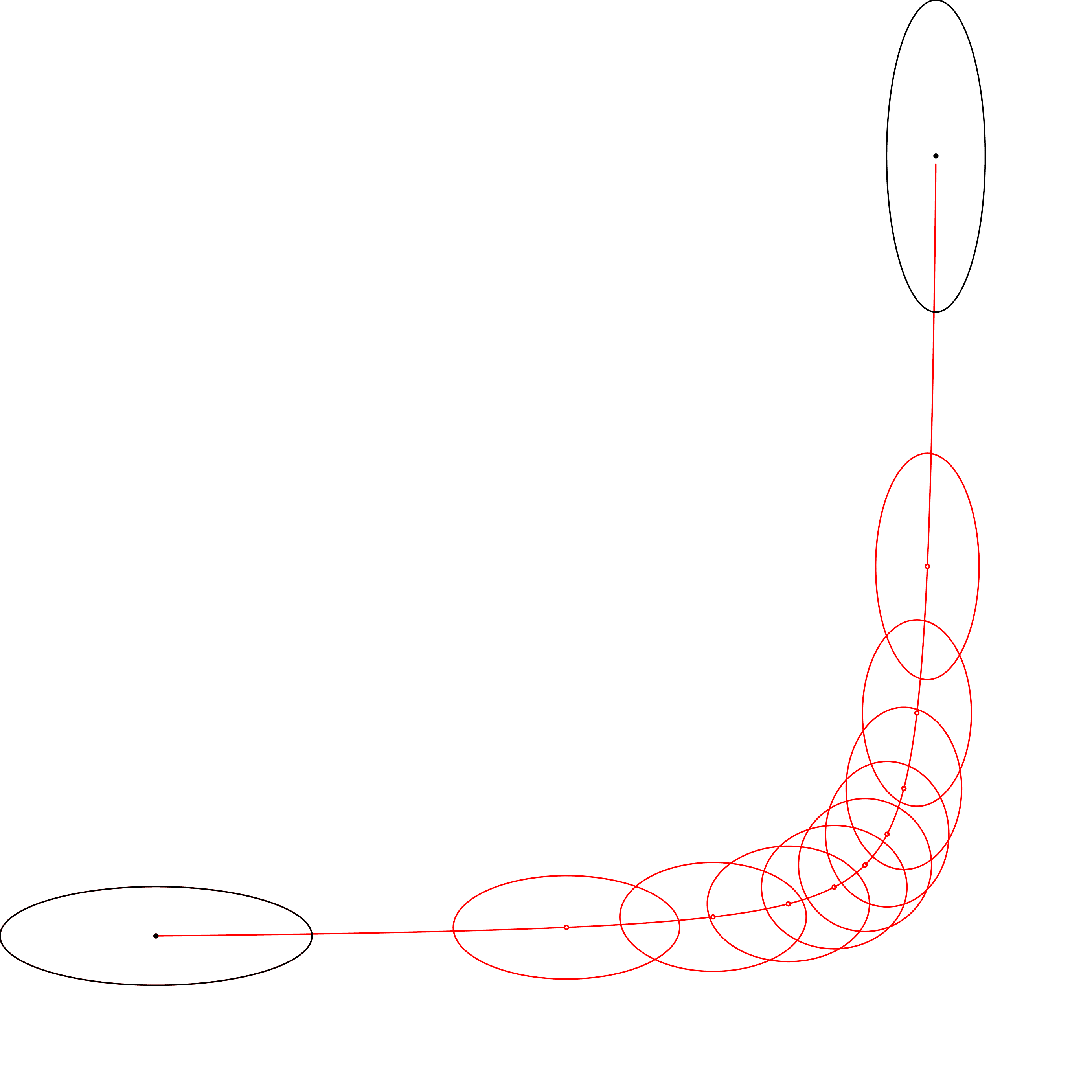} &
\includegraphics[width=0.3\textwidth]{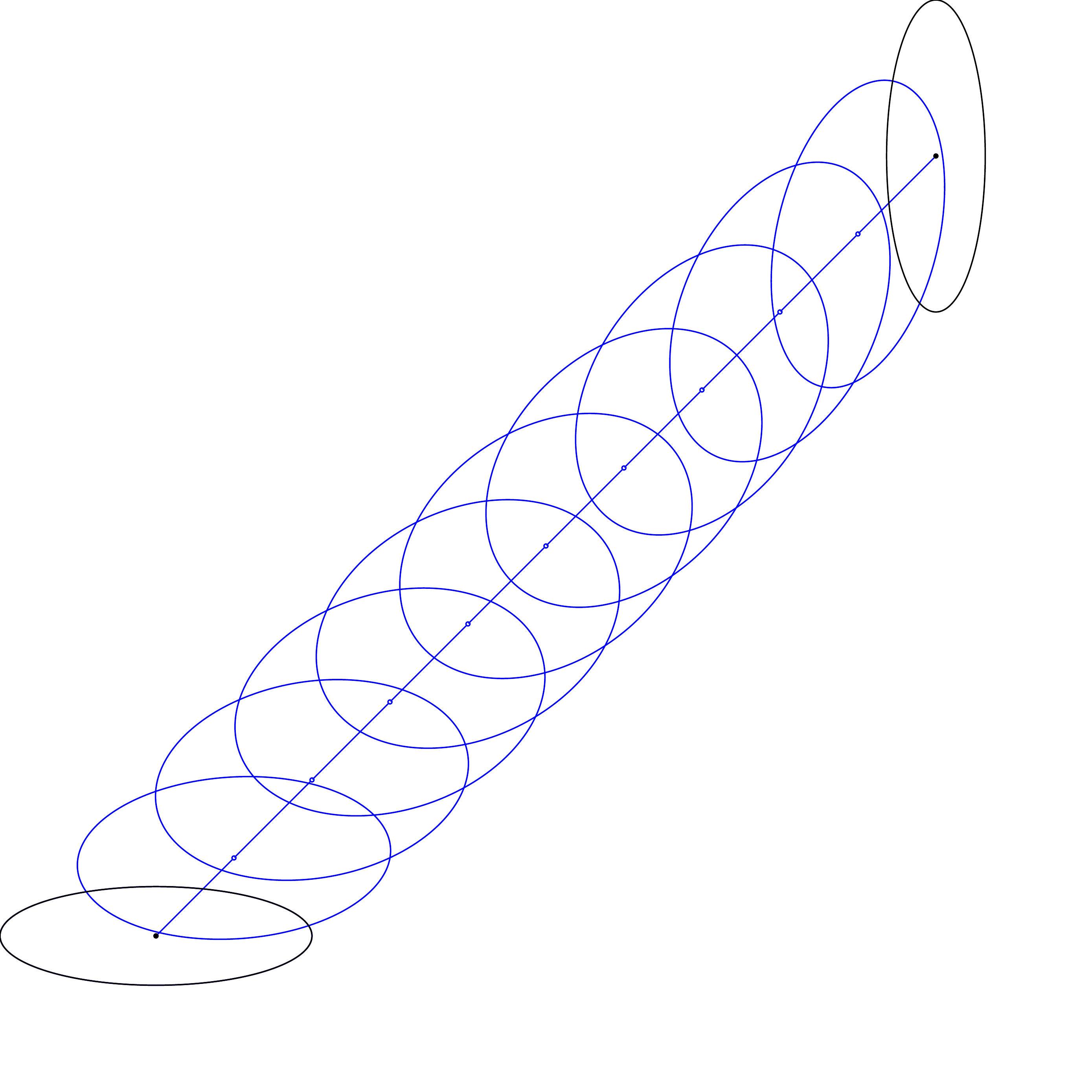} \\
(a) & (b) \\
\includegraphics[width=0.3\textwidth]{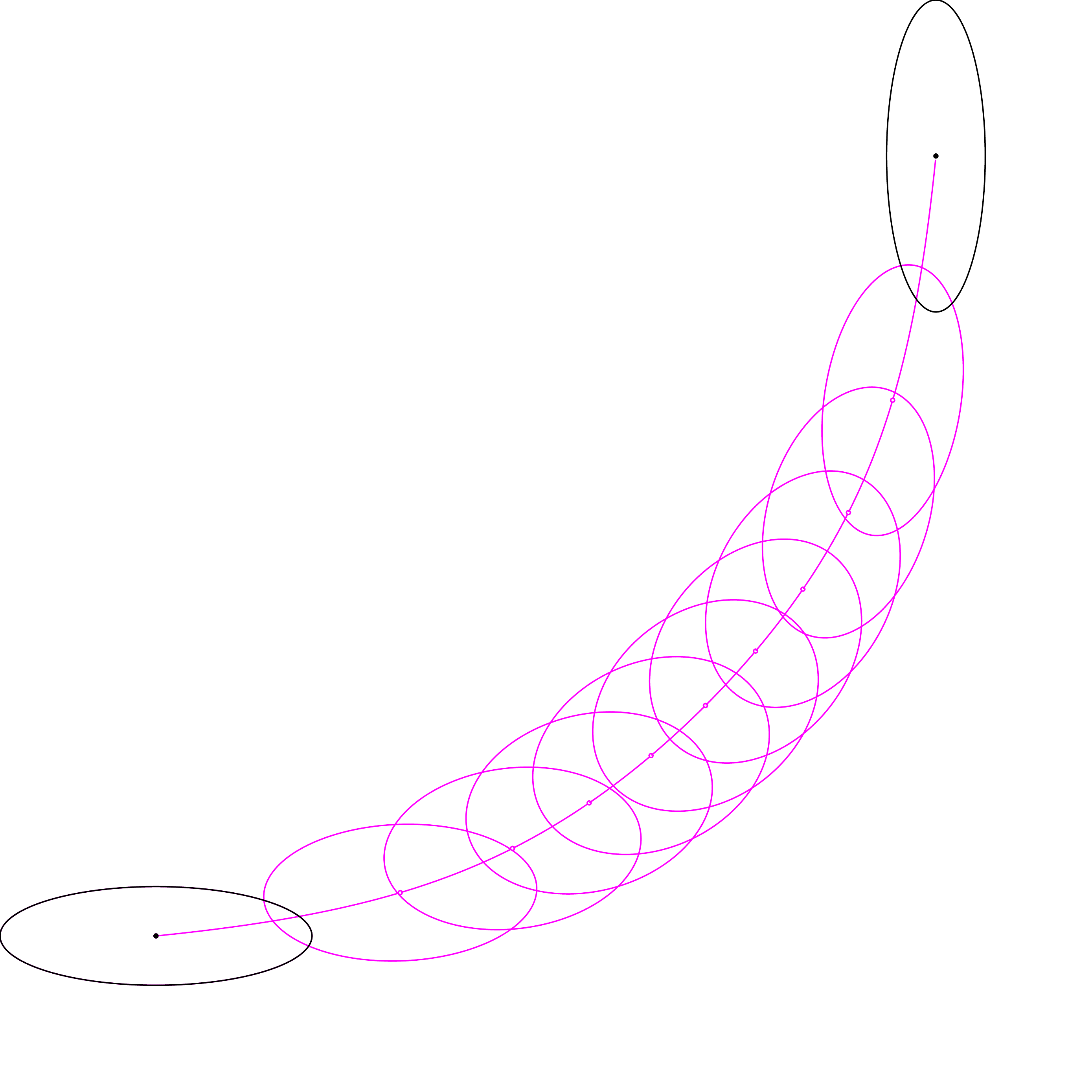} &
\includegraphics[width=0.3\textwidth]{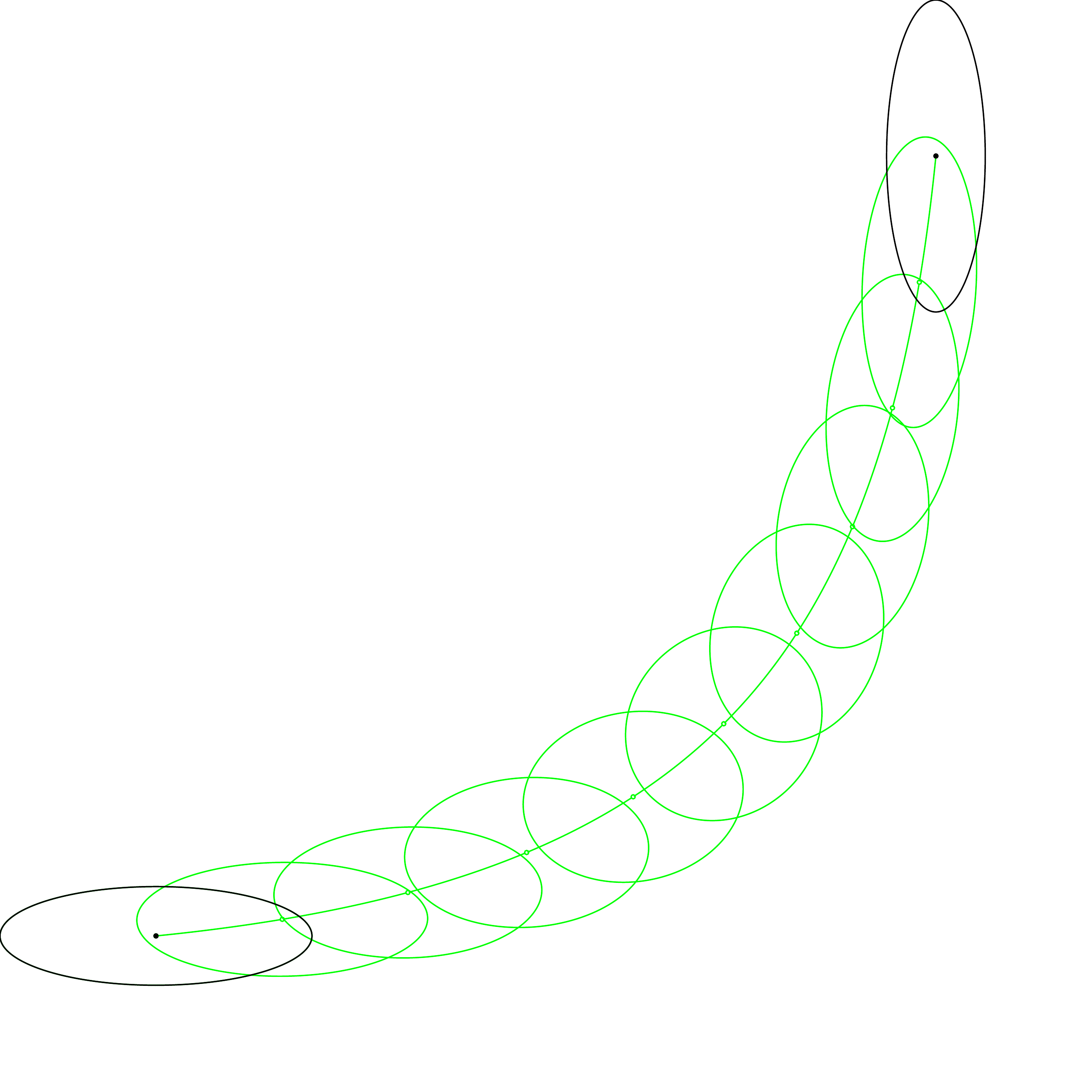} \\
(c) & (d) \\
\includegraphics[width=0.3\textwidth]{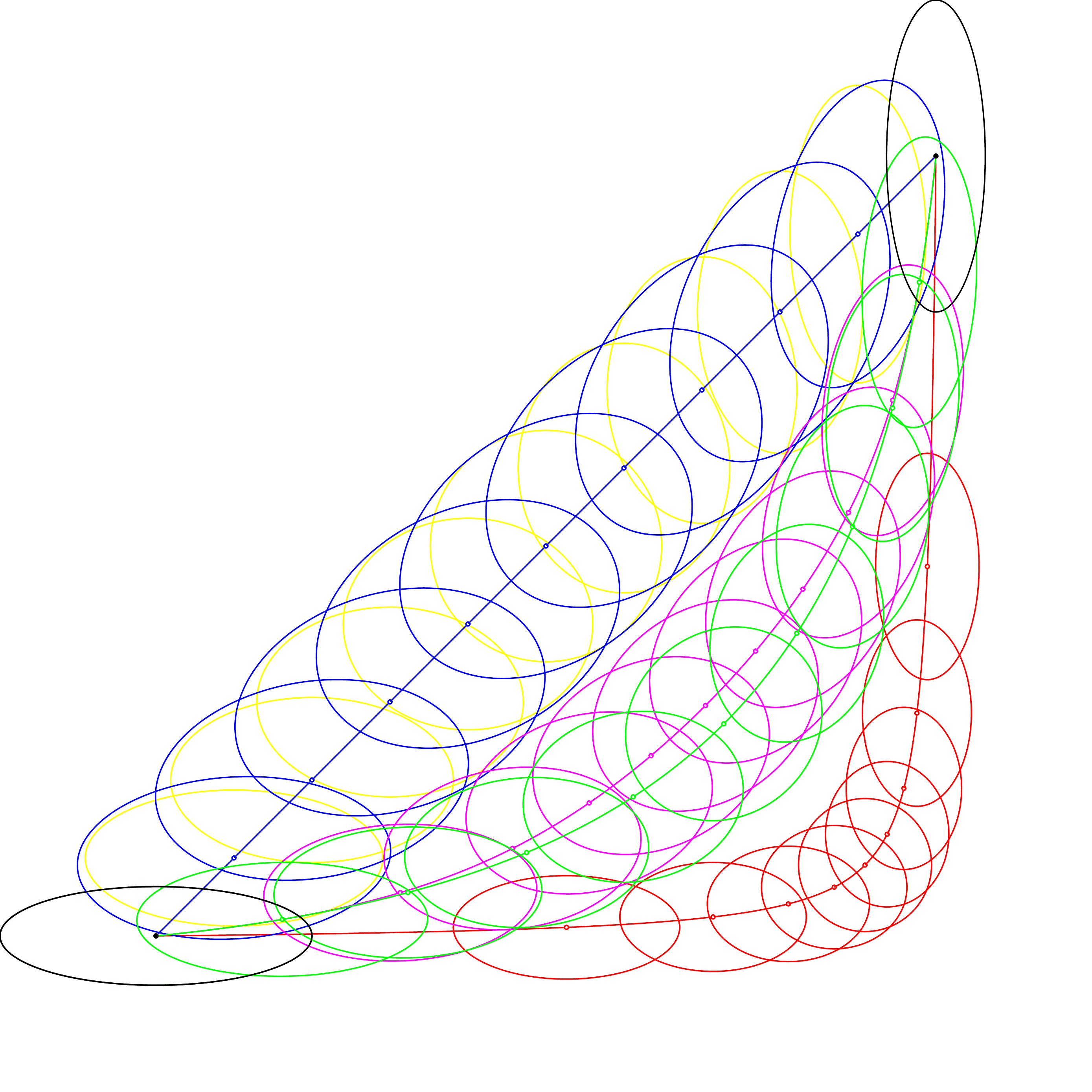} &
\includegraphics[width=0.3\textwidth]{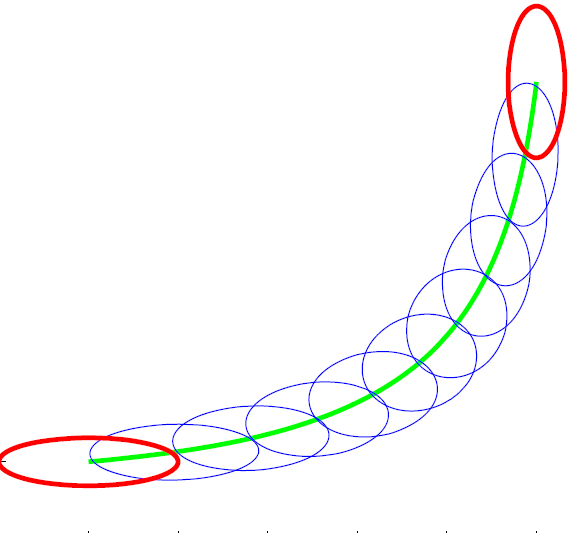} \\
(e) & (f) 
\end{tabular}
\caption{Comparison of our approximation curves with the Fisher-Rao geodesic (f) obtained by geodesic shooting (Figure~5 of~\cite{MVNGeodesicShooting-2014}). 
Exponential (a) and mixture (b) geodesics with the mid exponential-mixture curve (c), and the projected C\&O curve (d).
Superposed curves (e) and comparison with geodesic shooting (Figure~5 of~\cite{MVNGeodesicShooting-2014}).
Beware that color coding are not related between (a) and (b), and scale for depicting ellipsoids are different.}\label{Fig:HanPark}
\end{figure}

\begin{figure}
\centering
\includegraphics[width=0.65\textwidth]{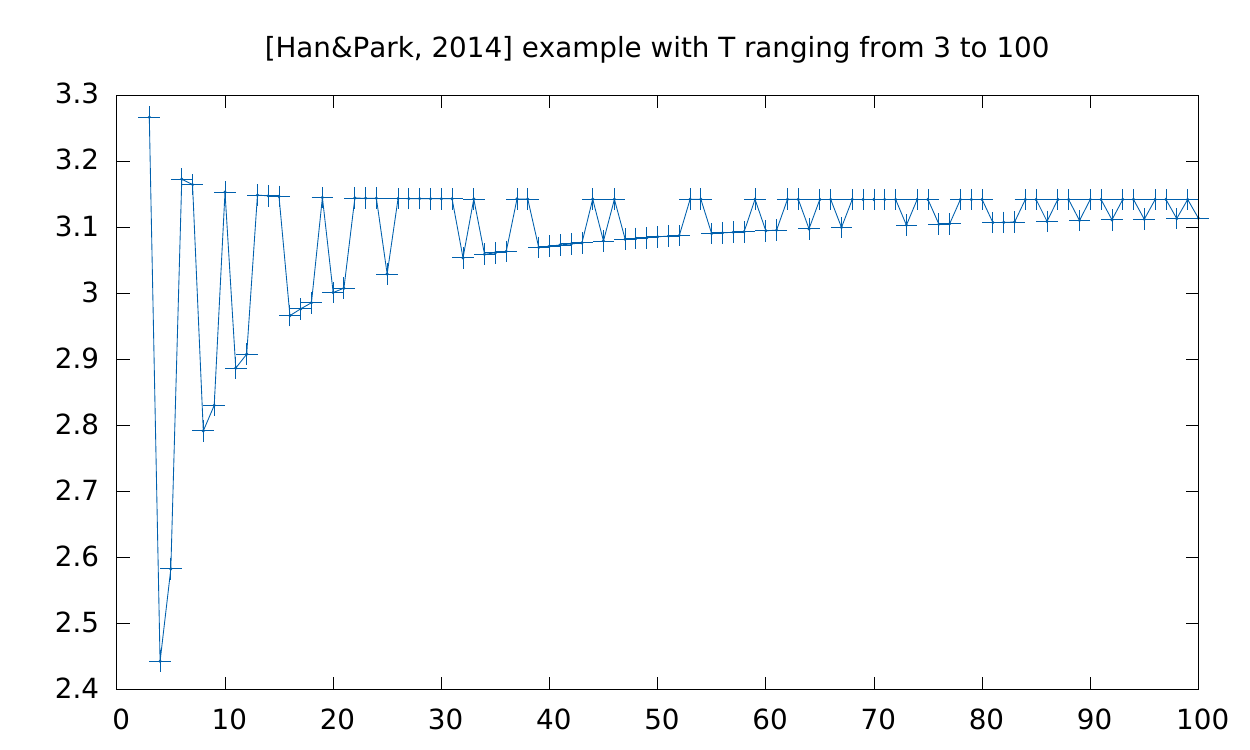} 
\caption{Approximating of the Fisher-Rao distance obtained by using the projected C\&O curve when $T$ ranges from $3$ to $100$.}
\label{Fig:HanParkRangeT}
\end{figure}
 

\item Bivariate normal $N_1=(0,I)$ and bivariate normal $N_2=(\mu_2,\Sigma_2)$ with $\mu_2=[1\ 0]^\top$ and 
$\Sigma_2=\mattwotwo{1}{-1}{-1}{2}$.
We get
\begin{itemize}
	\item Calvo \& Oller lower bound: {\bf $1.4498$}
	\item Upper bound~of Eq.~\ref{prop:USPC}: $2.6072$
	\item $\sqrt{D_J}$ upper bound: {\bf $1.5811$}
	\item $\tilde\rho^\lambda$: $1.5068$
	\item $\tilde\rho^m$: $1.5320$
	\item $\tilde\rho^e$: $1.5456$
	\item $\tilde\rho^{\mathrm{em}}$: $1.4681$
	\item $\tilde\rho^{\mathrm{co}}$: {\bf $1.4673$}
\end{itemize}
 \item Bivariate normal $N_1=(0,I)$ and bivariate normal $N_2=(\mu_2,\Sigma_2)$ with $\mu_2=[5\ 0]^\top$ and 
$\Sigma_2=\mattwotwo{1}{-1}{-1}{2}$.
We get:
\begin{itemize}
	\item Calvo \& Oller lower bound: {\bf $3.6852$}
	\item Upper bound of Eq.~\ref{prop:USPC}: {\bf $6.0392$}
	\item $\sqrt{D_J}$ upper bound: $6.2048$
	\item $\tilde\rho^\lambda$: $5.7319$
	\item $\tilde\rho^m$: $4.4039$
	\item $\tilde\rho^e$: $5.9205$
	\item $\tilde\rho^{\mathrm{em}}$: {\bf $4.2901$}
	\item $\tilde\rho^{\mathrm{co}}$: {$4.3786$}
\end{itemize}
\end{itemize}
\end{Example}



\section{Approximating the smallest enclosing Fisher-Rao ball of MVNs}\label{sec:MEBMVN}   

We may use these closed-form distance $\rho_\CO(N,N')$ between $N$ and $N'$ 
to compute an approximation (of the center) of the smallest enclosing Fisher-Rao ball $B^*=\ball(C^*,r^*)$ of a set 
$\calG=\{N_1=(\mu_1,\Sigma_1),\ldots,N_n=(\mu_n,\Sigma_n)\}$ of $n$ $d$-variate normal distributions:
$$
C^*=\arg\min_{C\in\calN} \max_{i\in\{1,\ldots,n\}} \rho_\calN(C,N_i)
$$
where $\ball(C,r)=\{N\in\calN\st \rho_\calN(C,N)\leq r\}$.

The method proceeds as follows:  
\begin{itemize}
\item First, we convert MVN set $\calG$ into the equivalent set of $(d+1)$-dimensional SPD matrices $\bar\calG=\{\barP_i=f(N_i)\}$ using the C\&O embedding. We relax the problem of approximating the circumcenter $C^*$ of the smallest enclosing Fisher-Rao ball by
$$
P^*=\arg\min_{P\in\bbP(d+1)} \max_{i\in\{1,\ldots,n\}} \rho_\CO(P,\barP_i).
$$

\item Second we approximate the center of the smallest enclosing Riemannian ball of $\bar\calG$ using the iterative smallest enclosing Riemannian  ball algorithm in~\cite{RieMinimax-2013} with say $T=1000$ iterations. 
Let $\tilde{P}\in\bbP(d+1)$ denote this approximation center: ${P}_T=\RieSEB_\SPD(\bar\calG,T)$.

\item Finally, we project back $P_T$ onto $\barN$: $\bar P_T=\proj_{\barN}({P}_T)$. 
We return $\bar P_T$ as the approximation of $C^*$.
\end{itemize}

Algorithm~\cite{RieMinimax-2013} $\RieSEB_\SPD(\{P_1,\ldots, P_n\},T)$ is described for a set of SPD matrices $\{P_1,\ldots, P_n\}$ as follows:

\begin{itemize}
	\item Let $C_1\leftarrow P_1$
	\item For $t=1$ to $T$
	\begin{itemize}
		\item Compute the index of the SPD matrix which is farthest to current circumcenter $C_t$:
		$$
		f_t=\arg\max_{i\in\{1,\ldots,n\}} \rho_\SPD(C_t,P_i)
		$$
		\item Update the circumcenter by walking along the geodesic linking $C_t$ to $P_{f_t}$:
		$$
		C_{t+1}=\gamma_\SPD\left(C_t,P_{f_t};\frac{1}{t+1}\right)
		= C_t^{\frac{1}{2}}(C_t^{-\frac{1}{2}}P_{f_t}C_t^{-\frac{1}{2}})^{\frac{1}{t+1}} C_t^{\frac{1}{2}}
		$$
	\end{itemize}
	\item Return $C_T$
\end{itemize}
Convergence of the algorithm $\RieSEB_\SPD$ follows from the fact that the SPD trace manifold is a Hadamard manifold (with negative sectional curvatures). See~\cite{RieMinimax-2013} for a proof of convergence (and~\cite{nielsen2015approximating} for a convergence proof of the similar algorithm in hyperbolic geometry).

The SPD distance $\rho_\calP(C_T,\bar C_T)$ provides an indication of the quality of the approximation.
Figure~\ref{fig:MiniBallCO} shows the result of implementing this heuristic.

Let us notice that when all MVNs share the same covariance matrix $\Sigma$, we have from Eq.~\ref{eq:FRsamecovar} or Eq.~\ref{eq:cosamecovar} 
that $\rho_\calN(\mu_1,\Sigma),N(\mu_2,\Sigma)$ and $\rho_{\CO}(N(\mu_1,\Sigma),N(\mu_2,\Sigma))$ are strictly increasing function of their Mahalanobis distance.
Using the the Cholesky decomposition $\Sigma^{-1}=LL^\top$, we deduce that the smallest Fisher-Rao enclosing ball coincides with the
smallest Calvo \& Oller enclosing ball, and  the  circumcenter of that ball can be found  as an ordinary Euclidean circumcenter~\cite{welzl2005smallest} 
(Figure~\ref{fig:MiniBallCO}(b)).
Note that in 1D, we can find the exact smallest enclosing Fisher-Rao ball as an equivalent smallest enclosing ball in hyperbolic geometry~\cite{nielsen2010hyperbolic}.

Furthermore, we may extend the computation of the approximated circumcenter to  $k$-center clustering~\cite{gonzalez1985clustering} of $n$ multivariate normal distributions. Since the circumcenter of the clusters are approximated and not exact, we extend straightforwardly the variational approach of $k$-means described in~\cite{acharyya2013bregman} to $k$-center clustering.
An application of $k$-center clustering of MVNs is to simplify a Gaussian mixture model~\cite{FRMVNReview-2020} (GMM).

Similarly, we can consider other Riemannian distances with closed form formula between MVNs like the Killing distance in the symmetric space~\cite{lovric2000multivariate} or the Siegel-based distance proposed in Appendix~\ref{sec:Siegel}.

\begin{figure}
\centering
\begin{tabular}{ccc}
\includegraphics[width=0.3\textwidth]{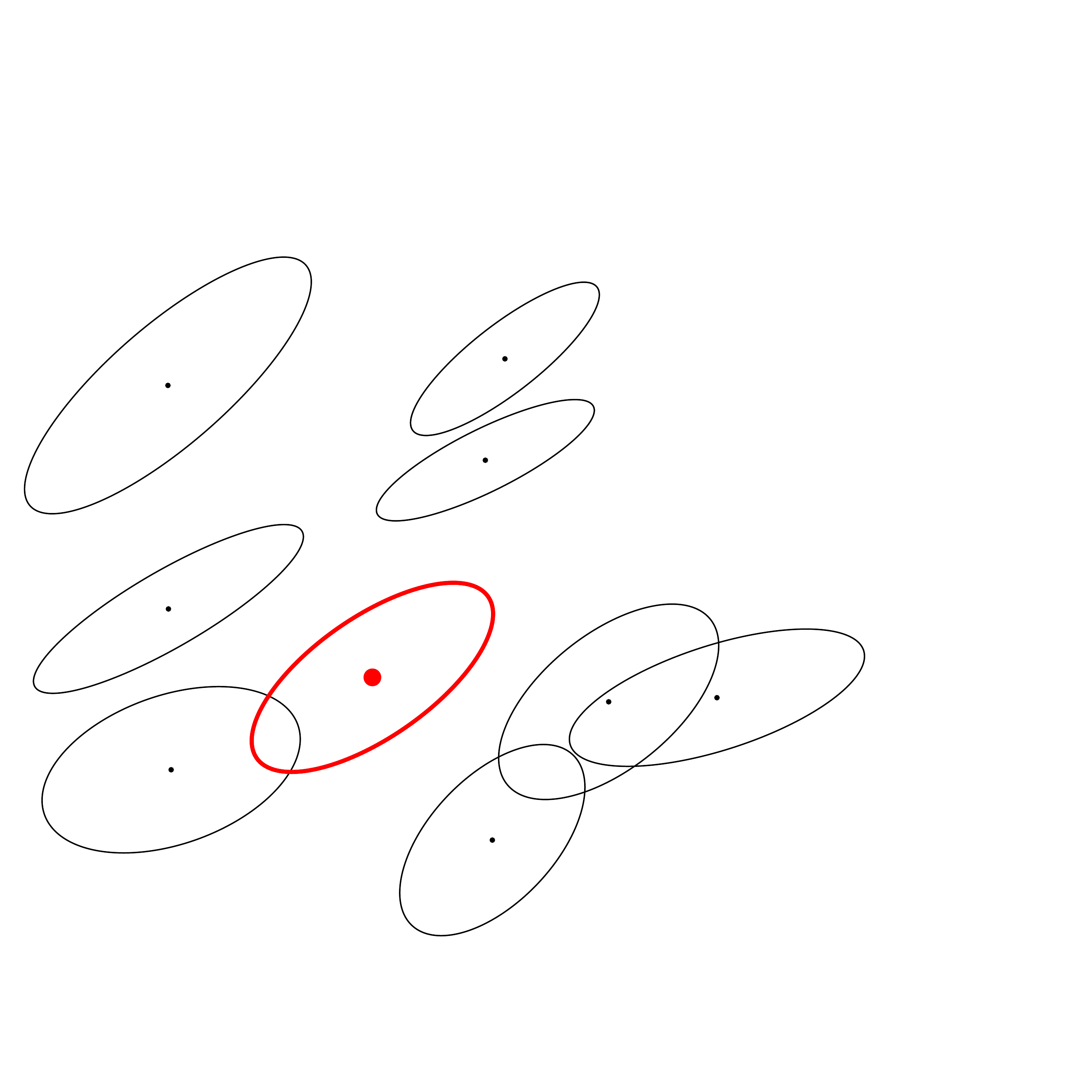} &
\includegraphics[width=0.3\textwidth]{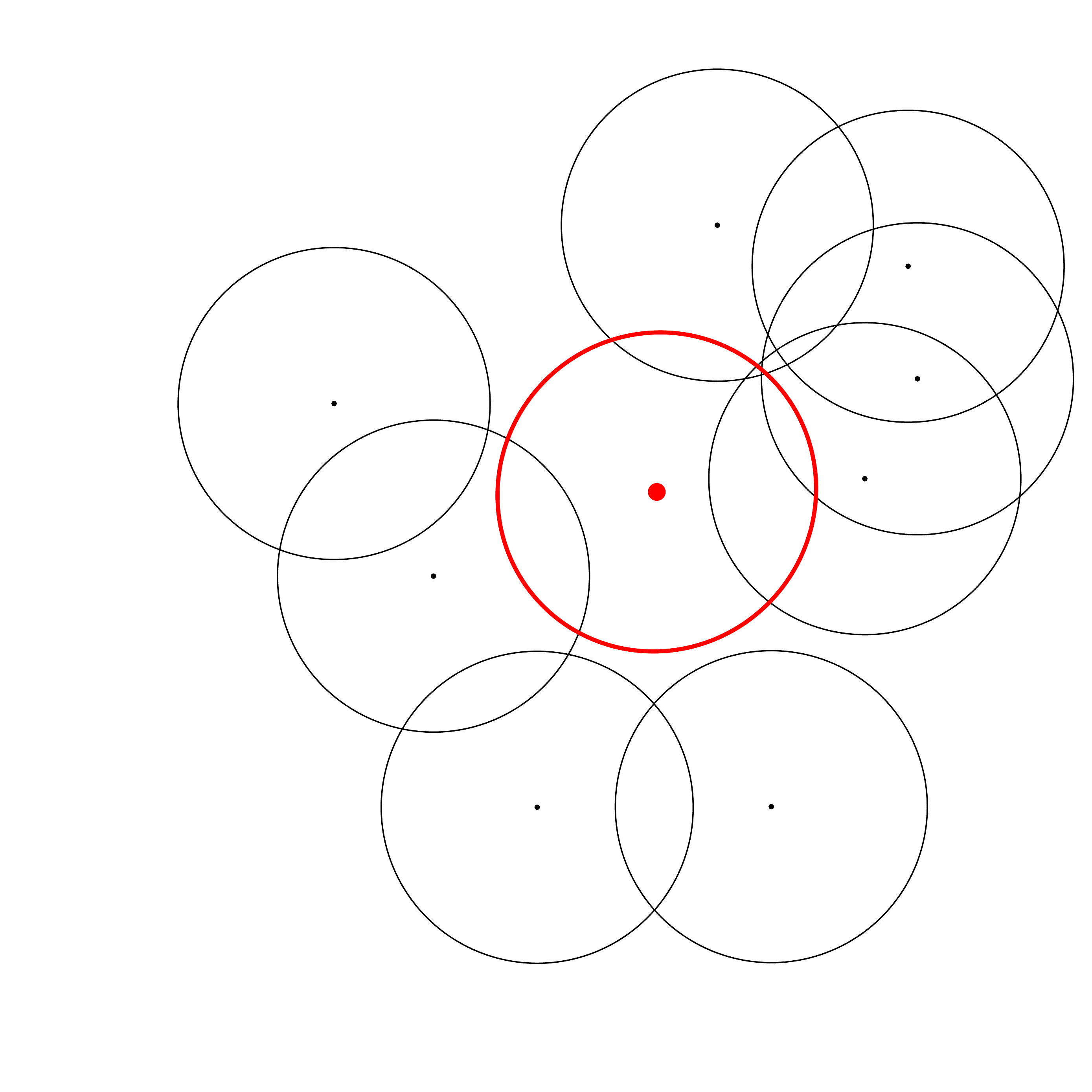} &
\includegraphics[width=0.3\textwidth]{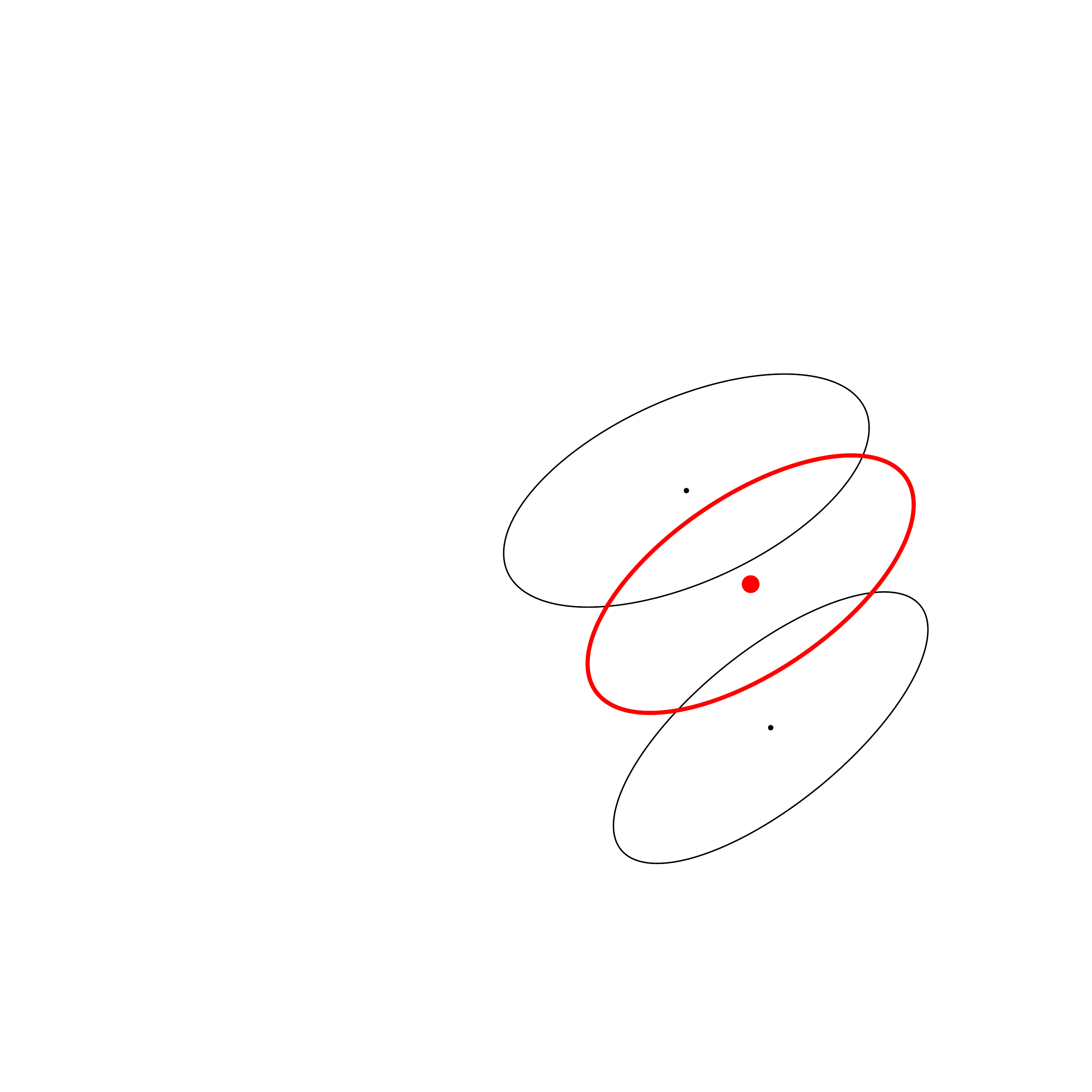}\\
(a) & (b) & (c)
\end{tabular}

\caption{Approximation of the smallest enclosing Riemannian ball of a set of $n$ bivariate normals $N_i=N(\mu_i,\Sigma_i)$ with respect to C\&O distance $\rho_\CO$ (the approximate circumcenter $\bar C_T$ is depicted as a red ellipse):
(a) $n=8$ with different covariance matrices,
(b) $n=8$ with identical covariance matrices amount to a smallest enclosing ball of a set of $n$ points $\{\mu_i\}$,
(c) $n=2$ displays the midpoint of the C\&O geodesic visualized as an equivalement bivariate normal distribution in the sample space.
 \label{fig:MiniBallCO}}
\end{figure}

\section{Some information-geometric properties of the C\&O embedding}\label{sec:prop}  

In information geometry~\cite{IG-2016}, the manifold $\calN$ admits a dual structure $(\calN,g^\Fisher_\calN,\nabla^e_\calN,\nabla^m_\calN)$ when equipped with the exponential connection $\nabla^e_\calN$ and the mixture connection $\nabla^m_\calN$. 
The connections $\nabla^e_\calN$ and $\nabla^m_\calN$ are said dual since $\frac{\nabla^e_\calN+\nabla^m_\calN}{2}=\bar\nabla_\calN$, the Levi-Civita connection induced by $g^\Fisher_\calN$.
Furthermore, by viewing $\calN$ as an exponential family $\{p_\theta\}$ with  natural parameter $\theta=(\theta_v,\theta_M)$ (using the sufficient statistics~\cite{nielsen2019jensen} $(x,-xx^\top)$), and taking the convex log-normalizer function $F_\calN(\theta)$ of the normals, we can build a dually flat space~\cite{IG-2016} where the canonical divergence amounts to a Bregman divergence which coincides with the reverse Kullback-Leibler divergence~\cite{ohara1996dualistic,IG-MVN-1999} (KLD).
The Legendre duality 
$$
F^*(\eta)=\inner{\nabla F(\theta)}{\eta}-F(\nabla F(\theta))
$$
 (with $\inner{(v_1,M_1)}{(v_2,M_2)}=\tr(v_1v_2^\top+M_1M_2^\top)=v_1\cdot v_2+\tr(M_1M_2^\top)$) yields:
$\theta =  (\theta_v,\theta_M)=\left(\Sigma^{-1}\mu,\frac{1}{2}\Sigma^{-1}\right)$,
$$
F_\calN(\theta) = \frac{1}{2}\left(d\log\pi-\log|\theta_M|+\frac{1}{2}\theta_v^\top\theta_M^{-1}\theta_v\right),$$
$\eta=(\eta_v,\eta_M)=\nabla F_\calN(\theta)=\left(\frac{1}{2}\theta_M^{-1}\theta_v,\theta_M^{-1}\right)$,
$$F^*_\calN(\eta)= -\frac{1}{2}\left(\log(1+\eta_v^\top\eta_M^{-1}\eta_v)+\log|-\eta_M|+d(\log 2\pi e)\right),$$
and we have
$$
B_{F_\calN}(\theta_1,\theta_2)=D_\KL^*(p_{\lambda_1}:p_{\lambda_2})=D_\KL(p_{\lambda_2}:p_{\lambda_1})=B_{F^*_\calN}(\eta_2:\eta_1),
$$
where $D_\KL^*[p:q]=D_\KL[q:p]$ is the reverse KLD.

In a dually flat space, we can express the canonical divergence as a Fenchel-Young divergence using the mixed coordinate systems $
B_{F_\calN}(\theta_1:\theta_2)=Y_{F_\calN}(\theta_1:\eta_2)$
where $\eta_i=\nabla F_\calN(\theta_i)$ and 
$$
Y_{F_\calN}(\theta_1:\eta_2):=F_\calN(\theta_1)+F^*_\calN(\eta_2)-\inner{\theta_1}{\eta_2}.
$$
The  moment $\eta$-parameterization of a normal is $(\eta=\mu,H=-\Sigma-\mu\mu^\top)$
with its reciprocal function $(\lambda=\eta,\Lambda=-H-\eta\eta^\top)$.

Let $F_\calP(P)=F_\calN(0,P)$, $\bar\theta=\frac{1}{2}\bar P^{-1}$, $\bar\eta=\nabla F_{\calP}(\bar\theta)$. 
Then we have the following proposition which proves that the Fenchel-Young divergences in $\calN$ and $\barN$ (as a submanifold of $\calP$)  coincide:

\begin{Proposition}\label{prop:embedpot}
We have
\begin{eqnarray*}
D_\KL[p_{\mu_1,\Sigma_1}:p_{\mu_2,\Sigma_2}]&=&
B_{F_\calN}(\theta_2:\theta_1)=Y_{F_\calN}(\theta_2:\eta_1)
=Y_{F_\calP}(\bar\theta_2:\bar\eta_1)\\
&=& 
B_{F_\calP}(\bar\theta_2:\bar\theta_1)=
D_\KL[p_{0,\bar P_1=f(\mu_1,\Sigma_2)}:p_{0,\bar P_2=f(\mu_2,\Sigma_2)}].
\end{eqnarray*}
\end{Proposition}

Consider now the $\nabla^e$-geodesics and $\nabla^m$-geodesics on $\calN$ (linear interpolation with respect to natural and dual moment parameterizations, respectively):
$\gamma_\calN^e(N_1,N_2;t)=(\mu_t^e,\Sigma_t^e)$ and
$\gamma_\calN^m(N_1,N_2;t)=(\mu_t^m,\Sigma_t^m)$.

\begin{Proposition}[Mixture geodesics preserved]\label{prop:geo}
The mixture geodesics are preserved by the embedding $f$: 
$f(\gamma_\calN^m(N_1,N_2;t))=\gamma_\calP^m(f(N_1),f(N_2);t)$.
The exponential geodesics are preserved for subspace of $\calN$ with fixed mean $\mu$: $\calN_\mu$.
\end{Proposition}

\begin{proof}
For the $m$-geodesics, let us check that 
$$
f(\mu_t^m,\Sigma_t^m)=
\mattwotwo{\Sigma_t^m+\mu_t^m{\mu_t^m}^\top}{\mu_t^m}{(\mu_t^m)^\top}{1}=
t \underbrace{f(\mu_1,\Sigma_1)}_{\bar P_1}+(1-t)\underbrace{f(\mu_2,\Sigma_2)}_{\barP_2},
$$
since $\Sigma_t^m+\mu_t{\mu_t^m}^\top=
\bar\Sigma_t+t\mu_1\mu_1^\top+(1-t)\mu_2\mu_2^\top$ $= t (\Sigma_1+\mu_1\mu_1^\top)+(1-t)(\Sigma_2+\mu_2\mu_2^\top)$.
Thus we have $f(\gamma_\calN^m(N_1,N_2;t))=\gamma_\calP^m(\bar P_1,\bar P_2;t)$.
\end{proof}

Therefore all algorithms on $\calN$ which only require $m$-geodesics or $m$-projections~\cite{IG-2016} by minimizing the right-hand side of the KLD can be implemented by algorithms on $\calP$. See for example, the minimum enclosing ball approximation algorithm called BBC in~\cite{BBC-2005}. 
Notice that $\barN_\mu$ (fixed mean normal submanifolds) preserve both mixture and exponential geodesics: 
The submanifolds $\barN_\mu$ are said doubly auto-parallel~\cite{ohara2019doubly}.

\begin{Remark}
In~\cite{calin2014geometric} (p. 355), exercises 13.8 and 13.9 ask to prove the equivalence of the following statements for $\calS$ a submanifold of $\calM$:
\begin{itemize}
	\item $\calS$ is an exponential family $\Leftrightarrow$ $\calS$ is $\nabla^1$-autoparallel in $\calM$ (exercise 13.8),
	\item $\calS$ is a mixture family $\Leftrightarrow$ $\calS$ is $\nabla^{-1}$-autoparallel in $\calM$ (exercise 13.9).
\end{itemize}  
Let ${\bar P}=\mattwotwo{\Sigma+\mu\mu^\top}{\mu}{\mu^\top}{1}$ (with $|{\bar P}|=|\Sigma|$),
$
{\bar P}^{-1}=
\mattwotwo{\Sigma^{-1}}{-\Sigma^{-1}\mu}{-\mu^\top\Sigma^{-1}}{1+\mu^\top \Sigma^{-1}\mu}$, and $y=(x,1)$.
Then we have
\begin{eqnarray*}
q_{\bar P}(y) &=& \frac{1}{(2\pi)^\frac{d+1}{2} \sqrt{|\bar P|}} \exp\left(-\frac{1}{2}y^\top \bar P^{-1} y\right),\\
 &=& \frac{1}{(2\pi)^\frac{d+1}{2} \sqrt{|\Sigma|}} \exp\left(-\frac{1}{2}y^\top \bar P^{-1} y\right),\\
&=& \frac{1}{(2\pi)^\frac{d+1}{2} \sqrt{|\Sigma|}} \exp\left([x^\top\ 1] \mattwotwo{\Sigma^{-1}}{-\Sigma^{-1}\mu}{-\mu^\top\Sigma^{-1}}{1+\mu^\top \Sigma^{-1}\mu} \vectortwo{x}{1}\right).
\end{eqnarray*}
Thus $\barN=\{q_{\bar P}(x,1)\}$ is an exponential family.
Therefore we deduce that $\calP$ is $\nabla^e$-autoparallel in $\calP$.
However, $\barN$ is not a mixture family and thus $\calP$ is not $\nabla^{m}$-autoparallel in $\calP$.
\end{Remark}

\section{Conclusion and discussion}\label{sec:concl}
The Fisher-Rao distance between multivariate normals is not known in closed form:
It is thus usually approximated by costly geodesic shooting techniques~\cite{MVNGeodesicShooting-2014,GeodesicShooting-2016,barbaresco2019souriau} in practice which requires time-consuming computations of the Riemannian exponential map. In this work, we consider an alternative approach of approximating the Fisher-Rao distance by approximating the Riemannian lengths of closed-form curves. In particular, we considered the mixed exponential-mixture curved and the projected symmetric positive-definite matrix geodesic obtained from Calvo \& Oller  isometric SPD submanifold embedding~\cite{SDPMVN-1990}.
We also reported a fast to compute simplex square root of Jeffreys' divergence for the Fisher-Rao distance which beats the upper bound of~\cite{strapasson2015bounds} when normal distributions are not too far from each others. Finally, we shows that not only
Calvo \& Oller  SPD submanifold embedding~\cite{SDPMVN-1990} is isometric, it also preserves the Kullback-Leibler divergence, the Fenchel-Young divergence and the mixture geodesics. 
Our approximation technique extends to elliptical distributions~\cite{SDPElliptical-2002,chen2022multisensor} which generalize multivariate normal distributions.
We may also consider the Calvo \& Oller metric distance~\cite{SDPMVN-1990} (a lower bound on the Fisher-Rao distance) or the metric distance of the symmetric space~\cite{lovric2000multivariate} (which enjoys asymptotically Fisher-Rao geodesics~\cite{globke2021information}) which admits closed-form formula. 
The C\&O distance is well-suited for short Fisher-Rao distances while the symmetric space distance is well-tailored for large Fisher-Rao distances.
The calculations of these closed-form distances rely on eigenvalues. 

Yet another alternative distance is the Hilbert projective distance on the SPD cone~\cite{nielsen2019clustering} which only requires to calculate the minimal and maximal eigenvalues:
$$
\rho_H(P_1,P_2)= \log\frac{\lambda_{\mathrm{max}}(P_1^{-1}P_2)}{\lambda_{\mathrm{min}}(P_1^{-1}P_2)}.
$$
The dissimilarity is said projective on the SPD cone because $\rho_H(P_1,P_2)=0$ iff. $P_1=\lambda P_2$ for some $\lambda>0$.
However, it is a proper metric distance on $\barN$:
$$
\rho_H(N_1,N_2):=\rho_H(\barP_1,\barP_2),
$$
since $\barP_1=\lambda\barP_2$ iff. $\lambda=1$ (because the array element $P_1[d+1,d+1]=P_2[d+1,d+1]=1$), i.e., $\barP_1=\barP_2$ implying  $P_1=P_2$ by the isometric diffeomorphism $f$.
\vskip 0.5cm
\noindent Additional materials is available online at \url{https://franknielsen.github.io/FisherRaoMVN}

\vskip 0.5cm
\noindent {\bf Acknowledgments.}
I warmly thank Fr\'ed\'eric Barbaresco (Thales) and Mohammad Emtiyaz Khan (Riken AIP) for fruitful discussions about this work.

\vskip 0.5cm
\section*{Notations}

\begin{supertabular}{ll}
\underline{Entities} & \\
$N(\mu,\Sigma)$ & $d$-variate normal distribution (mean $\mu$, covariance matrix $\Sigma$)\\
$p_{(\mu,\Sigma)}(x)$ & Probability density function of $N(\mu,\Sigma)$\\
$q_\Sigma(y)=p_{(0,\Sigma)}(y)$ & Probability density function of $N(0,\Sigma)$\\
$P$ & Positive-definite matrix\\
& \\
\underline{Mappings} & \\
$\barP=f_1(N)$ & Calvo \& Oller mapping~\cite{SDPMVN-1990} (1990)\\
$\hatP=f_{-\frac{1}{d+1},1}(N)=\hat{f}(N)$ & Calvo \& Oller mapping~\cite{SDPElliptical-2002} (2002) or~\cite{lovric2000multivariate}\\
& \\
\underline{Sets} & \\
$\calN$ & Set of multivariate normal distributions $N(\mu,\Sigma)$ (MVNs)\\
$\bbP$ & Symmetric positive-definite matrix cone (SPD matrix cone)\\
$\bbP_c$ & Set of SPD matrices with fixed determinant $c$ ($\bbP=\bbR_{>0}\times \bbP_c$)\\
SSPD, $\bbP_1$ & Set of SPD matrices with unit determinant\\
$\Lambda$ & Parameter space of $N(\mu,\Sigma)$: $\bbR^d\times\bbP(d)$\\
$\calN_0$, $\calP$ & Set of zero-centered normal distributions $N(0,\Sigma)$\\
$\calN_\Sigma$ & Set of normal distributions $N(\mu,\Sigma)$ with fixed $\Sigma$\\
$\calN_\mu$ & Set of normal distributions $N(\mu,\Sigma)$ with fixed $\mu$\\
$\barN$ & Set of SPD matrices $f(N)$\\
& \\
\underline{Riemannian length elements} & \\
MVN Fisher & $\ds_{\Fisher,\calN}^2=\dmu^\top \Sigma^{-1} \dmu + \frac{1}{2}\tr\left(\left(\Sigma^{-1}\dSigma\right)^2\right)$\\
$0$-MVN Fisher & $\ds_{\Fisher,\calN_0}^2=\frac{1}{2}\tr\left(\left(\Sigma^{-1}\dSigma\right)^2\right)$\\
SPD trace & $\ds_{\beta,\trace}^2=\beta \tr((P\,\dP)^2)$  (when $\beta=\frac{1}{2}$, $\ds_\trace=\ds_{\Fisher,\calN_0}$)\\
SPD Calvo\& Oller metric & $\ds^2_\CO=\frac{1}{2}\left(\frac{\dbeta}{\beta}\right)^2 + \beta\dmu^\top \Sigma^{-1}\dmu+\frac{1}{2}\tr\left(\left(\Sigma^{-1}\dSigma\right)^2\right)$\\
 & (with $\ds_\CO=\ds_{\calP}(f(\mu,\Sigma))$) \\
 & when $\beta=1$, $\ds_\CO=\ds_{\Fisher,\calN}$ in $\barN$\\
SPD symmetric space & $\ds_\SS^2=\frac{1}{2}\dmu^\top \Sigma^{-1} \dmu +\tr\left(\left(\Sigma^{-1}\dSigma\right)^2\right)-\frac{1}{2}\tr^2\left(\Sigma^{-1}\dSigma\right)$ \\
Siegel upper space & $\ds_\SH^2 (Z) = 2 \tr\left(Y^{-1}\dZ\ Y^{-1}\dZbar\right)$ ($\ds_\SH(iY)=2\ds_{\Fisher,\calN_0}$)\\
& \\
\underline{Manifolds and submanifolds} & \\
$\calM$ ($=\calM_{\calN}$) & Manifold of multivariate normal distributions\\
$\calS_\mu\subset \calM$ & Submanifold  of MVNs with $\mu$ prescribed\\
$\calS_\Sigma\subset \calM$ & Submanifold  of MVNs with $\Sigma$ prescribed\\
$\calM_\Sigma$ & manifold of $\calN_\Sigma$ (non-embedded in $\calM$)\\
$\calM_\mu$ &  manifold of $\calN_\mu$ (non-embedded in $\calM$)\\
$\calS_{[v],\Sigma}$ & Submanifold of MVN set $\{N(\lambda v,\Sigma)\st \lambda>0\}$ \\
$\calP$ & manifold of symmetric positive-definite matrices\\
& \\
\underline{Distances} & \\
$\rho_N(N_1,N_2)$ & Fisher-Rao distance between normal distributions $N_1$ and $N_2$\\
$\rho_{\SPD}(P_1,P_2)$ & Riemannian SPD distance between  $P_1$ and $P_2$\\
$\rho_\CO(N_1,N_2)$ & Calvo \& Oller distance from embedding $N$ to $\barP=f(N)$\\
$\rho_\SS(N_1,N_2)$ & Symmetric space distance from embedding $N$ to $\hatP=\hat{f}(N)$\\
$D_\KL(N_1,N_2)$ & Kullback-Leibler divergence  between   MVNs $N_1$ and $N_2$\\
$D_J(N_1,N_2)$ & Jeffreys divergence  between   MVNs $N_1$ and $N_2$\\
& \\
\underline{Geodesics and curves} & \\
$\gamma_\calN^\FR(N_1,N_2;t)$ & Fisher-Rao geodesic between  MVNs $N_1$ and $N_2$\\
$\gamma_\calP^\FR(P_1,P_2;t)$ & Fisher-Rao geodesic between  SPD $P_1$ and $P_2$\\
$\gamma_\calN^e(N_1,N_2;t)$ & exponential geodesic between  MVNs $N_1$ and $N_2$\\
$\gamma_\calN^m(N_1,N_2;t)$ & mixture geodesic between MVNs $N_1$ and $N_2$\\
$\gamma_\calN^\CO(N_1,N_2;t)$ & projection curve (not geodesic) of $\gamma_\calP(\barP_1,\barP_2;t)$ onto $\barN$\\
\end{supertabular}

 \bibliographystyle{plain}
 \bibliography{RaoMVN6BIB}

\appendix

\section{Fisher-Rao distance between normal distributions sharing the same covariance matrix}\label{sec:BFRsamecovar}

The Rao distance between $N_1=N(\mu_1,\Sigma)$ and $N_2=N(\mu_2,\Sigma)$ has been reported in closed-form~\cite{FRMVNReview-2020} (Proposition~3). We shall explain the geometric method with full description as follows:
Let $(e_1,\ldots,e_d)$ be the standard frame of $\bbR^d$ (ordered basis): The $e_i$'s are the unit vectors of the axis $x_i$'s. 
Let $P$ be an orthogonal matrix such that $P\, (\mu_2-\mu_1)=\|\mu_2-\mu_1\|_2\, e_1$ 
(i.e., matrix $P$ aligns vector $\mu_2-\mu_1$ to the first axis $x_1$). Let $\Delta_{12}=\|\mu_2-\mu_1\|_2$ be the Euclidean distance between $\mu_1$ and $\mu_2$.
Further factorize matrix $P\Sigma P^\top$ using the LDL decomposition (a variant of the Cholesky decomposition) as 
$P\Sigma P^\top = LDL^\top$ where $L$ is an lower triangular matrix with all diagonal entries equal to one (lower unitriangular matrix of unit determinant) and $D$ a diagonal matrix. 
Let $\sigma_{12}=\sqrt{D_{11}}$.
Then we have~\cite{FRMVNReview-2020}:
\begin{equation}\label{eq:FRSigma}
\rho_\Sigma(\mu_1,\mu_2)= \rho_\calN(N(\mu_1,\Sigma),N(\mu_2,\Sigma)) = \rho_\calN(N(0,\sigma),N(\Delta_{12}e_1,\sigma_{12})).
\end{equation}
Note that the right-hand side term is the Fisher-Rao distance between univariate normal distributions of Eq.~\ref{eq:FR1D}.

To find matrix $P$, we proceed as follows: Let $u=\frac{\mu_2-\mu_1}{\|\mu_2-\mu_1\|_2}$ be the normalized vector to align on axis $x_1$.
Let $v=u-e_1$. Consider the  Householder reflection matrix~\cite{householder1958unitary} $M=I-\frac{2vv^\top}{\|v\|_2^2}$, where $vv^\top$ is a outer product matrix.
Since Householder reflection matrices have determinant $-1$, we let $P$ be a copy of $M$ with the last row multiplied by $-1$ 
so that we get $\det(P)=1$. By construction, we have $Pu=\|\mu_2-\mu_1\|_2\, e_1$.
We then use the affine-invariance property of the Fisher-Rao distance as follows:
\begin{eqnarray*}
\rho_\calN(N(\mu_1,\Sigma),N(\mu_2,\Sigma)) &=& \rho_\calN(N(0,\Sigma),N(\mu_2-\mu_1,\Sigma)),\\
&=& \rho_\calN(N(0,P\Sigma P^\top),N(P(\mu_2-\mu_1),P\Sigma P^\top)),\\
&=& \rho_\calN(N(0,P\Sigma P^\top),N(\Delta_{12}\, e_1,P\Sigma P^\top)),\\
&=& \rho_\calN(N(0,L D L^\top),N(\Delta_{12}\, e_1,L D L^\top)),\\
&=& \rho_\calN(N(0,D),N(\Delta_{12}\, e_1, D)).
\end{eqnarray*}
The last row follows from the fact that $L^{-1}e_1=e_1$ since $L^{-1}$ is an upper unitriangular matrix, and 
$L^\top (L^{-1})^\top=(L^{-1}L)^\top=I$.
The right-hand side Fisher-Rao distance is computed from Eq.~\ref{eq:FR1D}.

\section{Embedding multivariate normal distributions in the Siegel upper space}\label{sec:Siegel}
The Siegel upper space is the space of symmetric complex matrices $Z=X+iY=Z^\top$ with imaginary positive-definite matrices $Y\succ 0$~\cite{siegel2014symplectic,nielsen2020siegel} (so-called Riemann matrices~\cite{frauendiener2019efficient}):
\begin{equation}
\SH(d) := \left\{ Z = X + iY \st X\in\Sym(d), Y\in\calP(d)\right\},
\end{equation}
where $\Sym(d)$ is the space of symmetric real $d\times d$ matrices.
$\SH(1)$ corresponds to the Poincar\'e upper plane.
See Figure~\ref{fig:SiegelDistance} for an illustration.

The Siegel infinitesimal square line element is
\begin{equation}
\ds_\SH^2 (Z) = 2 \tr\left(Y^{-1}\dZ\ Y^{-1}\dZbar\right).
\end{equation}
When $X=0$ and $Z=iY$, we have $\dZ=i\dY$, $\dZbar=-i\dY$,  and it follows that
$$
\ds_\SH^2 (iY)=2 \tr\left( (Y^{-1}\dY)^2\right).
$$
That is, four times the square length of the Fisher matrix of centered normal distributions 
$\ds^2_{\calN_0}=\frac{1}{2}\tr \left( (P^{-1}\dP)^2\right)$.

The  Siegel  distance~\cite{siegel2014symplectic} between $Z_1$ and $Z_2\in\SH(d)$ is  
\begin{equation}\label{eq:SiegelDistance}
\rho_\SH(Z_1,Z_2) = \sqrt{\sum_{i=1}^d \log^2\left(\frac{1+\sqrt{r_i}}{1-\sqrt{r_i}}\right)},
\end{equation}
where
\begin{equation}
r_i=\lambda_i\left(R(Z_1,Z_2)\right),
\end{equation}
with $R(Z_1,Z_2)$ denoting  the matrix generalization of the {cross-ratio}
\begin{equation}\label{eq:matrixcr}
R(Z_1,Z_2)  := (Z_1-Z_2)(Z_1-\barZ_2)^{-1} (\barZ_1 -\barZ_2) (\barZ_1 -Z_2)^{-1},
\end{equation}
and $\lambda_i(M)$ denoting the $i$-th largest (real) eigenvalue of (complex) matrix $M$.
(In practice, we numerically have to round off the tiny imaginary parts to get proper real eigenvalues~\cite{nielsen2020siegel}.)
The Siegel upper half space is an homogeneous space where the Lie Group $\mathrm{SU}(d,d)/S(U(d)\times U(d))$  acts transitively on it.

We can embed a multivariate normal distribution $N=(\mu,\Sigma)$ into $\SH(d)$ as follows:
$$
N(\mu,\Sigma)\rightarrow Z(N):=\left(\mu\mu^\top+i\Sigma\right),
$$
and consider the Siegel distance on the embedded normal distributions as another potential metric distance between multivariate normal distributions: 
\begin{equation}
\rho_\SH(N_1,N_2)=\rho_\SH(Z(N_1),Z(N_2)).
\end{equation}
Notice that the real matrix part of the $Z(N)$'s are all of rank one by construction. 

\begin{figure}
\centering
\includegraphics[width=0.55\textwidth]{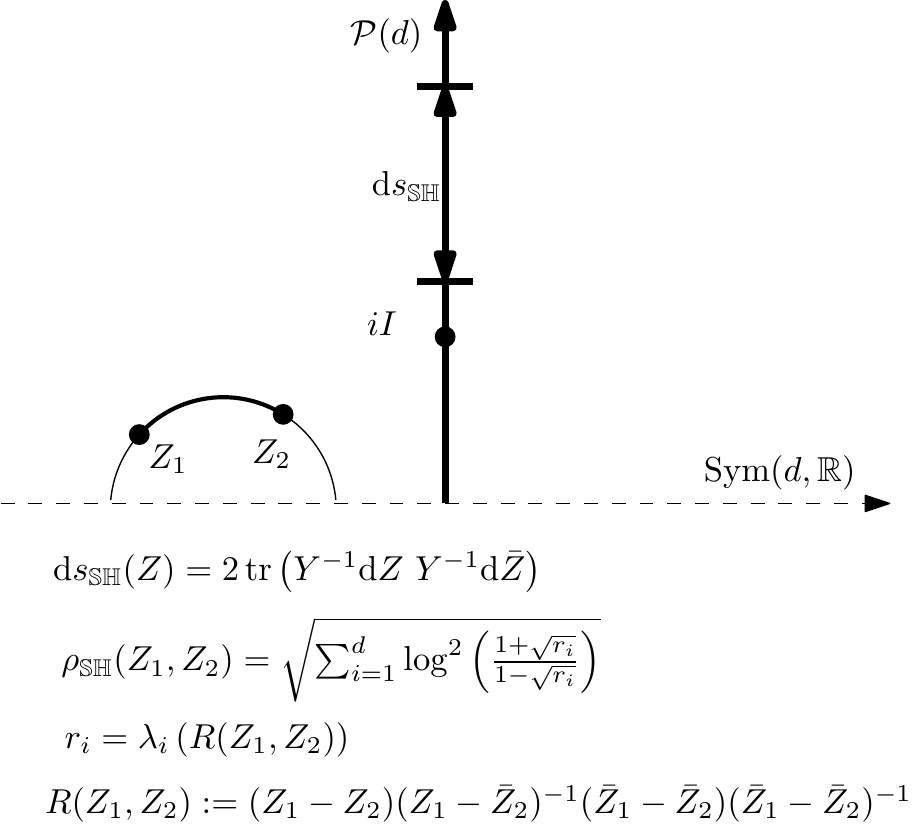}
\caption{Siegel upper space generalizes the Poincar\'e hyperbolic upper plane.
 \label{fig:SiegelDistance}}
\end{figure}

\end{document}